\documentclass[a4paper,11pt,twoside]{report}
\usepackage[left=2.5cm,right=2cm,top=2cm,bottom=2cm]{geometry}

\usepackage{graphicx}
\usepackage{float}
\usepackage{color}
\definecolor{mygreen}{RGB}{28,172,0} 
\definecolor{mylilas}{RGB}{170,55,241}
\usepackage{tikz}
\usetikzlibrary{calc,matrix}
\usepackage{sgame}
\usepackage{fancyhdr}
\usepackage{amsmath}
\usepackage[utf8]{inputenc}
\usepackage{authblk, bm}
\usepackage[english]{babel}
\usepackage[toc,page]{appendix}
\usepackage{listings}
\usepackage{caption}
\usepackage{subcaption}
\usepackage[toc,page]{appendix}

\usepackage{amsmath,amsfonts,amssymb,amsthm,epsfig,epstopdf,titling,url,array}

\usepackage{afterpage}

\newcommand\blankpage{%
    \null
    \thispagestyle{empty}%
    \addtocounter{page}{-1}%
    \newpage}

\theoremstyle{plain}
\newtheorem{thm}{Theorem}[section]
\newtheorem{lem}[thm]{Lemma}

\theoremstyle{definition}
\newtheorem{defn}{Definition}[section]

\newtheorem{exmp}{Example}[section]
\newtheorem{rem}{Remark}[section]

\date{September 2015}

\begin{document}


\begin{titlepage}

\newcommand{\HRule}{\rule{\linewidth}{0.5mm}} 



\center 


\textsc{\Large Imperial College London}\\[0.5cm] 
\textsc{\large Department of Computing}\\[0.5cm] 


\HRule \\[0.4cm]
{ \huge \bfseries {Distributionally Robust Game Theory}}\\ 
\HRule \\[1.5cm]
 

\begin{minipage}{0.4\textwidth}
\begin{flushleft} \large
\emph{Author:}\\
Nicolas Loizou
\end{flushleft}
\end{minipage}
~
\begin{minipage}{0.4\textwidth}
\begin{flushright} \large
\emph{Supervisors:} \\
Panos Parpas\\
Wolfram Wiesemann
\end{flushright}
\end{minipage}\\[3cm]

\begin{figure}[H]
\centering
\includegraphics[width = 0.4\hsize]{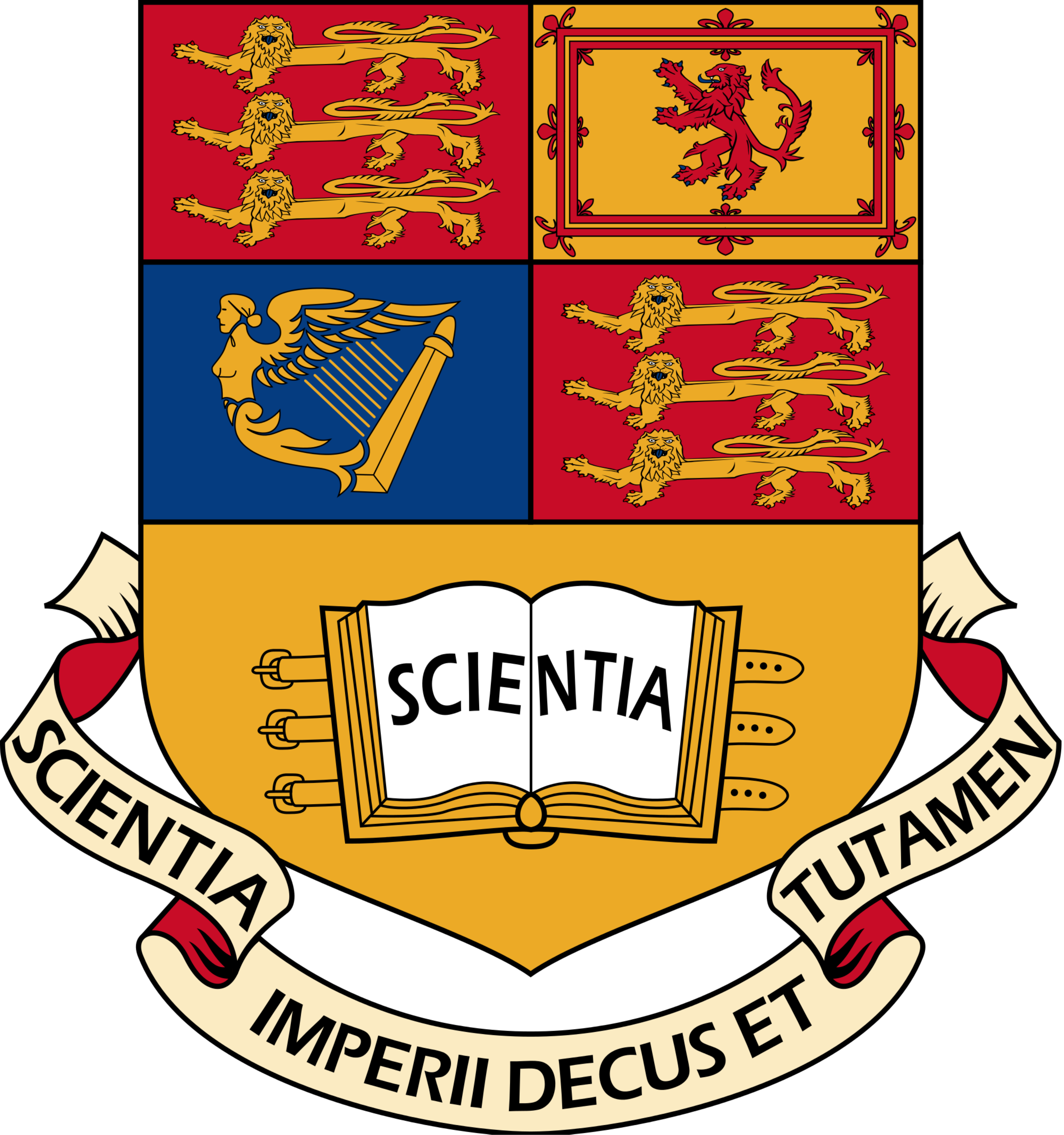}
\label{fig:logo}
\end{figure}
Submitted in partial fulfilment of the requirements for the MSc Degree in Computing(Computational Management Science) of Imperial College London \\ [0.5cm]
\makeatletter
\@date 
\makeatother

\end{titlepage}



  



\afterpage{\blankpage}
\chapter*{Abstract}
The classical, complete-information two-player games assume that the problem data (in particular the payoff matrix) is known exactly by both players. In a now famous result, Nash has shown that any such game has an equilibrium in mixed strategies. This result was later extended to a class of incomplete-information two-player games by Harsanyi, who assumed that the payoff matrix is not known exactly but rather represents a random variable that is governed by a probability distribution known to both players.

In 2006, Bertsimas and Aghassi \cite{aghassi2006robust} proposed a new class of distribution-free two-player games where the payoff matrix is only known to belong to a given uncertainty set. This model relaxes the distributional assumptions of Harsanyi's Bayesian games, and it gives rise to an alternative distribution-free equilibrium concept.

In this thesis we present a new model of incomplete information games without private information in which the players use a distributionally robust optimization approach to cope with the payoff uncertainty. With some specific restrictions, we show that our ``Distributionally Robust Game" constitutes a true generalization of the three aforementioned finite games (Nash games, Bayesian Games and Robust Games). Subsequently, we prove that the set of equilibria of an arbitrary distributionally robust game with specified ambiguity set can be computed as the component-wise projection of the solution set of a multi-linear system of equations and inequalities. Finally, we demonstrate the applicability of our new model of games and highlight its importance.

\afterpage{\blankpage}
\chapter*{Acknowledgements}
First of all, I would like to express my sincere gratitude,to my supervisors Dr. Panos Parpas and Dr. Wolfram Wiesemann for their great guidance and invigorating attitude through out this project. Their continuous advice and creative ideas contributed a great value to this work. \\

Moreover, I would like to thank PhD student, Vladimir Roitch for all the detailed explanations that helped me to deeply understand every aspect of this project.\\

Finally, I would like to thank my parents and my brother George for always believing in me and for their unconditional support throughout the completion of my studies. A special thanks goes also to Katerina for her endless patience and for her continuous encouragement over the past year.

\newpage
\afterpage{\blankpage}
\tableofcontents
\listoffigures
\listoftables

\chapter{Introduction}
Game theory could be considered the field of mathematics and economics that has witnessed the most explosive growth the last century. It is the study of multi-person decision problems, in which the payoff of each person(player) of the game, depends not only on his chosen strategy but also on the strategies of the other players. 

The reason for this explosive growth, is the number of real life applications that could be formulated using game theoretical models. For example, in the area of economics and business, game theory is used for modelling competing behaviours of interacting agents (auctions, bargaining, mergers and acquisitions pricing). Models of game theory could also describe applications in political science in which the players are voters, special interest groups, and politicians. In biology, games have been used for better understanding several phenomena like evolution and animal communication. Furthermore, games play an increasingly important role in logic and computer science (on-line algorithms, algorithmic game theory, algorithmic mechanism design).

The idea of a general theory of games was introduced by John von Neumann and Oskar Morgenstern in their famous book \emph{Theory of Games and Economic Behaviour} (1944,\cite{von2007theory}) in which they presented the first axioms of game theory. Based on this work many type of games have been proposed in the last decades.

The classical, complete-information finite games assume that the problem data (in particular the payoff matrix) is known exactly by all players. In a now famous result(1950, \cite{nash1950equilibrium}, \cite{nash1951non}), Nash has shown that any such game has an equilibrium in mixed strategies. More specifically, in his formulation Nash assumed that all players are rational and that all parameters(including payoff functions) of the game are common knowledge. With these two assumptions, the players can predict the outcome of the game. For this reason each player given the other players strategies is in a position to choose the mixed strategy that gives him the maximum profit. A tuple of these strategies is what we call ``Nash Equilibrium". 

The existence of equilibria in mixed strategies was later extended to a class of incomplete information finite games by Harsanyi (1967-1968), who assumed that the payoff matrix is not known exactly but rather represents a random variable that is governed by a probability distribution known to all players. In particular, Harsanyi assumed that a full prior distributional information for all parameters of the game is available and that all players use this information in order to compute the payoff functions of the game. This computation is made using the Bayes' rule. For this reason these games are called ``Bayesian Games" and their equilibrium ``Bayesian Nash Equilibrium".

In 2006, Bertsimas and Aghassi \cite{aghassi2006robust} proposed a new class of distribution-free finite games where the payoff matrix is only known to belong to a given uncertainty set. This model relaxes the distributional assumptions of Harsanyi's Bayesian games, and it gives rise to an alternative distribution-free equilibrium concept. Furthermore, in this model of games the players use a robust optimization approach to the uncertainty and this assume to be a common knowledge. That is, given the other players strategies each player tries to maximise his worst case expected payoff (worst case is taken with respect to the uncertainty set). The using of the robust optimization approach is the reason for calling these ``Robust Games" and their equilibrium ``Robust Optimization Equilibrium".

In this thesis we present for the first time in the literature a new model of incomplete information games without private information in which the players use a distributionally robust optimization approach to cope with payoff uncertainty. In our model players only have partial information about the probability distribution of the uncertain payoff matrix. This information is expressed through a commonly known ambiguity set of all distributions that are consistent with the known distributional properties. Similar to the robust games framework, players in distributionally robust games adopt a worst case approach. Only now the worst case is computed over all probability distributions within the ambiguity set. More specifically we use a worst case CVaR(Conditional Value at Risk) approach. This allows players to have several risk attitudes which make our model even more coveted since in real life applications players rarely are risk neutral.

We also show that under specific assumptions about the ambiguity set and the values of risk levels, Distributionally Robust Game constitutes a true generalization of the three aforementioned finite games (Nash games, Bayesian Games and Robust Games).
Subsequently, we prove that the set of equilibria of an arbitrary distributionally robust game with specified ambiguity set and without private information can be computed as the component-wise projection of the solution set of a multi-linear system of equations and inequalities. 
Finally, we demonstrate the applicability of Distributionally Robust Games to real life problems and highlight the importance of our model.\\
\section{Structure and Contributions of the Thesis}

The contributions of this thesis are structured as follows: 
\begin{itemize}
\item In Chapter 2 we present the mathematical background that one needs to understand every aspect of this work. We introduce the basic concepts of mathematical programming, risk measures and game theory. 
\item Following that, in Chapter 3 we develop the basic theory of Robust Games and make a small extension of the special class of these games. In the last section of this chapter, three different methods for approximately computing sample equilibria of such games are implemented.
\item The main contributions of our work are developed in Chapter 4. Here, we introduce and analyse the new model of Distributionally Robust Games. We begin with the formulation of the new model and show that any other existing finite game can be expressed as a distributionally robust game with specific ambiguity set and values of the risk levels. In addition we prove the equivalence of the set of equilibria of a distributionally robust game with the component-wise projection of the solution set of multi-linear system of equation and inequalities. For special cases of such games we also show equivalence to complete information finite games (Nash Games) with the same number of players and same action spaces. Finally to concretize the idea of a distributionallly robust game we present two examples.
\item In Chapter 5 we present numerical results for 2-player Distributionally Robust Games. These are the Distributionally robust Free Rider Game and the Distributionally Robust Inspection Game. Our main objective is to study how the players' payments change with small changes of the unknown parameters.
\item This project is concluded in Chapter 6, where we provide a brief summary of what we have done and the possibilities of future research in this area.
\end{itemize}

\newpage
\section{Notation}

The following notations are used in the sequel:\\
\begin{itemize}
\item $\mathbb{R}$ denotes the set of real numbers
\item Boldface upper case letters will denote matrices (e.g., $\bm{P}$)
\item Boldface lower case letters will denote vectors(e.g., $\bm{x}$)
\item $\tilde{\cdot}$ means that the input parameter ($\cdot$), which can be either scalar, vector, or matrix, is subject to uncertainty
\item  $\check{\cdot}$ denotes the nominal counterpart of the uncertain coefficient $\tilde{\cdot}$
\item $[\cdot]^+= \max\{\cdot,0\} $ denotes the maximum value of $\cdot$ and zero. 
\item $vec(\bm{A})$ denotes the column vector obtained by stacking the row vectors of the matrix $\bm{A}$ one on top of the other.
\end{itemize}

\textbf{General Hypothesis:}\\
In order to avoid repeated gender distinctions (his or her) in the following thesis we present all players as male decision makers.

\chapter{Mathematical Background}
In this chapter we provide the necessary mathematical background required for the rest of the
thesis. We first present an overview of unconstrained optimization, continue with optimization under uncertainty, discuss about risk measures and finish with the basic concepts of Game Theory. 

\section{Unconstrained Optimization}
\label{unconOptimization}

In this section we develop the basic theory that we will need for the experimentation part of section  ~\eqref{dsa}. 
More specifically in this thesis we will focus on solving \textbf{ unconstrained optimization} problem of the following form:
 \begin{equation}
 \underset{\bm{x} \in {\mathbb{R}}^n}{\operatorname{min}} f(\bm{x}) 
 \end{equation}
 The most common way to solve this kind of problem is to denote a starting point $\bm{x_0} \in {\mathbb{R}}^n$ and then find a sequence of points $\bm{x_0,x_1,....x_n..}$ where $ f(\bm{x_{k+1}}) \leq f(\bm{x_k})$. The goal in this method is to find a point $\bm{x_k}$ in the sequence that satisfied $\nabla f(\bm{x_k})=0$. This point is the desired solution, the point that minimize the function $f(\bm{x})$.
 In practice find a point with this property is very rare so we assume that our sequence converge to the minimum if one of the following constraints is satisfied:
 \begin{itemize}
  \item $\parallel \nabla f(\bm{x}) \parallel \leq \varepsilon_1 $ 
  \item $\parallel \bm{x_{k+1}-x_k} \parallel \leq \varepsilon_2 $
  \item $\parallel  f(\bm{x_{k+1}})-f(\bm{x_k}) \parallel \leq \varepsilon_3$ 
  \end{itemize}
  where $ \varepsilon_i$ for $i= 1,2,3$ are tolerance parameters.\\
  
  There are two fundamental strategies for moving from the current point $\bm{x_k}$ to the new iterate $\bm{x_{k+1}}$ of the sequence. These are \emph{line search method} and \emph{trust region method}. In this thesis we will develop only different approaches of the first method. An interesting reader can refer to \cite{wright1999numerical}, \cite{chong2013introduction} \cite{bertsekas1999nonlinear} for more details about optimization algorithms.
  
  In the \textbf{line search strategy} the algorithm chooses a direction $d_k$ and searches along this, from the current iterate $x_k$ for a new iterate that give lower value to the function $f(x)$.The distance that algorithm uses to move along $d_k$ it depends on the step size $a_k$ that each algorithm use. 
  
  Consider this, we can understand that each new iteration is described by:
  \begin{equation}
   \bm{x_{k+1}} = \bm{x_k} + a_k \bm{d_k}
\end{equation}  
   
More of the line search algorithms required the direction to be descent, that is to has $$\nabla f(\bm{x_k})^ \tau \bm{d_k} \leq 0$$.

In general the direction $\bm{d_k}$ has the following form:
\begin{equation}
   \bm{d_k} = - {\bm{(A_k)}}^{-1} \nabla f(\bm{x_k})
\end{equation}
  where the matrix $\bm{A_k}$ is symmetric and invertible, and take different values depending on the line search algorithm that we use. In the steepest descent algorithm $\bm{A_k}$ is the identity matrix and in the Newton method is the Hessian matrix of the requested function. Finally in the Quasi Newton method $\bm{A_k}$ is an approximation of the Hessian matrix and is updated in every iteration. \\
  
\subsubsection{Selection of Step size $a_k$ }
There are many ways to calculate the step-size $a_k$ in each iteration. Here we develop the two that we use in this thesis. For more details about the most widely in practice rules for choosing step-size we suggest \cite{bertsekas1999nonlinear}.
\begin{enumerate}
\item Exact Line Search:  
\begin{equation}
\bm{a_k} \in \underset{\bm{a_k}}{\operatorname{argmin}} f(\bm{x_k} + a_k \bm{d_k})
\end{equation}
In this rule $\bm{a_k}$ is chosen to be the value which minimize the $f(\bm{x_{k+1}})$ in each iteration.
\item Armijo Rule: In this rule the step size $s$ is chosen (most of the times $s=1$) and then check if $f(\bm{x_k} + s \bm{d_k}) \leq f(\bm{x_k})$. If this is not the case $s$ is reduced repentantly until the value of $f$ at $\bm{x_{k+1}}$ is less than $f(\bm{x_k})$. The formula for this is:
\begin{equation}
\label{armijo}
f(\bm{x_k}) - f(\bm{x_k} + \beta^m s \bm{d_k}) \geq -\sigma \beta^m s \nabla f(\bm{x_k})' \bm{d_k}
\end{equation}
where $s, \beta, \sigma$ are fixed positive scalars with $\beta \in (0,1)$ and $\sigma \in (0,1)$ and the step size $a_k =  \beta^m s$ where $m$ is the first no-negative integer $(m=0,1,2,3,...)$ which satisfies the equation ~\eqref{armijo}.\\

\end{enumerate}
  
 In the rest of this section the three basic methods of line search strategy that we have already mentioned are developed. 

\subsection{Steepest Descent Method}
One of the basic and more simple in concept algorithm for solving an unconstrained optimization problem is the Steepest Descent.\\
In this algorithm as we understand from the previous analysis the transition to the next point is equal to:
\begin{equation}
 \bm{x_{k+1}} = \bm{x_k} - a_k \nabla f(\bm{x_k})^\tau 
 \end{equation}

where $a_k$ and $d_k = - \nabla f(\bm{x_k})^\tau$ denote the step size and the descent direction respectively. 
For this choice of direction the method is called steepest since $- \nabla f(\bm{x})$is the direction of the greatest decrease.\\

\textbf{Explanation:} We know by definition that the rate of increase f along direction d is:
\begin{equation}
<\nabla f(\bm{x}), \bm{d} > = \nabla f(\bm{x_k})^\tau \bm{d} \quad where \|d\|=1
\end{equation}
By Cauchy-Shwarz inequality,
 \begin{equation}
 <\nabla f(\bm{x}), \bm{d} > \leq \|\nabla f(x)\| \|d\| = \|\nabla f(x)\|
 \end{equation}
 Rate of increase is not possible to be greater than $\|\nabla f(\bm{x})\|$, so if we choose as direction the $d_k=\nabla f(\bm{x})$ we achieve the greatest increase. Therefore when $d_k = - \nabla f(x)$ we obtain the greatest decrease.\\
 
\emph{ \underline{Advantages of Steepest descent:}}
\begin{itemize}
 \item Easy to implement
 \item Only requires first order information
 \end{itemize}
 
\emph{ \underline{Main disadvantage:}}\\
Steepest Descent is one of the \textbf{most slower methods} for solving unconstrained problems. The convergence is very slow.

\subsection{Newton-Raphson Method}
To overcome the slower convergence of steepest descent method a new algorithm which is more effective and faster was created. The Newton-Raphson method used except the first derivative and the second one and indeed perform better of the steepest descent if the starting point $\bm{x_0}$ is close to the minimizer.
The main idea behind the formula that this algorithm used for the transition to the next point is the following:\\

 Compute a quadratic approximation of $f(\bm{x})$.\\
\begin{equation}
f(\bm{x})\approx f(\bm{x_k}) + \nabla f(\bm{x_k})^\tau (\bm{x-x_k}) + \frac{1}{2} (\bm{x-x_k})^\tau \nabla^2 f(\bm{x_k}) (\bm{x-x_k})\triangleq q(\bm{x})
\end{equation}

Applying the First Order Necessary Condition (FONC) to $q(\bm{x})$ yields:\\
\begin{equation}
 0= \nabla q(\bm{x})= \nabla f(\bm{x_k}) + \nabla^2 f(\bm{x_k}) (\bm{x-x_k}) 
\end{equation}

Thus if we assume that the matrix $\nabla^2 f(\bm{x_k})$ is positive definite we obtain the following formula for the transition to the next point of the sequence
\begin{equation}
\label{tran}
 \bm{x_{k+1}} = \bm{x_k} - a_k \nabla^2 f(\bm{x_k})^{-1} \nabla f(\bm{x_k}) 
 \end{equation} \\
By equation ~\eqref{tran} can deduce, that the Newton-Raphson is a descent algorithm with a descent direction given by:\\
\begin{equation}
 \bm{d_k} = - \nabla^2 f(\bm{x_k})^{-1} \nabla f(\bm{x_k}) 
 \end{equation}\\

\emph{\underline{Disadvantages of Newton-Raphson method:}}
\begin{itemize}
\item The method is not guaranteed to converge from any starting point $x_0$(usually only locally convergence can be guaranteed)
\item In regions where the function f is linear the method could break down because the inverse of the Hessian matrix  $\nabla^2 f(x_k)^{-1}$ may fail to exists
\item The computation of the exact value of the Hessian matrix in each iteration could be time consuming.
\end{itemize}

\subsection{Quasi Newton Methods, BFGS}

All these drawbacks of the previous methods were eliminated when a new approach for solving the unconstrained optimization problem were proposed, the Quasi Newton methods. These algorithms are globally convergent and they not use the Hessian matrix but an approximation of this so they are also more computational beneficial compare to the simple newton method. \\
The general formula that all quasi newton methods use for find the next iteration in the sequence is:
\begin{equation}
 \bm{x_{k+1}} = \bm{x_k} - a_k \bm{H_k} \nabla f(\bm{x_k}) 
 \end{equation}
 where $\bm{H_k}$ is the positive definite approximation matrix of the Hessian matrix of the function $f(\bm{x})$.\\
 
 The most popular quasi newton algorithms are the DFP algorithm and the BFGS algorithm. The second one suggested independently in 1970 by Broyden \cite{broyden1970convergence}, Fletcher \cite{fletcher1970new}, Goldfarb \cite{goldfarb1970family} and Shanno \cite{shanno1970conditioning} and this one we used in our implementation in section ~\eqref{dsa}.\\
 
The steps of the BFGS Algorithm are the following:
\begin{enumerate}
\item First we choose the value of the starting point $\bm{x_0}$ and a real symmetric positive definite matrix $\bm{H_0}$ 
\item if $\nabla f(\bm{x_k}) =0 $ stop otherwise $\nabla f(\bm{x_k}) = - \bm{H_k} \nabla f(\bm{x_k})$
\item Calulate the \begin{equation}
a_k = \underset{a \geq 0}{\operatorname{argmax}} f(\bm{x_k} +a \bm{d_k})
\end{equation}
\begin{equation}
\bm{x_{k+1}} = \bm{x_k} + a_k \bm{d_k}
\end{equation}
\item Estimate the \begin{equation}
\Delta \bm{x_k}= \bm{x_{k+1}} - \bm{x_k} 
\end{equation}
\begin{equation}
\Delta (\nabla f(\bm{x_k}))= \nabla f(\bm{x_{k+1}}) - \nabla f(\bm{x_k})
\end{equation}
\item Calulate the \begin{equation}
\begin{split}
\bm{H_{k+1}}= \bm{H_k} + ( 1 + \frac{\Delta (\nabla f(\bm{x_k}))^\tau \bm{H_k} \Delta (\nabla f(\bm{x_k}))}{\Delta (\nabla f(\bm{x_k}))^\tau \Delta \bm{x_k} } ) \frac{\Delta \bm{x_k} \Delta {\bm{x_k}}^\tau}{\Delta {\bm{x_k}}^\tau \Delta (\nabla f(\bm{x_k})) }\\ - \frac{\bm{H_k} \Delta (\nabla f(\bm{x_k}))\Delta {\bm{x_k}}^\tau + (\bm{H_k} \Delta (\nabla f(\bm{x_k}))\Delta {\bm{x_k}}^\tau)^\tau }{\Delta (\nabla f(\bm{x_k}))^\tau \Delta \bm{x_k}}
\end{split}
\end{equation}
\item set k=k+1 and go back to step 2.
\end{enumerate}

\section{Optimization Under Uncertainty}
\label{distniel}
Organisations frequently need to take decisions based on uncertain or incomplete information (e.g., about future customer demands, raw material prices or exchange rates). Failure to take this uncertain parameters into consideration may lead to suboptimal unwanted decisions. In this section we present two of the most widely used approaches which address practical optimization problems affected by uncertainty. These are the Robust Optimization approach and the Distributionally Robust Optimization approach.\\

\subsection{ Robust Optimization}
Not surprisingly, decision-making under uncertainty has a long and distinguished history in Operations Research. Until now in many  research papers decision problems under uncertainty are solved by modelling uncertain data as random variables and then analysing and  discretising the outcomes of these random variables. This way, however, has several disadvantages. Firstly,  the probability distribution governing the random variables is typically unknown and has to be estimated from historical observations. Moreover the aforementioned discretisation implies that the computation times grow exponentially with problem size. The above-mentioned two drawbacks  have been a major impediment to the applicability of quantitative approaches to decision-making under uncertainty.

Nowadays,  both shortcomings have been addressed by a novel methodology termed \textbf{robust optimisation}. With different techniques robust optimisation avoids the aforementioned curse of dimensionality. Hence, this approach seems to be ideally suited for practical decision problems that are large-scale and subject to significant uncertainty.\\

To grasp the main concept of robust optimization approach for solving problems with parameters under uncertainty, let us consider the following mathematical optimization problem:

\begin{equation}
\label{nominal}
\begin{array}{l@{\quad}l@{\qquad}l}
\displaystyle \underset{{\bm{x}} \in \mathbb{R}^n }{\operatorname{minimize}} & f_0(\bm{x},\bm{u}) \\
\displaystyle \text{subject to} & \displaystyle f_i(\bm{x,u})  \leq 0 \quad i \in \{1,2,..m\} \\
& \displaystyle \bm{x} \in \mathcal{X}.
\end{array}
\end{equation}\\

In this problem $\bm{x} \in {\mathbb{R}}^n$ denotes the vector of decision variables. It is the vector that we must compute, with respect the m constraints $f_i(\bm{x,u})$, in order to find the minimum value of the objective function $f_0(\bm{x})$. The functions  $f_0$ and $f_i$ are $f: \mathbb{R}_n\longrightarrow \mathbb{R} $ and vector $\bm{u} \in {\mathbb{R}}^t$ take specific known value (it is fixed). Note that both objective function $f_0$ and the constraints function $f_i$ depend on the fixed vector $\bm{u}$.\\

Now, lets assume that $\bm{\tilde{u}} \in {\mathbb{R}}^t$ is a random vector. With this assumption we  obtain the following robust counterpart of the previous nominal problem.

\begin{equation}
\label{robustcounterpart}
\begin{array}{l@{\quad}l@{\qquad}l}
\displaystyle \underset{{\bm{x}} \in \mathbb{R}^n }{\operatorname{minimize}} & \underset{ u \in U }{\operatorname{max}} \,\, f_0(\bm{x},\bm{u}) \\
\displaystyle \text{subject to} & \displaystyle f_i(\bm{x,u}) \leq 0, \quad \forall \bm{u} \in U \subseteq {\mathbb{R}}^t \\
& \displaystyle \bm{x} \in \mathcal{X}.
\end{array}
\end{equation}\\

Our goal in the robust counterpart, is to compute the vector of decision variables $\bm{x}$ that minimise the objective function $\underset{ u \in U }{\operatorname{max}} f_0(\bm{x,u})$ which represents the worst case scenario. In this case, we must also consider all possible constraints that the disturbance of $\bm{\tilde{u}}$ could create.Under this approach we must be willing to accept a suboptimal solution to our problem in order to make sure that  the solution remains feasible for all nominal problems. \\

At this point, it is worth noting the following two remarks about the robust counterpart:

\begin{rem}
We could assume without loss of generality that the objective function is not affected by the uncertainty of the parameters.That happen because we can always reformulate the optimisation problem in such a way that makes the objective function uncertainty-free. (See Berstimas-Caramanis \cite{bertsimascaramanis2011theory}).
\end{rem}
  
\begin{rem}
It is obvious that the case of constraints without uncertainty is subsumed in the robust counterparts by assuming that the uncertainty set U is singleton.
\end{rem}

\subsubsection{Tractability of Robust Counterpart}

With the introduction of the potential huge number of constraints in the robust counterpart of a nominal problem, one can think that this kind of modelling is intractable. In general, this is true, but in this thesis and in the most papers of the literature about robust optimization we deal with tractable problems\footnote{In computational complexity theory, tractable problem is the problem that can be solved in polynomial time}. We achieved tractability by using specifying classes of $f_i$ and specific uncertainty sets $U_i$.

The main categories of tractable robust optimization problems are robust linear optimization, robust quadratic optimization, robust semi-definite optimization and robust discrete optimization.
In this thesis we focus on robust linear optimization problems and more specifically on the case that the uncertainty sets are bounded polyhedral.\\

When we refer to Robust Linear Optimization we mean the problem which the robust counterpart of linear optimization problem is:

\begin{equation}
\label{robustlinearoptimi}
\begin{array}{l@{\quad}l@{\qquad}l}
\displaystyle \underset{{\bm{x}} \in \mathbb{R}^n }{\operatorname{minimize}} & {\bm{c}}^{\tau}\bm{x} \\
\displaystyle \text{subject to} & \displaystyle \bm{Ax} \leq  \bm{b} \quad \forall  \bm{a_i} \in U \quad i = 1, \ldots, m.
\end{array}
\end{equation}\\

Where $\bm{A} \in \mathbb{R}^{m \times n}$ is the uncertainty coefficient matrix and $\bm{a_i}$ are the $i$ row of the matrix A.\\

Furthermore, when we specify even more to robust linear optimization with bounded polyhedral uncertainty set we assume that the uncertainty set U (following the notation from \cite{bertsimas2004robust}) has the following form:\\

\begin{equation}
\label{polyset}
 U= \{\bm{A} | \bm{F} vec(\bm{A}) \leq \bm{d}\}
\end{equation}

where $\bm{F} \in \mathbb{R}^{1 \times mn}$ and $vec(\bm{A}) \in \mathbb{R}^{mn \times 1}$ is the column vector obtained by stacking the row vectors of the matrix $\bm{A}$ one on top of the other.\\

From the first approach of Soyster \cite{soyster1973technical} in 1973 the robust optimization has been drawing a lot of attention from researchers mainly because of its tractability. For developing a more sophisticated understanding of the robust optimisation approach the reader may refer to Ben-Tal and Nemirovski\cite{bennamiro1998robust}, \cite{bennamiro1999robust} , El Ghaoui \cite{elGhaoui1997robust},\cite{elGhaoui1998robust},Bertsimas and Sim \cite{bertsimas2003robust}, \cite{bertsimas2004price}, Bertsimas Pachamonova and Sim \cite{bertsimas2004robust}. Finally for a thorough discussion, we suggest the \cite{ben2009robust} of Ben-Tal, El Ghaoui and Nemirovski, which gives an excellent and stimulating account of the classical and advanced results in this field.

\subsection{Distrinutionally Robust Optimization}

In robust optimization approach we mentioned that the decision maker is not aware or choose to ignore any probability distribution of the random vector $\bm{\tilde{u}}$. He make a decision based on the worst case scenario, of all possible values of the uncertainty vector $\bm{\tilde{u}}$, that belong in an uncertainty set. \\
The distributionally robust optimization approach is closely related to robust optimization framework but now the decision maker instead of uncertainty set has partial information about the probability distribution Q of the random vector. With this approach one does not know the true probability distribution of the random variable of interest, but instead, attempts to create bounds within which the true distribution should lie in, based on available data. Then, to guard against the uncertainty of where the true distribution lies within these bounds, one makes decisions that are optimal with respect to the worst distribution in this set. \\
Let $\mathcal{F}$ denote the ambiguity set of all distributions Q that satisfy the known distributional properties.
Then the distributionally robust counterpart of the nominal problem ~\eqref{nominal} is the following:

\begin{equation}
\label{distrirobustcounterpart}
\begin{array}{l@{\quad}l@{\qquad}l}
\displaystyle \underset{{\bm{x}} \in \mathbb{R}^n }{\operatorname{minimize}} & \underset{ Q \in \mathcal{F} }{\operatorname{sup}} \,\,\mathbb{E}_Q [f_0(\bm{x},\bm{\tilde{u}})] \\
\displaystyle \text{subject to} & \displaystyle Q[f_i(\bm{x,\tilde{u}}) \leq 0] \geq \varepsilon \quad \forall Q \in \mathcal{F}\\
& \displaystyle \bm{x} \in \mathcal{X}.
\end{array}
\end{equation}\\

The objective function $\underset{ Q \in \mathcal{F} }{\operatorname{sup}} \,\,\mathbb{E}_Q [f_0(\bm{x},\bm{\tilde{u}})]$ of the distributionally robust counterpart represents the worst case scenario of the ambiguity set. In this problem the decision maker try to minimise the expected value of function $[f_0(\bm{x},\bm{\tilde{u}})]$ of the worst case distribution. Furthermore, in the distributionally robust problem the uncertain constraints must be satisfied with high probability. The $\varepsilon$ in the formulation denotes the risk factor of this probability and it is obvious that as the value of $\varepsilon$ decrease then the constraint has to be satisfied with lower probability. ( $\varepsilon \in (0,1)$)\\ 

The distributionally robust optimization, through the freedom of choosing the right dis-utility functions $f_i$ captures the risk attitude of the decision maker. In addition, through the examination of all distributions of the ambiguity set in order to find the worst distribution, distributionally robust approach expresses an aversion towards ambiguity( for details see \cite{wiesemann2014distributionally}). These are the two main advantages that make the distributionallly robust optimization one of the main tools that the operational research analysts use to overcome the noisy and incomplete data. \\ 

For a thorough discussion on distributionally robust optimization the reader may refer to \cite{scarf1958min}, \cite{gilboa1989maxmin}, \cite{bertsimas2010models},\cite{goh2010distributionally} \cite{delage2010distributionally} and \cite{wiesemann2014distributionally} which are considered important milestones of the area.

\section{Risk Measures}
In the area of financial mathematics and financial risk management, risk measures can be divided into two main categories: the moment-based and the quantile based risk measures (see \cite{natarajan2009constructing}). One can understand this separation in these two categories by thinking that the first one is related to classical utility theory and that the second has created as a need after the continuously development theory of stochastic dominance theory. In this thesis we will not use moment-based risk measures. For familiarity with this kind of measures, and more specifically with the Mean-Variance which is the most popular measure of this category we suggest the reading of the pioneering work of Markowitz \cite{markowitz1952portfolio}.\\

In this section we describe the two most widely used quantile risk measures: the Value at Risk(VaR) and the Conditional Value at Risk(CVaR).

\subsection{Value At Risk}
The most widely used quantile based risk measure is the Value at Risk \cite{jorion1997value}. \\
$\text{VaR}_\varepsilon$ is defined as the $(1 - \varepsilon)$-quantile of the loss distribution L (typically $ \varepsilon $ is equal to $1\%$ or $5\%$). That is, $\text{VaR}_\varepsilon$ is the smallest number $x$ for which the probability that loss distribution $L$ exceeds $x$ is not larger than $\varepsilon$:\\

\begin{equation}
\text{VaR}_\varepsilon(L) = \inf\{\bm{x} \in \mathbb{R} : \mathbb{P}(L \geq \bm{x}) \leq \varepsilon\}
\end{equation}

To get more intuitively understanding of what Value at Risk represents and how we can estimate its value lets consider the following example:

\begin{exmp}(from \cite{hull2012risk})\\
Suppose that the outcomes of a project are between a gain of $\$50$ million and a loss of $\$50$ million and all are equally considered. For this reason the distribution of loss is a uniform distribution from $-\$50$ million to $+\$50$ million. We want to find the value for which there is $1\%$ chance to loss greater than $\$49$ million. This is equivalent to the $VaR_1$ from the above definition. Therefore is equivalent to $\$49$ million.
\end{exmp}

\subsection{Conditional Value at Risk}
\label{CVARDEFIni}
Another quantile-based measure of risk is the Conditional Value at Risk which firstly introduce by Rockafellar and Uryasev (see \cite{rockafellar2000optimization} and \cite{rockafellar2002conditional}). It has been gaining popularity due to its attractive  computational properties and to its very desirable property that is a coherent risk measure (see \cite{artzner2002coherent}).\\
$\text{CVaR}_\varepsilon$  of a loss distribution $L$ is the expected value of all losses that exceed $(1 - \varepsilon)$-quantile of the distribution. As we mention before this  $(1 - \varepsilon)$-quantile is the $\text{VaR}_\varepsilon$ of the loss distribution.\\

Conditional Value at Risk can be formalized as:\\
\begin{equation} 
Q\text{-CVaR}_\varepsilon (L) = \min_{\zeta \in \mathbb{R}} \; \zeta + \frac{1}{\varepsilon} \mathbb{E}_Q [ L - \zeta]^+ 
\end{equation}
where $[x]^+ = \max\{x,0\}$.\\

\textbf{More Important Properties of CVar:}
\begin{itemize}
\item $\text{CVaR}_\varepsilon$ of a loss distribution increases as $\varepsilon$ decreases and vice versa.
\item For $\varepsilon = 1$ the Conditional Value at Risk is equivalent to the expected value of the loss distribution.
\item It can be shown that $\text{CVaR}_\varepsilon(L)$ is always greater than the $\text{VaR}_\varepsilon(L)$. For this reason CVaR is often used as a conservative approximation of VaR(see \cite{rockafellar2002conditional}) 
\item $\text{CVaR}_\varepsilon$ is a coherent risk measure \cite{artzner2002coherent}.
\end{itemize} 

The new model of Games that we propose in this thesis is very relevant with the concept of the CVaR. Hence, the CVaR risk measure will be considerably used , throughout this thesis.\\

\section{Game Theory}
\label{gametheory}

Game theory could be considered the field of mathematics and economics that has witnessed the most explosive growth the last century. It is the study of multi-person decision problems, in which the payoff of each person(player) of the game, depends not only on his chosen strategy but also on the strategies of the other players. 

The reason of this explosive growth, is the number of real life applications that could be formulated using game theoretical models. For example, in the area of economics and business, game theory is used for modelling competing behaviours of interacting agents (auctions, bargaining, mergers and acquisitions pricing). Models of game theory could also describe applications in political science in which the players are voters, special interest groups, and politicians. In biology, games have been used for better understanding of several phenomena like evolution and animal communication. Furthermore, games play an increasingly important role in logic and computer science (on-line algorithms, algorithmic game theory, algorithmic mechanism design).

Depending on the problem that we have to solve, several types of games have been proposed. Each of this type can described by specific characteristics. For example if the players can communicate and they could make agreements about their chosen strategies the game called \emph{cooperative}. Contrary, in \emph{non-cooperative games}, this communication is not possible. In particular, in this thesis we focus on non-cooperative, simultaneous-move, one shot, finite games with complete information or incomplete information. ``Simultaneous-move" means that the strategies that adopted by the players are chosen simultaneously. That is, the players are not aware of their opponents' strategies before choose their own strategies. ``One shot" refers to the fact that the game is played only one time. Finally, ``Finite" denotes that the number of players and the number of their  possible actions are finite. 

In this section we present the basic theory of finite game with complete(Nash Games) and finite games with incomplete information(Bayesian Games). For more information about these two categories of games and their many applications the reader may refer to \cite{gibbons1992primer}, \cite{fudenberg1991game} and \cite{osborne1994course}.

\subsection{Finite Games with Complete Information}
``Complete Information" refers to the fact that all parameters of the game including individual players' payoff functions are common knowledge. Thus, all players know exactly the payoffs that the other players will receive for any combinations of strategies.\\
More specifically, in this section we focus in how we can present this type of games (Normal Form Representation) and how we can solve the resulting game-theoretic problem (Dominated Strategies, Nash Equilibrium).\\
 
\subsubsection{Games in Normal Form}

Let us begin with a standard representation of a game, which is known as a normal form representation. This is the most conveniently handled way to represent the complete information games and through this we could later identify the dominated strategies and Nash equilibrium.
To define the normal-form representation of a game, we need to specify the number of players, the available strategies of each player, and the players’ payoffs.\\
  
  \begin{defn}
  \emph{(Normal Form Representation):}\\
  \label{normal}
In static games of complete information, a normal-form representation is a specification of players' strategy spaces and payoff functions. In particular, a game G is considered to be in a normal form representation if it consists of $N={1,2,....,n}$ players, one set $S_i$ for each player $i$ which denote his strategy space ($\bm{s_i}$ denote an arbitrary member of this set) and from one function $u_i(\bm{s_1,s_2,...,s_n})$ which represent the payoff of player $i$ if the players play the strategies $\bm{s_1,s_2,...,s_n} $. We denote this game by $G= \{S_1,S_2,....S_n;u_1,u_2,....u_n\}$.\\
\\
\end{defn}

Normal form games are often represented and are more understandable by their matrix form (a table). 
This can be clarified by the following example.
\\

\begin{exmp}
  \emph{(Normal Form Representation):}
  \label{normal} 

\begin{table}[h]
\centering
\caption{A normal form Game}
\label{bla}
\begin{tikzpicture}[element/.style={minimum width=1.75cm,minimum height=0.85cm}]
\matrix (m) [matrix of nodes,nodes={element},column sep=-\pgflinewidth, row sep=-\pgflinewidth,]{
         & Left  & Right  \\
Left & |[draw]|(5,3) & |[draw]|(4,6) \\
Right & |[draw]|(2,2) & |[draw]|(1,2) \\
};

\node[above=0.25cm] at ($(m-1-2)!0.5!(m-1-3)$){\textbf{Player 2}};
\node[rotate=90] at ($(m-2-1)!0.5!(m-3-1)+(-1.25,0)$){\textbf{Player 1}};
\end{tikzpicture}
\end{table}

This matrix is a normal-form representation of a game in which players move simultaneously (or at least do not observe the other player's move before making their own) and receive the payoffs as specified for the combinations of actions played. For example, if player 1 plays Left and player 2 plays Right, player 1 receives 4 and player 2 receives 6. In each cell, the first number represents the payoff to the row player (in this case player 1), and the second number represents the payoff to the column player (in this case player 2).\\
\end{exmp}

Although, normal form representation of a game may look trivial it is very important for the intuitively understanding of the game and very helpful in the procedure of solving it.
As we mentioned before, there are two ways of solving games with complete information 
\begin{itemize}
\item Iterated elimination of strictly dominated strategies and 
\item Nash Equilibrium
\end{itemize}  
In this thesis we focus on solutions with Nash Equilibrium, hence in this subsection we present only the definition of the Dominated strategies and the disadvantages of the Iterated elimination method. For a comprehensive review (theory and examples) of this method we suggest \cite{gibbons1992primer}. \\

 \begin{defn}
  \emph{(Strictly Dominated Strategy):}\\
  \label{Dominated} 
  In the normal-form game $G = \{S_1,S_2,....S_n;u_1,u_2,....u_n\}$, let
$\bm{s_i'}$ and $\bm{s_i"}$ be feasible strategies for player $i$ (i.e.$\bm{s_i'}$ and $\bm{s_i"}$  are members of $S_i$). Strategy $\bm{s_i'}$ is \emph{ strictly dominated} by strategy $\bm{s_i"}$ if for each feasible combination of the other players' strategies, i's payoff from playing $\bm{s_i'}$ is strictly less than i's payoff from playing $\bm{s_i"}$ :
$$u_i(\bm{s_1,s_2,..s_{i-1},s_i',s_{i+1}..,s_n}) < u_i(\bm{s_1,s_2,..s_{i-1},s_i",s_{i+1}..,s_n}) $$
for each $(\bm{s_1,s_2,..s_{i-1},s_{i+1}..,s_n})$ that can be constructed from the other
players' strategy spaces $ S_1,S_2,....S_n$.\\
 \end{defn}
 
From the above definition ~\eqref{Dominated} and assuming that all players of the game are rational we could perceive that no one of the players will play any strictly dominated strategy.
Using this assumption we can solve a game by a process called \textit{ iterated elimination of strictly dominated strategies}. With this method we simply cross out all the strictly dominated strategies (in the payoof matrix) of each player consecutively. (For example see \cite{gibbons1992primer})\\

\textbf{Drawbacks of the iterated elimination of strictly dominated strategies:}
\begin{enumerate}
 \item  To use this method we can not only assume that players are rational but we also have to make the assumption that this fact is a common knowledge. That is, for the elimination of one strategy each player has to assume that all the other players have rational behaviour and that all of them knows that he is also a rational player. 
  \item The process often produces a very inaccurate prediction about the play of the game. The possibility that non of the nominated strategies can be eliminated it is large so we can have a game which this method does not give us a specific solution.
\end{enumerate}

For these limitations, the concept of Nash Equilibrium is considered a stronger solution concept than the strictly dominated strategies. \\

\begin{defn}
  \emph{(Best Response, (from \cite{aghassi2006robust}) )}\\
  \label{best}
 A player’s strategy is called a \textbf {best response} to the other players’ strategies if, given the latter, he has no incentive to unilaterally deviate from his aforementioned strategy.\\
\end{defn}

\begin{defn}
  \emph{(Nash Equilibrium in pure strategies)}\\
  \label{Nash}
In the n-player normal-form game $G = \{S_1,S_2,....S_n;u_1,u_2,....u_n\}$, the strategies $ \bm{s_1^*,s_2^*,....s_n^*}$ are a Nash equilibrium if, for each player
$i$, $\bm{s_i*}$ is player i's \textbf{best response} to the strategies specified for the $n-1$ other players, $\bm{s_1^*,...s_{i-1}^*,s_{i+1}^*,....,s_n^*}$:

$$u_i(\bm{s_1^*,...s_{i-1}^*,s_i^*,s_{i+1}^*,....,s_n^*}) < u_i(\bm{s_1^*,...s_{i-1}^*,s_i,s_{i+1}^*,....,s_n^*}) $$

for every feasible strategy $\bm{s_i}$in $S_i$; that is, $\bm{s_i*}$ solves
max $u_i(\bm{s_1^*,...s_{i-1}^*,s_i^*,s_{i+1}^*,....,s_n^*})$. \\
\end{defn}

To be more concrete and fully understand the concept of Pure Nash Equilibrium, we present two of the most known examples of the Game Theory literature(The battle of the sexes and the Matching Pennies).

\begin{exmp}
  \emph{(The battle of the sexes):}
  \label{exsexes}
 Game with two Nash Equilibrium in Pure Strategies 

\begin{table}[H]
\centering
\caption{The battle of the Sexes}
\label{battle of the sexes}
\begin{tikzpicture}[element/.style={minimum width=1.75cm,minimum height=0.85cm}]
\matrix (m) [matrix of nodes,nodes={element},column sep=-\pgflinewidth, row sep=-\pgflinewidth,]{
         & Ballet  & Football  \\
Ballet & |[draw]|(2,1) & |[draw]|(0,0) \\
Football & |[draw]|(0,0) & |[draw]|(1,2) \\
};

\node[above=0.25cm] at ($(m-1-2)!0.5!(m-1-3)$){\textbf{Player 2}};
\node[rotate=90] at ($(m-2-1)!0.5!(m-3-1)+(-1.25,0)$){\textbf{Player 1}};
\end{tikzpicture}
\end{table}
\end{exmp}

In the battle of the sexes the scenario of our game is the following. There are two players which are couple. Lets say that player 1 is the man and player two is the woman and they are planning where they will meet for their first date. Both players want to have a date, but man prefers a night at the ballet show and woman prefers the football match (weird preferences). If they choose to meet at the ballet show then man and woman have profits equal to 2 and 1 respectively. If if they choose football match then we have exactly the opposites payoffs. If they do not make an agreement about the place they stay home and get both 0 as a profit.(see Table ~\ref{battle of the sexes})
Using definition ~\eqref{Nash} we can find that in this game we have two Pure Nash Equilibrium, the one that both players choose to attend in ballet and the other one that they will choose to meet to the football match. For example, the combination of strategies (Ballet,Ballet) is a Nash equilibrium because both players respond to their component using their best response(definition ~\eqref{best}). If player 1 knows that player 2 will choose Ballet then the Ballet is the best response to his component strategy because the football give him payoff equal to 0 instead of 2.\\

Although the solution concept of the previous definition of Pure Nash Equilibrium produces tight predictions in a very broad class of games, there are simple games in which we could not find any Nash Equilibrium. For this reason we must extend this definition.\\

\begin{exmp}
  \emph{(Matching Pennies):}
  \label{Match} Game without Nash Equilibrium in Pure Strategies \\
  The following game is played by two players and each one of them has two possible actions, Heads or Tails. Each player has a penny and he through it. If the two coins match(both Heads or both Tails) the player 2 wins one penny and the player 1 loses. If the result is Heads in the one coin and Tails in the other, then the opposite happen. Again by the definition of best response and Nash Equilibrium we realize that this game has no Nash equilibrium in pure strategies.For example the combination of strategies (Heads,Heads) is not Nash Equilibrium because if the player 1 know that player 2 will play Heads then he will prefer to choose Tails in order to improve his payoff from -1 to 1.

\begin{table}[H]
\centering
\caption{Matching Pennies}
\label{bla}
\begin{tikzpicture}[element/.style={minimum width=1.75cm,minimum height=0.85cm}]
\matrix (m) [matrix of nodes,nodes={element},column sep=-\pgflinewidth, row sep=-\pgflinewidth,]{
         & Heads  & Tails  \\
Heads & |[draw]|(-1,1) & |[draw]|(1,-1) \\
Tails & |[draw]|(1,-1) & |[draw]|(-1,1) \\
};

\node[above=0.25cm] at ($(m-1-2)!0.5!(m-1-3)$){\textbf{Player 2}};
\node[rotate=90] at ($(m-2-1)!0.5!(m-3-1)+(-1.25,0)$){\textbf{Player 1}};
\end{tikzpicture}
\end{table}

\end{exmp}

\subsubsection{Mixed Strategies}
As we have showed in the previous example ~\eqref{Match} a game in normal-form representation does not always have a Nash equilibrium in which each player deterministically chooses one of his strategies(pure strategies).In the game of Matching Pennies none of the possible combinations of "pure strategies" is Nash Equilibrium. This happen because for every combination, one of the two players prefers to switch his aforementioned strategy.\\
Like the Gibbons quotes in \cite{gibbons1992primer}:"In any game which each player would like to outguess the other there is no Nash equilibrium (in Pure strategies) because the solution to such game necessarily involves uncertainty about what the players will do". \\
For this reason we have to introduce the notion of \textbf{mixed strategies}. \\
Up to this point, we assume that all players choose a ``pure strategy" which means that they just pick one of their possible actions. With the mixed strategy concept each player $i$ has the right to select a strategy which could be also a probability distribution over his possible pure strategies $\bm{s_i}$.\\

\begin{defn}
  \emph{(Mixed Strategy):}
  \label{Mixed}
In the normal-form game $G = \{S_1,S_2,....S_n;u_1,u_2,....u_n\}$, a strategy of a player $i$ is called \textbf{mixed strategy} if it is a probability distribution $\bm{p_i} = (p_{i1},...,p_{ik})$ over the strategy space $S_i= \{\bm{s_{i1},....,s_{ik}}\}$. Each $p_{ij}$ denotes the probability of player $i$, choose the pure strategy $\bm{s_j}$ of the strategy space $S_i$ where $j \in \{1,2,.....k\}$ , $0 < p_{ij} < 1$ and $ p_{i1}+...+ p_{ik} = 1$.\\
 \end{defn}
 
 In all previous examples the players allowed to use only pure strategies. For this reason, for each tuple of strategies, the amount of the payoff function for each player it was simply be the corresponding value of the payoff matrix. With the introduction of mixed strategies we understand that this is not always the case. The players now are not restricted to choose only one specific action from their strategy space but instead they have the right to select a probability distribution over their possible actions. Thus, we will need to make some calculations in order to compute the payoff of each player instead of simply take it from the payoff matrix.
 
 \subsubsection{Calculations of the payoffs if the players use mixed strategies}
Let $S_i=\{\bm{s_{i1},s_{i2},....,s_{i_{Hi}}}\}$ denotes the strategy space of player $i$ and $Hi$ denotes the number of possible actions of player $i$(the possible pure strategies). The arbitrary pure strategy of player $i$ is $\bm{s_{ik}}$ where $k\in\{1,2,....,Hi\}$.
 If player $i \in \{1,2,....n\}$ believes that all other players $j \in \{1,2,..i-1,i+1,...n\}$ will play the strategies $(\bm{s_{j1},s_{j2},....,s_{j_{Hj}}})$with probabilities $\bm{p_j} = (p_{j1},...,p_{j_{Hj}})$ then players i's expected payoff from playing the pure strategy $\bm{s_{ij}}$ is:
 \begin{equation}
  k = \sum\limits_{{h_1}=1}^{H_1}\sum\limits_{{h_2}=1}^{H_2}... \sum\limits_{{h_{i-1}}=1}^{H_{i-1}}\sum\limits_{{h_{i+1}}=1}^{H_{i+1}}...\sum\limits_{{h_N}=1}^{H_N}  p_{1h_1}p_{2h2}.. p_{(i-1)h_{i-1}}p_{(i+1)h_{i+1}}..p_{Nh_N}    u_i(\bm{s_{1h_1},..,s_{ij},.. s_{Nh_N}})
 \end{equation}
 and players i'expected payoff from playing the mixed strategy $\bm{p_i} = (p_{j1},...,p_{j_{Hj}})$ is:
 \begin{equation}
\begin{split}
&\quad u_i(p_1,p_2,....p_n) = \sum\limits_{h_i=1}^{H_i} p_{ih_i} [ k ] \\ 
&\quad = \sum\limits_{{h_1}=1}^{H_1}\sum\limits_{{h_2}=1}^{H_2}... \sum\limits_{{h_{i-1}}=1}^{H_{i-1}}\sum\limits_{{h_{i+1}}=1}^{H_{i+1}}...\sum\limits_{{h_N}=1}^{H_N}  p_{1h_1}p_{2h2}.. p_{(i-1)h_{i-1}} p_{ih_i}p_{(i+1)h_{i+1}}..p_{Nh_N}    u_i(\bm{s_{1h_1},..,s_{ij},.. s_{Nh_N}}) 
 \end{split}
  \end{equation}
Now that we have shown how we can find the payoff of each player when mixed strategies are allowed we can extend the Nash equilibrium concept to include mixed strategies. The extension is simple. We simply use the same idea as before. We require that \textbf{the mixed strategy of each player must be  the best response to the given mixed strategies of the other players}.\\
This extension of Nash Equilibrium includes the earlier one from definition ~\eqref{Nash} since we can express every pure strategy as a mixed strategy. It is the strategy which has 1 only at one coordinate of the probability distribution(the one corresponds to the playing action) and zeros to all the other coordinates. \\

In $n$ player game a tuple of mixed strategies $(\bm{p_1^*,p_2^*,...p_i^*...,p_n^*})$ is a Nash equilibrium if $ p_i^*$ satisfy $$u_i(\bm{p_1^*,p_2^*,...p_i^*..,p_n^*}) \geq u_i(\bm{p_1^*,p_2^*,..,p_i,...,p_n^*})\quad  \forall i \in \{1,2,....n\} $$ for every probability distribution $\bm{p_i}$ over the strategy space $S_i$. This means that the mixed strategy of each player is best response to other players' strategies. A player, given the other players strategies, has no reason to unilaterally deviate from his aforementioned strategy.\\

To illustrate the above theoretical part of the mixed strategies we extend the example ~\eqref{exsexes} of Battle of the Sexes, and allow players to choose mixed strategies.\\

\begin{exmp}
  \emph{(Battle of the sexes):}
  \label{Batlle} Game with two Nash Equilibrium in Pure Strategies but with another one Equilibrium in Mixed Strategies.\\
  \\As previously discussed the known example of battle of the sexes have two Pure Nash equilibria(equilibria in pure strategies). Now lets assume that the players can also opt mixed strategies.\\
  
  Let $\bm{p_1}=(x,1-x)$ be the mixed strategy in which player 1 plays Ballet with probability $x$ and let $\bm{p_2}=(y,1-y)$ the mixed strategy in which player 2 plays Ballet with probability $y$. By using these two mixed strategies we get the following payoffs.\\
  Payoff of player 1:
   $$ u_1(p_1,p_2)=  2xy+(1-x)(1-y) = (3y-1)x + 1-y$$
   Therefore the payoff of player 1 for $ x \in [0,1]$ corresponds to a line segment which its gradient and its section with the vertical axis depends on the mixed strategy $y\in[0,1]$ of player 2. If the gradient is positive then the maximum payoff of player 1 it will be at $x=1$, if gradient is negative at $x=0$ and if the gradient is equal to zero then the payoff of the first player is become maximum for all $ x \in [0,1] $. Thus, the best response of player 1 to the mixed strategy of player 2 is:
  $$u_1(\bm{p_1,p_2}) = \begin{cases} 0 &\mbox{if } y \lneq  1/3 \\ 
[0,1]& \mbox{if } y=1/3. \\ 
1& \mbox{if } y \gneq 1/3 \end{cases} $$

with similar way we could find that the payoff of player 2:
$$u_2(\bm{p_1,p_2})=  xy+2(1-x)(1-y) = (3x-2)y + 2(1-x)$$
is the best response to the mixed strategy of player 1 and 
$$u_2(\bm{p_1,p_2}) = \begin{cases} 0 &\mbox{if } x \lneq  2/3 \\ 
[0,1]& \mbox{if } x=2/3. \\ 
1& \mbox{if } x \gneq 2/3 \end{cases} $$

From the above,we can easily conclude that one tuple of strategies $(\bm{p_1,p_2})$ is a Nash equilibrium if corresponds to point of intersection of the two payoff functions $u_1(\bm{p_1,p_2}) $ and $u_2(\bm{p_1,p_2})$  and vice versa. 
In our example we have three equilibria. In particular, there are two equilibria in pure strategies: $(x,y)= (0,0)$ and  $(x,y)= (1,1)$ and one equilibrium in mixed strategies: $(x,y)= (2/3,1/3)$.

  \end{exmp}
  
  \begin{exmp}
  \emph{This is how a Graphical Representation of a game looks like:}
 \begin{figure}[H]
  \centering
    \includegraphics[scale = 0.7]{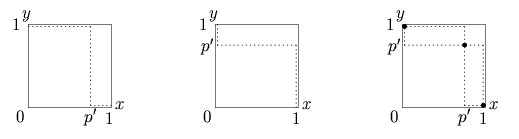}
  \caption{ {\footnotesize Example of a Graphical representation of a two player game. The Nash equilibria are exactly the intersections of the best responses. In  this example there are three points of intersection which are the equilibria. $(x,y)=(0,1),(x,y)=(1,0)$ and $(x,y)=(p',p')$. Image taken from http://commons.wikimedia.org/wiki/File:Reaction-correspondence-hawk-dove.jpg }}
  \label{dsgd}
\end{figure}
\end{exmp}

We conclude this section with the now very famous theorem of Nash who showed that any finite complete information game has an equilibrium if players can play mixed strategies.\\

\begin{thm}
  \emph{Nash(1950):}
  \label{nashtheorem}
In the n-player normal-form game $ G = \{S_i,S_2,...S_n;u_1,u_2,....u_n\}$ if $n$ is finite and $S_i$ is finite for every $i$ and the players allowed to choose mixed strategies then there exists at least one Nash equilibrium.
\end{thm}

Proof: Nash gave two proofs for this theorem one in \cite{nash1950equilibrium} based in Kakutamis Fixed Point Theorem \cite{kakutani1941generalization} and one in \cite{nash1951non} based on the Brouwer's Fixed point \cite{brouwer1911abbildung}. \qed \\

\subsection{Finite Games with Incomplete Information - Bayesian Games}

In the previous section we discussed about games with complete information, which means that all the parameters that structure the game are known from all players. Therefore the payoff functions are common knowledge. However in the real world situations, players are often uncertain about some aspects of the game's structure.\\ 
Harsanyi in \cite{harsanyi2004games} managed to model these \textbf{incomplete information games} and with the introduction of the notion of ``type"  he expressed the uncertainty of a player about the other players' payoff functions.\\

\subsubsection{Understanding of ``Types"}
More specifically, Harsanyi introduced a set of types $T_i$ for each player $i$. This set consists from a finite number of possible types $t_i$, which each one of them corresponds to a specific payoff $u_i(\bm{a_1,.....a_n};t_i)$. Therefore, by saying that player $i$ knows his payoff function is equivalent to say that he knows his type.\\

By using types, the need to establish a new, in some way extended(compare to the complete information games), notation is created.\\
With $u_i(\bm{a_1,.....a_n};t_i)$, we express the payoff function of player $i$ when he knows his type $t_i$ and when the player $j\in \{1,2,....,n\}$ choose to play the strategy $\bm{a_j}$.
The types of the others players except player i denoted by: $t_{-i}= (t_1,t_2,....t_{i-1},t_{i+1},.,t_n)$ and with $ p_i(t_{-i}|t_i)$ we express the belief of player i about the types of the other players.\\

The notion of belief, is also due to Harsanyi who suggested that the players of an incomplete information game share a common knowledge prior probability distribution over types. Given this distribution and with the knowledge of his own type, each player could derive his beliefs about the types of all other players using the Bayes rule. Exactly for this application of the Bayes rule these games called, by Harsanyi, "Bayesian Games".\\

\begin{defn}
\emph{(Normal Form representation of Bayesian Games)}\\
\label{Norm Bayes}
In Bayesian Games, a normal-form representation of a game is a specification of players' action spaces $A_1,....A_n$ (where with actions we denote all possible actions including the probability distributions over all possible pure actions), the type spaces $T_1,.....,T_n$, the beliefs of all players $p_1....p_n$ and the payoff functions $u_1....u_n$.In particular, a Bayesian game G is considered to be in normal form representation if it consists of $N={1,2,....,n}$ players, one set $T_i$ for each player $i$ which denote his type space ($t_i$ denote an arbitrary member of this set) and from one function $u_i(\bm{a_1,a_2,...,a_n})$ which represent the payoff of player $i$ if the players play the actions $\bm{a_1,a_2,...,a_n} $. Player i's  belief, ($ p_1(t_{-i}|t_i)$) describes i's uncertainty about the
$n - 1$ other players' possible types,$ t_{-i}$ (given $i$'s own type, $t_i$).We denote this game by $G= \{A_1,....A_n;T_1,.....,T_n;p_1....p_n;u_1....u_n\}$.\\
\end{defn}

When we developed the theory of complete information games we mentioned that the players simultaneously choose their actions and then they receive their expected payoffs. In Bayesian Games this is not the case because we have the concept of \textbf{private information} which means that the players know their own payoff function but they are uncertain for the types of all other players. For this reason the procedure until the receiving of the expected payoffs in Bayesian games is presented  by the following four main stages:

\begin{enumerate}
\item We assume (following Harsanyi's model, \cite{harsanyi2004games}) that nature is in the game. Nature draws a type vector $ t = (t_1,.....,t_n)$,where $t_i$ 
is drawn from the set of possible types $T_i$. There is a probability distribution over types, known to all players with which nature choose this specific type vector.
\item Then, nature reveals $t_i$  to player $i$ but not to any other player(private information) 
\item After that, the players simultaneously choose actions (player $i$ choosing $\bm{a_i}$, from the feasible set $A_i$). 
\item Finally the payoffs $ u_i(\bm{a_1,.....a_n};t_i) $ are received. \\
\end{enumerate}

Like the complete information game, to introduce the concept of Nash equilibrium in Bayesian game we must first present the definition of strategy.\\

\begin{defn}
\emph{(Strategy in Bayesian Games)}
\label{strategybay}\\
 In static Bayesian game $ G=\{A_1,....A_n;T_1,.....,T_n;p_1....p_n;u_1....u_n\}$, strategy of player $i$ is the function  $s_i:T_i\longrightarrow A_i$. For each type $t_i$, in $T_i$ ,$s_i(t_i)$ specifies the \textit{mixed strategy} from the feasible set $A_i$ that type $t_i$, would choose if drawn by nature.
\end{defn}

Again, like any other game the idea behind the Equilibrium concept is the same: A tuple of strategies is said to be equilibrium if each player's strategy is \textbf{best response} to the other player's strategies. \\

\begin{defn}
\emph{(Bayesian Nash Equilibrium)}
\label{Bayeqiul}

In the static Bayesian game $ G=\{A_1,....A_n;T_1,.....,T_n;p_1....p_n;u_1....u_n\}$, the tuple of  \textit{mixed strategies} $(\bm{s_1^*,..., s_n^*})$ is \textbf{Bayesian Nash equilibrium} if for each player $i$ and for each type $t_i \in T_i$  of $i$, $s_i^*(t_i)$ solves: 

 \[
    \max\limits_{a_i \in A_i}\sum\nolimits_{t_{-i} \in T_{-i}} u_i(s_{1}^*(t_{1}),....s_{i-1}^*(t_{i-1}),a_i,s_{i+1}^*(t_{i+1})....,s_1^*(t_n);t )p_i(t_{-i}|t_i)
\]  

With other words, for any player $i \in \{1,2,...N\}$ the $s_i^*(t_i)$ is the mixed strategy that gives him the maximum value in his payoff function. For this reason no one of the players wants to change his strategy even if this changing is only for one action in any type.  \\
\end{defn}

\begin{thm}
  \emph{(Existence of Bayesian Nash Equilibrium):}\\
  \label{HarsanyiTheo}
Every Finite Bayesian Game $ G=\{A_1,....A_n;T_1,.....,T_n;p_1....p_n;u_1....u_n\}$ has a Bayesian Nash Equilibrium.
\end{thm}

Proof: in  \cite{harsanyi2004games}.\\
\\

Subsequently to make the concept of Bayesian Games more comprehensive, we point some remarks, that following from one concrete example.\\

\begin{rem}
A very important assumption is that the strategy spaces, the payoff  functions, the possible types and the prior probability distribution over types are all \textbf{common knowledge}. That is, the players know every aspect of the game except the type that nature reveals to each player.\\
\end{rem}

\begin{rem}
By using the prior probability distribution and the knowledge of his type, each player can compute his beliefs $p_i(t_{-i}|t_i)$ about the other player's types. That is, given his type, each player can estimate all elements of the type vector using the Bayes rule:
$$   p_i(t_{-i}|t_i)=\frac{p_i(t_{-i},t_i)}{p_i(t_i)}=
\frac{p_i(t_{-i},t_i)}{\sum\nolimits_{t_{-i} \in T_{-i}}p_i(t_{-i},t_i)}$$\\
\end{rem}

\begin{exmp}
\emph{(Bayesian Nash Equilibrium in the Battle of the Sexes)}
\label{equisexes}\\

The example that we present is taken from the lecture notes of MIT course "Game Theory with Engineering Application" of Dr Asu Ozdaglar.\footnote{http://ocw.mit.edu/courses/electrical-engineering-and-computer-science/6-254-game-theory-with-engineering-applications-spring-2010/}\\

Lets recall the example of the Battle of the Sexes from the previous section. In this game we have found that there are two pure strategy equilibria and one mixed strategy equilibria. Now lets change it, and assume that player 2 has two \textbf{equal probability} choices. He can choose if he want to meet or avoid player 1.  With this change, we make our game Bayesian, because we could think that the two choices of player 2 (wishes to meet or wishes to avoid player 1) are his two possible types(each one of them can be chosen with probability 1/2).  \\

The following two matrices illustrate the Bayesian game in a more understandable way. We can think player 1(row player), as the player who could only have one type and player 2(column player) the player who has 2 possible types. These two types correspond to the two matrices. Therefore, each matrix is chosen with probability 1/2, like the probability of each type of player 2. Thus, after the nature reveals the type vector, its obvious that player 2 knows the exact game but player 1 must compute his beliefs to find the type of player 2.\\

\begin{table}[H]
\parbox{.45\linewidth}{
\centering
\begin{tabular}{ |c|c|c| }
    \hline
      & Ballet & Football \\ \hline
    Ballet & (2,1) & (0,0) \\ \hline
    Football & (0,0) & (1,2) \\
    \hline
  \end{tabular}
\caption{Type 1 of player 2}
}
\hfill
\parbox{.45\linewidth}{
\centering
\begin{tabular}{ |c|c|c| }
    \hline
      & Ballet & Football \\ \hline
    Ballet & (2,0) & (0,2) \\ \hline
    Football & (0,1) & (1,0) \\
    \hline
  \end{tabular}
\caption{Type 2 of player 2}
}
\end{table}

By the definition of this game we can understand that the possible strategies are only the pure ones, since each person could choose only one of his two choices. For example, it is not possible for player 1 to play mixed strategy (1/2,1/2) because there is no way, half of player 1 be located at Ballet and the other half at the football match.\\

In this game the strategy profile (B,(B,F)) is a Bayesian Equilibrium. This notation expresses the strategy that chosen by each player. In particular, it means that player 1 choose action Ballet, and player 2 choose Ballet, if he is in type 1 and choose Football if we is in type 2.
This is Bayesian Equilibrium because given the other player strategy each player's response is the best. B is the best response of player 1 given that the strategy of player 2 is (B,F). This is true because  
$$u_1(B,(B,F)) =1/2 \times2 + 1/2 \times 0 = 1 $$
 
but if player 1 plays F then his payoff is equal to 
 
 $$u_1(B,(B,F)) =1/2 \times0 + 1/2 \times 1 = 1/2 $$

In addition, if player 1 plays B then (B,F) is the best response of player 2 because for type 1 (left table), $1>0$ (1 and 0 are the two possible types of player 2 if player 1 choose Ballet) and for type 2(see right table), $2>0$.\\

With similar procedure we can check if a tuple of strategies is Bayesian Equilibrium for any Bayesian game.\\

\end{exmp}

\chapter{Robust Game Theory}

In 2006 Aghassi and Bertsimas \cite{aghassi2006robust} presented a distribution-free model of incomplete information games in which the players use a robust optimization approach to choose their strategies.This model relaxes the distributional assumptions of Harsanyi's Bayesian games, and it gives rise to an alternative distribution-free equilibrium concept.\\
In this chapter we focus on this recently class of Games. We begin by giving the new notation, that the two authors used in their paper and the one that we will use in the rest of this thesis. Following that,  we develop the basic theory of this new approach and make a small extension of one specific class of this games. At the end of this chapter,  several methods for finding the set of Robust Optimization  Equilibria are implemented. Our finding results are compared with the original paper.\\ 
Aghassi and Bertsimas developed both robust games with private and without private information. In this thesis we focus only in the case of games with \textbf{no private information}

\section{New Notation}

In section ~\eqref{distniel}, we discussed about the robust optimization approach for solving problems with parameters under uncertainty. More specifically we mentioned that in this approach one makes optimal decisions based on his worst case scenario. After that, we developed the basic theory related to games in order to make the reader more familiar with concepts like \emph{Best Response,Pure and Mixed Strategies and Equilibrium}. In both of these sections of Chapter 2 we used the most common notation of the literature.\\
In this section, we formulate again some of the basic notations and definitions that we will need in the rest of this thesis. We also, redefine the notions of Best Response and Equilibrium for the complete and incomplete games without private information.\\

Lets start from the N-player complete information games (Nash Games)in which player $i$ has $a_i$ possible actions. In these games, as we have mentioned in the Mathematical background Chapter the payoff matrix  $\bm{\check{P}} \in \mathbb{R}^{N\times\prod\limits_{i=1}^N a_i} $ is fixed, and as a result all players are able to know the exact payoff functions of all other players.\\
In particular,  ${\bm{\check{P}}}_{(j_1,j_2,....j_N)}^i$ denotes the payoff to player $i$ when player $k \in \{1,2,....N\}$ plays action $j_k \in \{1,2,....,a_k\}$ and 
\begin{equation}
S_{a_i} = \{ \bm{x^i} \in {\mathbb{R}}^{a_i}| \bm{x^i} \geq 0, \sum\limits_{{J_i}=1}^{a_i} {x^i}_{j_i}=1 \}
\end{equation}
expresses the set of all possible mixed strategies of player $i$ over all actions $\{1,2,...a_i\}$. Moreover, let $\pi_i(\bm{P;x^1,x^2,...x^N})$ indicate the expected payoff of player i when the payoff matrix is given by $\bm{P}$ and player $k \in \{1,2,....N\}$ plays mixed strategy $\bm{x^k} \in S_{a_k}$. That is,
\begin{equation}
  \pi_i(\bm{P;x^1,x^2,}...\bm{x^N}) = \sum\limits_{{j_1}=1}^{a_1}...\sum\limits_{{j_i}=1}^{a_i}... \sum\limits_{{j_{N}}=1}^{a_{N}}{\bm{P}}_{(j_1,j_2,....j_N)}^i \prod\limits_{i=1}^N {x^i}_{j_i}
 \end{equation}

Finally, in this thesis, we use exactly like Bertsimas and Aghassi (\cite{aghassi2006robust})  the following shorthands:
$$ \bm{x^{-i}} = (\bm{x^1,x^2,..,x^{i-1},x^{i+1},...x^N})$$
$$ (\bm{x^{-i},u^i}) =(\bm{x^1,x^2,..,x^{i-1},u^i,x^{i+1},...x^N})$$
$$ S= \prod\limits_{i=1}^N S_{a_i}$$
$$ S_{-i}= \prod\limits_{k=1,k\neq i}^N S_{a_k}$$\\

Using the this new notation we develop again the following very important definitions:\\

\textbf{Nash Games:}\\
In the complete information Games with fixed payoff matrix $\bm{\check{P}}$ the \textbf{best response} of player $i$ to the other players' strategies $\bm{x^{-i}} \in S_{-i}$ belongs to:
\begin{equation}
\label{br}
\underset{\bm{u^i}\in S_{a_i}}{\operatorname{argmax}}\,\pi_i(\bm{\check{P};x^{-i},u_i})
\end{equation}
and following the same reasoning of definitions ~\eqref{best} and ~\eqref{Nash} we can derive that the tuple of strategies $(\bm{x^1,x^2,...x^N})\in S$ is \textbf{Nash Equilibrium} if for each player $i\in \{1,2....N\}$:\\

\begin{equation}
\label{lllll}
\bm{x_i} \in \underset{\bm{u^i}\in S_{a_i}}{\operatorname{argmax}}\,\pi_i(\bm{\check{P};x^{-i},u_i})
\end{equation}\\

\textbf{Bayesian Games:}\\
In Bayesian Games the payoff matrix $\bm{\tilde{P}}$ is subject to uncertainty. As we have already mentioned we are interested for \textbf{games with no private information}. Therefore each player has only one possible type and players $i's$ \textbf{best response} to the other players' strategies $\bm{x^{-i}} \in S_{-i}$ belongs to:

\begin{equation}
\label{chicago}
\underset{\bm{u^i}\in S_{a_i}}{\operatorname{argmax}}\,[\underset{\bm{\tilde{P}}}E\pi_i(\bm{\tilde{P};x^{-i},u_i})]
\end{equation}

and similar with definition ~\eqref{Bayeqiul} the $(\bm{x^1,x^2,...x^N})\in S$ is a \textbf{Bayesian Nash Equilibrium} if for each player $i\in \{1,2....N\}$:\\

\begin{equation}
\label{lmljsd}
\bm{x_i} \in \underset{\bm{u^i}\in S_{a_i}}{\operatorname{argmax}}\,[\underset{\tilde{P}}E\pi_i(\bm{\tilde{P};x^{-i},u_i)}]
\end{equation}\\

\begin{rem}
\label{isa}
Notice that by linearity of $\pi_i$ over set $U$ and by linearity of expectation operator $\underset{\bm{\tilde{P}}}E$ the following is true:
\begin{equation}
\label{makalass}
\underset{\bm{\tilde{P}}}E \pi_i(\bm{\tilde{P};x^{-i},u_i}) = \pi_i ( \underset{\bm{\tilde{P}}}E [\bm{\tilde{P}}] ; \bm{x^{-i},u_i})
\end{equation}
where $\underset{\bm{\tilde{P}}}E [\bm{\tilde{P}}]$ is the element-wise expected value of $\bm{\tilde{P}}$.\\ 

This observation has as a result, that a Bayesian Game without private information is equivalent to a Complete Information Game (Nash Game) with payoff matrix equal to $\underset{\bm{\tilde{P}}}E [\bm{\tilde{P}}]$. (This can be shown by combining the equations ~\eqref{chicago}, ~\eqref{makalass} and ~\eqref{br}).
\end{rem}
 
\section{The Distribution-Free Model and the Basic Theory}

\subsection{The Robust Optimization Model}
To classify a game, in a specific category of games it must have some certain characteristics. For example one of the main characteristics of a Bayesian game is that the distribution over types is common knowledge. Following that, we understand that the new class of robust games can be described by specific features.

\subsubsection{What we call Robust Game}
We call robust the game with the following two features:
\begin{enumerate}
\item The players are not informed or choose to ignore any probability distribution over their payoff functions.The only knowledge, which is common, is that all players being aware about an uncertainty set, in which, all values of the uncertain parameters of payoff matrix belong. 
\item All players choose to use a robust optimization approach (see previous chapter) to the uncertainty and this is also common knowledge. Therefore, in this model of games, given the other players' strategies, each player try to maximise his worst case expected payoff.\\
\end{enumerate}
 
 \begin{defn}
  \emph{(Robust optimization Equilibrium):}\\
  \label{robustequi}
In an N-player robust game we called Robust optimization Equilibrium the tuple of strategies $(\bm{x_1,x_2,....x_N})$ iff given the other players' strategies $\bm{x^{-i}}$ each player $i$ plays the strategy that ensures him the maximum payoff under the worst case scenario. \\
\end{defn}

From the above modelling and analogous to Harsanyi's Bayesian Games(equations ~\eqref{chicago}, ~\eqref{lmljsd}), Aghassi and Berstimas defined the notions of best response and Robust Optimization Equilibrium for finite robust games with \textbf{no private information}.\\

Players i's \textbf{best response} to the other players' strategies $x^{-i}$ is the strategy that belongs to:\\

\begin{equation}
\label{robustbestresponse}
\underset{\bm{u^i}\in S_{a_i}}{\operatorname{argmax}}\,[\underset{\bm{\tilde{P}}\in U}\inf\pi_i(\bm{\tilde{P};x^{-i},u_i})]
\end{equation}

and $(\bm{x_1,x_2,....x_N})$ is a \textbf{Robust Optimization Equilibrium} if and only if for each player  $i\in \{1,2....N\}$,

\begin{equation}
\label{loioi}
\bm{x_i} \in \underset{\bm{u^i}\in S_{a_i}}{\operatorname{argmax}}\,[\underset{\bm{\tilde{P}}\in U}\inf\pi_i(\bm{\tilde{P};x^{-i},u_i})]
\end{equation}\\

\textbf{
Important observations about the new model has to be emphasized: }

\begin{rem}
\textbf{Infimum represents the worst case scenario }\\
One can think the infimum of formulas ~\eqref{robustbestresponse} and ~\eqref{loioi} with the following way. In Harsanyi model the players know the distribution over their payoff functions and for this reason each of them tries to maximize his expected payoff function (the average value of his profit) based on this distribution. Now, in case of robust games, players are not aware about any distribution. They only know that the payoff matrix of the game belongs in an uncertainty set U. Therefore, each player tries to maximize the payoff that he can receive under the worst case scenario, in the scenario that the payoff matrix is the one that will give him the minimum (infimum) profit. 
\end{rem}

\begin{rem}
\textbf{Benefit of no private Information assumption}\\
In both, Bayesian and Robust games we assume that the players commonly know all possible tuples of payoff parameters and all possible type vectors as well(in the occasion of private information). In the case now, that robust game has no private information each player has one and only specific type and that is a common knowledge. Consequently, we have only one type vector in each game which mean that all players have the same information about the uncertainty payoff matrix $\bm{\tilde{P}}$. 
\end{rem}

\begin{rem}
\label{infimum}
\textbf{In general not equivalent to a complete information game(Nash Game)}\\
In robust games, unlike Bayesian games(see ~\eqref{makalass}), the worst case expected payoff is greater than the expected worst case payoff(the element-wise worst case of the payoff matrix $\bm{\tilde{P}}$ that denote with the $\underset{P \in U}{inf} [\bm{\tilde{P}}]$) .
\begin{equation}
\label{sadassdd}
\underset{ \bm{\tilde{P}}\in U} {inf} \pi_i(\bm{\tilde{P}};\bm{x^{-i},u_i}) \geq \pi_i ( \underset{ \bm{\tilde{P}}\in U} {inf}[\bm{\tilde{P}}] ; \bm{x^{-i},u_i})
\end{equation}
The two parts of ~\eqref{sadassdd} are equal only in a special case of robust games in witch the uncertainty set U have a specific form.(see section ~\eqref{specialcase})
\end{rem}

\begin{rem}
\textbf{Reasonable combination of Equilibrium and Worst case notions}\\
One could argue that the concepts of equilibrium and worst case scenario are irrelevant. This is not true. By claiming that each player use a worst case approach we mean that given the strategies of all other players, he select a robust view only for the parameters that included in his payoff function and he totally ignores any other parameter of the game. \\
The robust optimization approach of the players is common knowledge. That is, all players know the best response of all other players in any combination of strategies. For this reason the players can predict the outcome of the robust game and like the complete information games they know which tuple of strategies is equilibrium.
\end{rem}

\subsection{Two concrete examples of finite robust games}
\label{finiterobustgames}
To illustrate the robust games, Aghassi and Bertsimas use the robust form of three popular games: The Inspection Game, The Free Rider Game and the Network Routing Game.
In this thesis in all experiments that we will make we focus on the Robust Free Rider Game and the Robust Inspection Game. In this subsection we present these two games. \\

\subsubsection{Robust Free Rider Game}

The \textbf{Free rider} is a two player game in which each player has two possible actions, to contribute or not contribute in the common good. If a player decides to contribute then he bears a cost of amount c (fixed in the complete information game($\check{c}$) but uncertain under the robust approach($\tilde{c}$)). If one of the two player decide to contribute then both of them enjoy a payoff of 1. The players choose strategies simultaneously. The representation of this game summarized in Table ~\ref{rfrider} where $\tilde{c}$ belongs in a the uncertainty set $[\check{c}- \Delta,\check{c} + \Delta]$.

\begin{table}[H]
\centering
\caption{Payoff matrix of Robust Free Rider Game}
\label{rfrider}
\begin{tikzpicture}[element/.style={minimum width=1.75cm,minimum height=0.85cm}]
\matrix (m) [matrix of nodes,nodes={element},column sep=-\pgflinewidth, row sep=-\pgflinewidth,]{
         & Contrib  & NoCon  \\
Contrib & |[draw]|($1-\tilde{c},1-\tilde{c}$) & |[draw]|($1-\tilde{c},1$) \\
NoCon  & |[draw]|($1,\,\,\,\,\,\,\,\,\,1-\tilde{c}$) & |[draw]|(0,0) \\
};

\node[above=0.25cm] at ($(m-1-2)!0.5!(m-1-3)$){\textbf{Player 2}};
\node[rotate=90] at ($(m-2-1)!0.5!(m-3-1)+(-1.25,0)$){\textbf{Player 1}};
\end{tikzpicture}
\end{table}

\subsubsection{Robust Inspection Game}
The Robust inspection Game is a two player game in which the row player is the employee (possible actions:Shirk or Work) and the column player is the employer( possible actions: Inspect or not Inspect). The two players choose their actions simultaneously and then they receive the corresponding to the combination of their strategies payoffs. When the employee works he has cost $\tilde{g}$ and his employer has profit equal to $\tilde{v}$. Each inspection costs to the employer $\tilde{h}$ but if he inspects and find the employee shirking then he does not pay him his wage w. In all other cases employee's wage is paid. All values except the payment w of the employee are uncertain. In the general case: $(\tilde{g},\tilde{v},\tilde{h}) \in [ \underline{g} , \overline{g} ]\times [ \underline{v} , \overline{v} ]\times [ \underline{h} , \overline{h} ]$. \\

\begin{table}[H]
\caption{Payoff Matrix for the Robust Inspection Game}
\centering
\begin{tabular}{ |c|c|c| }
    \hline
      & Inspect & NotInspect \\ \hline
    Shirk & ($0 ,-\tilde{h}$) & (w, -w) \\ \hline
    Work & ($w-\tilde{g},\tilde{v}-w-\tilde{h}$) & ($w-\tilde{g},\tilde{v}-w$) \\
    \hline
  \end{tabular}
\label{robinspectiongame}
  \end{table}

\subsection{Existence of robust optimization equilibria}
In the previous chapter we discussed about the two very important theorems ~\eqref{nashtheorem} and  ~\eqref{HarsanyiTheo} in which Nash and Harsanyi respectively proved the existence of Equilibrium in the complete(Nash Games) and incomplete information games(Bayesian Games).\\
Following that, Aghassi and Bersimas \cite{aghassi2006robust}  except that formalized the robust games they proved the existence of equilibrium in finite games with bounded uncertainty sets.\\
To show this, they first proved that if the uncertainty set U of all possible payoff matrices $\bm{\tilde{P}}$ is bounded then: 
\begin{enumerate}
  \item The worst case expected payoff functions of all players is continuous on $\mathbb{R}^{a_1+a_2+....a_N}$ 
  \item For any player $i\in \{1,2,...N\}$ and for any fixed $\bm{x^{-i}} \in S_{-1}$ the worst case expected payoff is concave in $\bm{x^i} $
\end{enumerate}

Using these two statements and Kakutani's Fixed Point Theorem \cite{kakutani1941generalization}(like Nash did in \cite{nash1950equilibrium}for the complete information games) they proved the following theorem. \\

\begin{thm}
  \emph{(Existence of Equilibria in Robust Finite Games)}\\
  \label{EQuirob}
Any N-person, non cooperative, simultaneous-move, one shot finite robust game in which the uncertainty set of the payoff matrices $U \subseteq \mathbb{R}^{N\times\prod\limits_{i=1}^N a_i}$ is bounded  and there is no private information, has an equilibrium.
\end{thm}

\begin{proof}
\cite{aghassi2006robust}
\end{proof}

\subsection{Computing sample equilibria of Robust Finite Games}
\label{k323}
As we have already discussed in section ~\eqref{gametheory} to examine if one specific tuple of strategies is equilibrium is a trivial procedure in any kind of game. However, to find the set of all equilibria of a game, with complete or incomplete information, is a very difficult task. To overcome this problem many algorithms were developed by the researchers. In most of them the game was transformed, in an equivalent problem like stationary point problems, unconstrained penalty function minimization problem or a system of miltilinear equalities and inequalities.\\ For example for the particular group of two player, non fixed-sum, finite games the most commonly accepted algorithm for finding the equilibria is the Lemke-Howson algorithm \cite{lemke1964equilibrium}.\\

In \cite{aghassi2006robust}, a theorem for approximately finding all equilibria of a \textbf{robust finite game with bounded polyhedral uncertainty set} and no private information was developed. In particular the authors showed that the set of equilibria is a component-wise projection of the solution set of a system of multi linear equalities and inequalities. By component-wise projection we simply mean that if one solution of the multilinear system is $(\bm{x_1,x_2,....x_N}, z_i,\phi_i,\theta_i)$ then the equilibrium of the corresponding game it will be $(\bm{x_1,x_2,....x_N})$. We keep the coordinates of the solution that express the mixed strategies of the players.\\

\emph{The main steps of the computational method are presented as follows:}\\
\begin{enumerate}
\item Check if the uncertainty set of the payoff matrix is bounded polyhedral. With other words,find $\bm{F} \in {\mathbb{R}}^{l \times mn}$ and $ \bm{d} \in {\mathbb{R}}^{l \times 1}$ that satisfy the equation ~\eqref{polyset} for all possible $vec(\bm{\tilde{P}})$.
\item Find the extreme points $\bm{G}\ell), \ell \in \{1,2,...k\}$ of the uncertainty set $U $.
\item Estimate the multilinear system which corresponds to the robust finite game. (following the theorem ~\eqref{equilicomput})
\item Calculate the corresponding to the multi-linear system penalty function $h(\bm{y})$(see ~\eqref{penalty}.
\item Solve the unconstrained minimization problem  $\underset{\bm{y} \in {\mathbb{R}}^v}{\operatorname{min}} h(\bm{y})$
\end{enumerate}

The first two steps of this procedure are trivial for each given game. For example, the uncertainty set of Free Rider Game(see subsection ~\eqref{finiterobustgames}) has only two extreme points( definition ~\eqref{extremepoint}) which are the two vertices $c = \check{c}- \Delta$ and $c = \check{c} + \Delta$.  \\

The most demanding part of the steps that we have to follow, in order to find the equilibria of a robust game is the estimation of the system of equations and inequalities.\\

If we are aware that the uncertainty set U of the game is closed and bounded then we can derive the following:\\

\begin{equation}
\begin{aligned}
(x_1,x_2,....x_N) \text{is equilibrium} 
& \Leftrightarrow \bm{x_i} \in \underset{\bm{u^i}\in S_{a_i}}{\operatorname{argmax}}\,\{\underset{\bm{\tilde{P}}\in U}\inf\pi_i(\bm{\tilde{P}};\bm{x^{-i},u_i})\}\qquad \forall i \enspace in\enspace \{1,2,...N\}\\
& \Leftrightarrow \underset{\bm{\tilde{P}}\in U}\inf\pi_i(\bm{\tilde{P}};\bm{x^{-i},x_i})\geqslant\underset{\bm{\tilde{P}}\in U}\inf\pi_i(\bm{\tilde{P}};\bm{x^{-i},u_i}),\, \forall u_i \enspace in\enspace S_{a_i},\,\, \forall i \enspace in\enspace \{1,2,...N\}\\
& \Leftrightarrow \underset{\bm{\tilde{P}}\in U}\inf\pi_i(\bm{\tilde{P}};\bm{x^{-i},x_i})\geqslant\underset{u^i\in S_{a_i}}{\operatorname{max}}\underset{\bm{\tilde{P}}\in U}\inf\pi_i(\bm{\tilde{P}};\bm{x^{-i},u_i})\qquad \forall i \enspace in\enspace \{1,2,...N\}\\
&\Leftrightarrow \underset{\bm{\tilde{P}}\in U}\min\pi_i(\bm{\tilde{P}};\bm{x^{-i},x_i})\geqslant\underset{u^i\in S_{a_i}}{\operatorname{max}}\underset{\bm{\tilde{P}}\in U}\min\pi_i(\bm{\tilde{P}};\bm{x^{-i},u_i})\qquad \forall i \enspace in\enspace \{1,2,...N\}\\
 \end{aligned}
 \end{equation}
 
The change from infimum to minimum in the last equivalence of the above expression is due to the fact that U is close and bounded uncertainty set.\\
 
From the above statements we understand that $(\bm{x_1,x_2,....x_N})$ is a robust optimization equilibrium of the finite robust game iff $\enspace\forall i= 1,2...N$ :\\
 
$\underset{\bm{\tilde{P}}\in U}\min\pi_i(\bm{\tilde{P}};\bm{x^{-i},x_i})-\underset{u^i\in S_{a_i}}{\operatorname{max}}\underset{\bm{\tilde{P}}\in U}\min\pi_i(\bm{\tilde{P}};\bm{x^{-i},u_i})\geq 0 \enspace $\\

with $e'\bm{x^i}=1 $ and $\bm{x^i} \geq 0$. \\

In the case now that the uncertainty set \textbf{U is bounded polyhedral} the above system can be  converted to a multi-linear system with equalities and inequalities. (see theorem ~\eqref{equilicomput})

\begin{defn}(extreme point)
\label{extremepoint}
An extreme point of set S is a point that belongs to S, which does not lie in any open line segment joining two points of S. 
Intuitively, in our case where our set S is polyhedral the extreme points are its vertices. 
\end{defn}

\begin{thm}
  \emph{(Computation of Equilibria in Robust Finite Game(from \cite{aghassi2006robust}))}
  \label{equilicomput}
Assume that we have N-player robust finite game($N\leq\infty$ with each player has $1 \leq a_i \leq \infty$ posssible  actions)  with no private information and in which the payoff uncertainty set $ U \subseteq {\mathbb{R}}^{N\prod\limits_{i=1}^N a_i}$ is a bounded polyhedral given by ~\eqref{polyset}.
Let $\bm{G}(\ell), \ell \in \{1,2,....k\}$ denote the extreme points of U. Then the following conditions are equivalent.\\

\textbf{Condition 1:}\\
$(\bm{x^1,x^2,....x^N})$ is an equilibrium of the robust game.\\

\textbf{Condition 2:}\\
For all players $i\in \{1,2,....N\}$ there exists $z_i \in \mathbb{R}, \bm{\theta^i} \in {\mathbb{R}}^k,\phi_i \in \mathbb{R} $such that $(\bm{x^1,x^2,....x^N},z_i,\bm{\theta^i},\phi_i)$ satisfies the following constraints:
 \begin{equation}
 \label{kkkkkk}
 \begin{split}
 &\quad z_i= \phi_i \\
 &\quad z_i - \pi_1 (\bm{G}(\ell); \bm{x^1,x^2,....x^N}) \leq 0, \qquad \ell=1,2,...,k \\
 &\quad \bm{e^\top x^i}=1 \\
 &\quad \bm{x^i} \geq 0 \\
 &\quad \bm{e^\top\theta^i} = 1 \\
 &\quad \sum\limits_{\ell=1}^k \ \theta_\ell^i (\bm{G}(\ell); \bm{x^{-i}, e_{j_i}^i)} - \phi_i \leq 0 \qquad j_i= 1,...a_i \\
 &\quad \bm{\theta^i} \geq 0
  \end{split}
 \end{equation}\\
 
\textbf{Condition 3:}\\
For all players $i\in \{1,2,....N\}$ there exists $\bm{\eta_i} \in {\mathbb{R}}^m, \bm{\xi^i} \in {\mathbb{R}}^{N\prod\limits_{i=1}^N a_i} $such that $(\bm{x^1,x^2,....x^N},\bm{\eta_i},\bm{\xi^i})$ satisfies the following constraints:
\begin{equation}
 \label{Condition3}
 \begin{split}
 &\quad \bm{{\xi_i}^\top Y^i(x^{-i}) e_{j_i}^i} \leq d^\top \eta^i \qquad j_i=1,2,...,a_i \\
 &\quad \bm{F^\top\eta_i - Y^i(x^{-i})x^i} = \bm{0}, \\
 &\quad \bm{e'x^i}=1 \\
 &\quad \bm{x^i} \geq \bm{0} \\
 &\quad \bm{\eta^i} \geq \bm{0}\\
 &\quad \bm{F\xi_i} \geq \bm{d}
  \end{split}
 \end{equation}

where $\bm{Y^i(x^{-i})} \in \mathbb{R}^{(N \prod_{i=1}^N a_i)\times a_i}$ denotes the matrix such that 
 \begin{equation}
 \label{ooo} 
 vec(\bm{P})^\top \bm{Y^i(x^{-i})x^i} = \pi_i(\bm{P;x^{-i},x^i})
  \end{equation} 

\end{thm}

\begin{rem}
All computations of this chapter are made based on system of Condition 2. The reason, is that the size (number of variable and constraints) of the multi-linear system of Condition 3 is much larger (see Table 1 in \cite{aghassi2006robust}).\\
However, in chapter ~\ref{chapter4}, when we develop our new distributionally robust approach(see theorem ~\eqref{alldistribequilibria}), the method for compute all equilibria is similar to the system of Condition 3(with no extreme points)\\  
\end{rem}

Up to this point, we have discussed that to approximately find the set of all equilibria of a robust game we must convert the problem to a system of equations and inequalities. 
The only part that still remains unspecified is how we can compute all feasible solutions of this multi-linear system. \\
An answer to this was also given by Bertsimas and Aghassi who used the following \textbf{Penalty Function}:\\

\begin{equation}
\label{penalty}
h(\bm{y})= 1/2 \sum\nolimits_{n \in E} [g_n(\bm{y})]^2 + 1/2 \sum\nolimits_{n \in I} [max(g_n(\bm{y}),0)]^2
\end{equation}

where $\bm{y} \in R^v, \quad g_n(\bm{y})=0 \quad if\quad n \in E \quad  and  \quad g_n(y)\leq 0\quad if \quad n \in I$. $y$ denotes the vector of all uncertain variables of the multi-linear system and each $g_n(\bm{y})$ one constraint (equality if $n \in E$ and inequality if $n \in I$).\\

Finally, by solving the unconstrained minimization problem $\underset{\bm{y} \in R^v}{\operatorname{min}} h(\bm{y})$ we collect the desired feasible solutions.

In the numerical results of \cite{aghassi2006robust} authors used what they called pseudo-Newton method with Armijo rule to solve the unconstrained optimization.\\
In this thesis we develop different methods (section ~\eqref{dsa}) for the estimation of the feasible solutions of the multi-linear system and we compare our results with the original paper (see table ~\eqref{table:nonlin}).

\section{Extension of the Special Class of robust games}
\label{specialcase}

As previously mentioned in remark ~\eqref{infimum}, in the general case of robust finite game the following is true:
\begin{equation}
\label{jimmys}
\underset{ \bm{\tilde{P}}\in U} {inf} \pi_i(\bm{\tilde{P}};\bm{x^{-i},u_i}) \geq \pi_i ( \underset{ \bm{\tilde{P}}\in U} {inf}[\bm{\tilde{P}}] ; \bm{x^{-i},u_i})
\end{equation}

However, there is one special class of games that the inequality in the previous equation become equality. Using this property for this class of games we reduce the necessary computations to the minimum, since instead of estimate the set of equilibria of our initial robust game we can simply compute the equilibria of the complete information game with fixed payoff matrix $\underset{ \bm{P} \in U} {inf}[\bm{P}]$.\\

\begin{thm}
\emph{(Special Class of Robust Games , from \cite{aghassi2006robust}) :}\\
  \label{special}
    Consider the robust game with the uncertainty set:
 \begin{equation}
   U = \{ \bm{P}(f_1,f_2,....f_v)| (f_1,f_2,....f_v) \in U_f\} \}
  \end{equation}
where 
\begin{equation}
\label{ses}
U_f= \{ (f_1,f_2,....f_v) | f_\ell \in [ \underline{f_\ell} , \overline{f_\ell} ], \ell \in \{1,2,....v\}\}.
\end{equation}
Suppose that $\forall i \in \{1,2,...N\}$ and $\forall \ell\in \{1,2,....v\} \quad \exists k(i,\ell)$ such that: 
\begin{equation}
sign( \frac{\partial}{\partial f_\ell} [P_{(j_1,....,j_n)}^i (f_1,f_2,....f_v)]_{(f_1,f_2,....f_v)=(\tilde{f_1},\tilde{f_2},....\tilde{f_v})}) = k(i,\ell)
\end{equation}
Then $(\bm{x^1,x^2,...x^N})$ is robust optimization equilibrium of this game if it is Nash equilibrium of the complete information game with same number of players and same number of actions $a_i$ for each player $i$. The payoff matrix for the complete information game is:
\begin{equation}
{\bm{Q}}_{(j_1,...j_N)}^i = {P}_{(j_1,...j_N)}^i (h_1^i,....h_v^i)
\end{equation}
where \\$$h_\ell^i = \begin{cases} \overline{f_\ell}  &\mbox{if } k(i,\ell) < 0 \\ 
\underline{f_\ell} & \mbox{if } k(i,\ell) \geq 0 \end{cases} $$
\end{thm}

\begin{rem}
Therefore, using this theorem, instead of solving the robust game like we developed in the previous section we can solve the complete information game with payoff matrix $ {\bm{Q}}_{(j_1,...j_N)}^i $. $\bm{Q}$ is the matrix that can be obtained if we choose the right values for each parameter $f_\ell \in [ \underline{f_\ell} , \overline{f_\ell} ]$ in order to achieve the minimum payoff for each player $ \forall (j_1,j_2,...,j_n) \in \prod \limits_{i=1}^N \{1,2...,a_i\}.$
\end{rem}

\subsubsection{The extension}  
Aghassi and Bertsimas in their theorem (theorem ~\eqref{special}) assumed that the uncertainty set must be equal to $ U = \{ P(f_1,f_2,....f_v)| (f_1,f_2,....f_v) \in U_f\}$ where $U_f$ is like ~\eqref{ses}.
In our extension each parameter $f_\ell$ its not necessary to belongs in a closed continuous set $[ \underline{f_\ell} , \overline{f_\ell} ]$. \\
The only necessary assumption is that the set where each parameter $f_\ell$ belongs must have maximum($ (maxf_\ell)$) and minimum ($ (minf_\ell) $) values.\\
That is, if $f_\ell \in K_\ell $ then $(maxf_\ell), (minf_\ell) \in K_\ell$ and this means that  $ \forall x\in K_\ell$ the following must be true: $$ (minf_\ell) \leq x \leq (maxf_\ell)$$.\\
In our formulation, set $ K_\ell$, will be always like this,  $\forall \ell\in \{1,2,....v\}$.

\begin{thm}
\emph{(Extension of special class):}
  \label{ourspecial}
    The formulation is exactly the same as the original theorem except that the uncertainty set is: 
  \begin{equation}
   U = \{ P(f_1,f_2,....f_v)| (f_1f_2,....f_v) \in U_f\} \}
  \end{equation}
where 
\begin{equation}
U_f= \{ (f_1,f_2,....f_v) | f_\ell \in K_\ell, \ell \in \{1,2,....v\}\}
\end{equation}
Thus, the parameters in the payoff matrix $ {\bm{Q}}_{(j_1,...j_N)}^i = P_{(j_1,...j_N)}^i (h_1^i,....h_v^i)$
take values:
$$h_\ell^i = \begin{cases} (max {f_\ell})  &\mbox{if } k(i,\ell) < 0 \\ 
(min {f_\ell}) & \mbox{if } k(i,\ell) \geq 0 \end{cases} $$

    \end{thm}

\begin{proof}
 If we show that the 
 \begin{equation}
 \label{hetra}
 \underset{\bm{\tilde{P}}\in U} {min} \pi_i (\bm{\tilde{P}; x^{-1},u^i})= \pi_i(\bm{Q;x^{-1},u^i}) 
 \end{equation}
 then
 \begin{equation}
\label{sdsdd}
 x_i \in \underset{\bm{u^i}\in S_{a_i}}{\operatorname{argmax}}\,\{\underset{\bm{\tilde{P}}\in U}\min\pi_i(\bm{\tilde{P}};\bm{x^{-i},u_i})\}=\underset{\bm{u^i}\in S_{a_i}}{\operatorname{argmax}} \pi_i(\bm{Q;x^{-1},u^i}) 
\end{equation} 
 and by definitions of the Robust Optimization equilibrium ~\eqref{loioi} and Nash Equilibrium ~\eqref{lllll} we obtain the desired result. Proving this we show that $(\bm{x^1,x^2,....,x^n})$ is equilibrium of the robust game if and only if is a Nash Equilibrium in complete information game with $\bm{Q}$ as a fixed payoff matrix.\\

$\bm{Q}$ is payoff matrix so it belongs to the uncertainty set $U$. This implies that: 
  \begin{equation}
  \pi_i (\bm{Q;x^{-1},u^i}) \geq {\underset {\bm{\tilde{P}}\in U} \min \pi_i (\bm{\tilde{P}};\bm{x^{-i},u_i})} \quad \forall i \in \{1,2,..N\} \quad \forall \{\bm{x^{-i},u^i}\}\in S.
  \end{equation}
  
Conversely, by the definition of matrix $\bm{Q}$ and the values of the parameters $(h_1^i,....h_v^i)$ the follow is true $\forall i \in \{1,2...,N\}$ and $\forall \{j_1,j_2,..j_N\} \in \prod_{i=1}^N\{1,2,...a_i\}$:
$${\bm{Q}}_{(j_1,...j_N)}^i \leq P_{(j_1,...j_N)}^i (f_1^i,....f_v^i),\,\, \forall (f_1,f_2,...f_v) \in U_f$$
Hence for every tuple of mixed strategies and for every player:

$${\underset {\bm{\tilde{P}}\in U} \min \pi_i (\bm{\tilde{P}};\bm{x^{-i},u_i})} = {\underset {\bm{\tilde{f}} \in U_f} \min \pi_i (P(\tilde{f_1},\tilde{f_2},..\tilde{f_v});\bm{x^{-i},u_i})}$$

 By the definition of expected payoff in complete information game and from the fact that 
 
$$ \underset {\bm{\tilde{f}} \in U_f} \min  P_{(j_1,...j_N)}^i (\tilde{f_1},\tilde{f_2},..\tilde{f_v}) =  P_{(j_1,...j_N)}^i (h_1,....h_v)$$
 we obtain the following:
\begin{equation}
\pi_i (\bm{Q;x^{-1},u^i}) \leq {\underset {\bm{\tilde{P}}\in U} \min \pi_i (\bm{\tilde{P}};\bm{x^{-i},u_i})} \quad \forall i \in \{1,2,..N\} \quad \forall \{\bm{x^{-i},u^i}\}\in S.  
  \end{equation}
  \end{proof}

 \subsubsection{Importance of the extension:}
 \label{importansce}
Consider the Robust Inspection Game of subsection ~\eqref{finiterobustgames}.
In the general case of this game the uncertainty set is:  $(\tilde{g},\tilde{v},\tilde{h})\in [ \underline{g} , \overline{g} ]\times [ \underline{v} , \overline{v} ]\times [ \underline{h} , \overline{h} ]$ which is like the desired form of equation ~\eqref{ses}.\\
For example if $ \tilde{g}\in [8,12], \tilde{v} \in [16,24], \tilde{h} \in [4,6] and w=15 $ then following the theorem ~\eqref{special} the initial robust inspection game is equivalent with the complete information game with payoff matrix:\\

\begin{table}[H]
\caption{Payoff Matrix for the Robust Inspection Game}
\centering
\begin{tabular}{ |c|c|c| }
    \hline
      & Inspect & NotInspect \\ \hline
    Shirk & (0 ,-6) & (15, -15) \\ \hline
    Work & (3,-5) & (3,1) \\
    \hline
  \end{tabular}
  \label{specificrobinspectiongame}
  \end{table}
  
Now with our extension(theorem ~\eqref{ourspecial}), instead solving the initial robust game we can solve the complete information game of table ~\eqref{specificrobinspectiongame} for more possible values of the uncertain parameters. The only limitation is that the set $K_\ell$ of each parameter $f_\ell$ must have maximum and minimum values.\\
For example the equilibria of a robust game with uncertainty parameters  $ \tilde{g} \in \{8,8.5,9,12\}, \tilde{v} \in [16,18] \cup \{23,24\}, \tilde{h} \in \{4,6\}$ and $ w=15 $ are exactly the same with the complete information game of table ~\eqref{specificrobinspectiongame}.

\section{Experimentation and Implementation of algorithms}
\label{dsa}

As previously discussed in subsection ~\eqref{k323}, to approximately compute the equilibria of a robust game, Aghassi and Bertsimas first estimate a multi-linear system (see theorem ~\eqref{equilicomput} ) and then to find the feasible solutions of this system they minimise the penalty function ~\eqref{penalty}, using Pseudo-Newton Method with Armijo rule.\\

In this section we provide three other methods for approximately computing sample robust optimization equilibria. In particular, using these methods we solve the two games that Berstimas and Aghassi first presented in \cite{aghassi2006robust} and that we have already described in subsection ~\eqref{finiterobustgames}. Then we make a comparison of our results with the original paper.\\

The specific robust games that we deal with are:
\begin{itemize}
\item The Robust Free Rider Game with uncertainty set $\tilde{c} \in [ \underline{c} , \overline{c} ]= [1/4, 5/8]$
\item The Robust Inspection Game with
$(\tilde{g},\tilde{v},\tilde{h})\in [ \underline{g} , \overline{g} ]\times [ \underline{v} , \overline{v} ]\times [ \underline{h} , \overline{h} ]= [8,12]\times [16,24] \times [4,6]$ and $w=15$.
\end{itemize}

For the implementation of the methods we used Matlab R2014b and all numerical evaluations of this section were conducted on a 2.27GHz, Intel Core i5 CPU 430 machine with 4GB of RAM.\\

\subsubsection{The new methods}

The methods that we use for solving the aforementioned games are the following:
\begin{itemize}
\item \emph{BFGS algorithm with Armijo rule}: Like Bertsimas and Aghassi estimate the multi-linear system and then minimise the penalty function ~\eqref{penalty} using BFGS algorithm with Armijo rule.
\item \emph{Steepest Descent method with Armijo rule}: Again estimate the multi-linear system but now for the minimization of the penalty function ~\eqref{penalty} we use Steepest Descent  with Armijo rule.
\item \emph{Modelling method using the Matlab toolbox of YALMIP}.\cite{lofberg2004yalmip}
\end{itemize}

In all methods, for the computation of the desired multi-linear system we used Theorem ~\eqref{equilicomput} and we formulated the set of equilibria using System ~\eqref{kkkkkk}.\\

\underline{\emph{BFGS and Steepest Descent algorithms}}\\
In all runs of these two algorithms we initialized the strategy $\bm{x^i}$ of each player $i \in \{1,2...N\}$ in a way that satisfy the non negativity and normalization constraints. The vector $\bm{\theta^i}$, that corresponds to the number of the extreme points, was initialized in the same way. Furthermore, we initialized the variable $z_i$ of each player to be the maximum value of the upper-bound of the constraint on $z_i$ and the variable $\phi_i$ to be either equal to $z_i$ or equal to the minimum value imposed by the lower-bound constraint on $\phi_i$.\\

To be able to compare the pseudo-Newton Method with the Steepest Descent and the BFGS algorithms that we create, the penalty method ~\eqref{penalty} of the original paper was employed.
For the implementation of these algorithms we used the Optimization Toolbox of Matlab. More specifically we used the \emph{fminunc} function which is able to provide the minimum of unconstrained multi-variable function.

\begin{equation}
[\bm{y},fval,exitflag,output] = fminunc ( h(\bm{y}),  \bm{Y_0}  )\\
\end{equation}

As inputs \emph{fminunc} accepts the $h(\bm{y})$ which is the penalty function~\eqref{penalty} that we want to minimise and $Y_0$ which represents the starting point of the algorithm. In the outputs, $y$ is the desired minimum point of $h(\bm{y})$, fval is the value of the objective function $h(\bm{y})$ at the solution $\bm{y}$, exitflag is a value that describes the exit condition and output is a structure that contains information about the optimization.\\
By using this function is not necessary to calculate the gradient or the Hessian matrix of the penalty method $h(\bm{y})$ that we have to minimize.\\
In particular, when we use fminunc, two line search strategies are used, depending on whether gradient information is readily available or whether it must be calculated using a \textbf{finite difference method}. When gradient information is available, the default is to use a cubic polynomial method. When gradient information is not available which is our case , the default is to use a mixed quadratic and cubic polynomial method.\\

The \textbf{termination criteria} for each run of these two methods are the same with  \cite{aghassi2006robust}.
\begin{itemize}
\item Robust Free Rider Game: We terminated each run of the algorithms either when the objective function fallen below $10^{-10}$ or if the number of iterations exceeded 2000. 
\item Robust Inspection Game: We terminated the each run of the algorithms if the objective function fallen bellow $10^{-8}$.
\end{itemize}

\underline{\emph{Modelling method using YALMIP}}\\
The YALMIP model that we generated is directly computed from the multi-linear system. Therefore, is not necessary to use any penalty function like the other two algorithms. This is great advantage since there is no need for extra calculations. The output of this algorithm is a 3-dimensional matrix whose each 2-dimensional face represents the strategies of the two players that form an equilibria. Each of these 2-dimensional face has the following form:
\begin{equation*}
\begin{pmatrix}
x^1_1 & x^2_1\\
x^1_2 & x^2_2 
\end{pmatrix}
\end{equation*}

\subsubsection{Results and Comparison of the Methods}

Note that all methods produce exactly the same equilibria with the original paper. That is, for the robust free rider game with $\tilde{c} \in  [1/4, 5/8]$ we obtain three equilibria ($x^1_1, x^2_1$): (1,0),(0,1) and (3/8,3/8) and for the robust inspection game with $\tilde{g} \in [8,12], \tilde{v}\in [16,24], \tilde{h} \in [4,6]$ and $w=15$ only one unique equilibrium ($x^1_1, x^2_1$)=(0.4,0.8).\\

Nevertheless, our goal is not to find the equilibria of these games but to compare which one of these methods compile faster.\\
The differences of the four methods that we studied are summarized on Table ~\ref{table:nonlin}.
The main compared features of the algorithms are the \textbf{time} that each algorithm takes to compile and the number of \textbf{average iterations}.  The first row of each game corresponds to the Pseudo-Newton method and the results that Berstimas and Aghassi found when they run their method. The rest three rows show the results of our experimentation for the three other methods with which we choose to solve the two games.\\

To estimate the compilation time for the BFGS and Steepest Descent algorithms we use the \emph{``tic-toc"} command of Matlab and for the YALMIP method we use the \emph{``sol.solvertime"} command, where sol is the solution of our model.\\

For the Robust Inspection game in \cite{aghassi2006robust}, authors run their numerical method only one time since they know that this game has unique equilibrium (0.4,0.8). In our results we run all methods for both games 15 times. This is the reason for the difference between sum and average time in the robust inspection game.

\begin{table}[H]
\caption{Comparison of the 4 algorithms which approximately compute sample robust optimization equilibria (time is presented in seconds)}
\centering
\begin{tabular}{c c c c c c}
\hline
Algorithms & min-time & max-time & sumtime & average-time &average-iterations \\ [0.5ex]
\hline
\textbf{\emph{Robust Free Rider Game}} \\
Pseudo-Newton&-&-&110.748&7.3827&1652.1 \\
BFGS&0.0565& 0.1063&1.2740& 0.0849&52.9333  \\
Steepest-Descent&0.0495&0.1177 &1.3080 &0.0872& 44.9333 \\
YALMIP &  0.2091& 0.2657 &  3.4716 &  0.2314 & -  \\ [1ex]
\hline
\textbf{\emph{Robust Inspection Game}} \\
Pseudo-Newton&-&-&-&0.5&71 \\
BFGS&  0.1179&0.1823&1.9331&0.1289&74.0667 \\
Steepest-Descent&   0.1186&0.1564 & 2.1241&  0.1416&70.6 \\
YALMIP & 0.2688 & 1.8690 &  9.0447& 0.6030& - \\ [1ex]
\hline
\end{tabular}
\label{table:nonlin}
\end{table}

\textbf{Observations and Remarks:}
\begin{itemize}
\item  As we expect the BFGS algorithm is faster than the Steepest Descent (see section ~\eqref{unconOptimization}). In fact the fastest method in both problems is the BFGS and the slowest is the modelling with YALMIP. \\
Even though YALMIP is the slowest method, we choose to use it because there is no need for extra pre-calculations that are prone to human error. (YALMIP model is directly formulated from the multi-linear system so the using of penalty function is unnecessary).

\item Note, that the average computational time of all methods that we present is preferred compare to the original Pseudo-Newton Method of \cite{aghassi2006robust}. \\
We consider that this difference is attributable to the use of different computers.
\end{itemize}

\chapter{The New Model - Distributionally Robust Game Theory} 
\label{chapter4}
In this chapter we present for the first time in the literature a new model of incomplete information games without private information in which the players use a distributionally robust optimization approach to cope with payoff uncertainty. 

In the first part of this chapter we propose the new model of games and show that under specific assumptions about the ambiguity set and the values of risk levels, Distributionally Robust Game constitutes a true generalization of the three finite games that we developed in previous chapters (Nash games, Bayesian Games and Robust Games).

In the next section, we prove that the set of equilibria of an arbitrary distributionally robust game with ambiguity set like the one of equation ~\eqref{adas} and without private information can be computed as the component-wise projection of the solution set of a multi-linear system of equations and inequalities. 

Subsequently, for special cases of such games we show equivalence to complete information finite games (Nash Games) with the same number of players and same action spaces. Thus, when our game falls within these special cases one can simply solve the corresponding Nash Game.

Finally to concretize the idea of a distributionallly robust game we present two examples, the Distributionally Robust Free Rider Game and the Distributionally Robust Inspection Game.\\

$$\hbox{\rule{2cm}{1pt}}$$

\begin{equation}
 \label{adas}
\mathcal{F} = \{ Q : Q[\bm{W} \cdot vec(\bm{\tilde{P}})\leq \bm{h}] = 1 , \;\; \mathbb{E}_Q [vec \bm{\tilde{P}}] = \bm{m}, \;\; \mathbb{E}_Q [\left \lVert vec (\bm{\tilde{P}}) - \bm{m} \right \rVert_1] \leq s \}
\end{equation}

We require that any probability distribution that belongs in this ambiguity set must satisfy three properties:
\begin{enumerate}
\item  All payoff matrices $\bm{\tilde{P}}$ of the probability distributions must belong to an uncertainty bounded polyhedral set $\bm{U}= [\bm{P} : \bm{W} \cdot vec(\bm{P})\leq \bm{h}] $.
\item All distributions of $\bm{\tilde{P}}$ that belong in this ambiguity set must necessarily have fixed expected value and equal to vector $\bm{m}$.
\item We require that any payoff matrix of any distribution Q must be close to its expected value. The maximum possible distance is denoted with variable s. The 1-norm is preferable since it leads to linear constraints which are easier to handle numerically.\\
\end{enumerate}

\section{Formulation of the New Model}
\label{mlml}

In previous chapters we discussed about three very important classifications of Games. First in Chapter 2 we developed the Static Games of Complete Information(Nash Games) and the Static Games of Incomplete Information(Bayesian Games). Then in the Chapter 3 we focused on the recently proposed class of distribution-free model of incomplete-information games (Robust Games).\\
The theory for every aforementioned game was developed based on two very important concepts of Game Theory area: \textbf{the Best Response and the Equilibrium}.\\
Following this, we find appropriate to introduce the new model of Distributionally Robust Games by giving these two definitions and explain them later in details.

\begin{defn}
\label{bestreD}
\emph{(Best Response in Distributionally Robust Games)}\\
In the distributionally robust model, for the case without private information, players i's best response to the other players strategies $\bm{x^{-i}} \in S{-i}$ must belong to:\\

\begin{equation} 
\underset{u^i\in S_{a_i}}{\operatorname{argmin}}\,\sup_{Q \in \mathcal{F}} \; Q\text{-CVaR}_{\varepsilon_i} [-\pi_i(\bm{\tilde{P};x^{-i},u^i})]  
\end{equation}\\

\end{defn}

\begin{defn}
\label{DROE}
\emph{ (Distributionally Robust Optimization Equilibrium)}\\
$(\bm{x^1,x^2,}...\bm{x^N}) \in S$ is said to be a Distributionally Robust Optimization Equilibrium of the corresponding game with incomplete information iff $\forall i \in \{1,2,..N\},$\\

\begin{equation} 
\bm{x^i} \in \underset{u^i\in S_{a_i}}{\operatorname{argmin}}\,\sup_{Q \in \mathcal{F}} \; Q\text{-CVaR}_{\varepsilon_i} [-\pi_i(\bm{\tilde{P};x^{-i},u^i})]  
\end{equation}\\
\end{defn}

As we developed in Section ~\eqref{distniel} Distributionally Robust Optimization is closely related to both stochastic programming and robust optimization. Based on this, our new model of games can be considered a concept closely related to Harsanyi's Bayesian Games and Robust Games that we previously discussed. \\
In more detail, in Bayesian Games we assume that all players of the game know the exact distribution of the payoff matrix. Now, in the Distributionally Robust approach the players do not know the exact distribution. Instead, they only be aware of a common known ambiguity set $\mathcal{F}$ of all possible probability distributions $Q$ that satisfy some specific properties. These distributions have no restriction in their form. That is, the ambiguity set may consists by both, discrete and continuous distributions of the payoff matrix.\\
In addition, similar to the robust games framework, we assume that each player  adopts a worst case approach to the uncertainty. Only now the worst case is computed over all probability distributions within the set  $\mathcal{F}$. For this formulation that is similar with distributionally robust optimization concept we named these games \textbf{Distributionally Robust Games} and we refer to their equilibrium as \textbf{Distributionally Robust Optimization Equilibrium}.\\

\subsubsection{Explanation of the CVaR in the formula}

As we have mentioned in subsection ~\eqref{CVARDEFIni} Conditional Value at Risk(CVaR) is one of the most popular quantile-based risk measures because of its desirable computational properties. \\
Exactly for these properties we chose to introduce CVaR in the formulas of best response and distributionally robust optimization equilibrium of the new model.\\ 
Using $Q\text{-CVaR}_{\varepsilon_i}$, we allowed the players to have several risk attitudes which is a major difference compared to all other games that we thoroughly analysed in this thesis and in which the players are always risk neutral. Important hypothesis is that risk attitude is a fixed characteristic of each player and it can not be changed depending the game. It is not a notion like the mixed strategy that a player can choose in order to achieve his best response and minimise his loss. More specifically, the parameter $\varepsilon_i \in (0,1)$ determines the risk-aversion of each decision-maker. In detail, if player i has risk level $\varepsilon_i = 1$ this means that he is risk neutral since the Conditional Value at Risk is equal to the expected value of his loss function ($Q\text{-CVaR}_{\varepsilon_i} = \mathbb{E}_Q $). On the other hand if $\varepsilon_i \leq 1$ the player is risk averse and as $\varepsilon_i \longrightarrow 0$ the risk aversion of the player become larger. \\
Conclusively, as parameter $\varepsilon_i$ decreases the value of  $Q\text{-CVaR}_{\varepsilon_i}$ increases and the risk aversion of the player becomes larger and vice versa. \\
With the introduction of a risk measure in our model we take into account not only that player wish to maximize his gain (minimize loss) but and how much is willing to risk to achieve this maximum value(minimum value).\\

\textbf{$1^{st}$ assumption:}\\
\emph{The first assumption that we have to make for the new model is that the risk attitude of each player is assumed to be common knowledge. That is, each player knows how much risk averse are the other players and that all other players know that he knows.} \\

In general, CVaR can be calculated from either the probability distribution of gains or the probability distribution of losses. In this thesis we decide to follow the original formulation of Rockafellar and Uryasev (see \cite{rockafellar2002conditional} and \cite{rockafellar2000optimization}) and calculate it from the distribution of losses.\\
For this reason, in definitions of best responce (definition ~\eqref{bestreD}) and distributionally robust optimization equilibrium  (definition ~\eqref{DROE}) we use the expected loss function of player $i$, $-\pi_i(\bm{P;x^{-i},x^i})$.
Loss distributions are also responsible for the using of $ \underset{\bm{u^i}\in S_{a_i}}{\operatorname{argmin}}\,\sup_{Q \in \mathcal{F}} \;$ instead of $\underset{\bm{u^i}\in S_{a_i}}{\operatorname{argmax}}\,\inf_{Q \in \mathcal{F}} \;$ that we use in robust games. 

\begin{rem}
In distrinutionally robust game the `performance' of player's mixed strategy is measured by his expected loss. Given the other player's strategies each player seeks to minimize his worst case CVaR. The worst case CvaR is taken with respect to the ambiguity set, and the expected loss is taken over the mixed strategies of the players.
\end{rem}

Finally, in the formulation of the distributionally robust game we have to make another two assumptions.\\

\textbf{$2^{nd}$ assumption:}\\
\emph{The players commonly know the ambiguity set of all possible distributions (discrete and continuous) of the payoff matrix.} \\

\textbf{$3^{rd}$ assumption:}\\
\emph{Each player adopts, like the Robust games, a worst case approach to the uncertainty, only now the worst case is computed over all probability distributions within the set  $\mathcal{F}$. In particular we assume that all players use a worst case CVaR approach.}

\subsection{A generalization of all other finite Games}

From the formulation of the distributionally robust games (definitions of best response and equilibrium) we can easily understand that only two are the parameters that are amenable to change. These are, risk level  $\varepsilon_i$ of each player i and the ambiguity set $\mathcal{F}$.
In particular, we assume that parameter $\varepsilon_i$ which shows the risk level of each player can take any value in the interval $(0,1)$ and we make no assumptions for the ambiguity set. That is, depending on the game that one faces he can choose the more suitable properties that the distributions of the ambiguity set must satisfy.\\

Hence, if we assume some extra constraints for parameter $\varepsilon_i$ and set $\mathcal{F}$, the previous general formulation can become very specific.\\

Let's assume that all players $i \in \{1,2,...N\}$ have the same risk level, $\varepsilon_i = 1$. Then from the definition of CVaR we obtain the following $ \forall i \in \{1,2,....,N\}$:

\begin{equation}
\label{elnet3}
\begin{aligned}
Q\text{-CVaR}_{\varepsilon_i} [-\pi_i(\bm{\tilde{P};x^{-i},u^i})] 
& = Q\text{-CVaR}_{1} [-\pi_i(\bm{\tilde{P};x^{-i},u^i})] \\
& = \mathbb{E}_Q [ -\pi_i(\bm{\tilde{P};x^{-i},u^i})]  
\end{aligned}
\end{equation}

Therefore the definition of best response ~\eqref{bestreD} become:
\begin{equation}
\label{elnet4}
\begin{aligned}
\underset{u^i\in S_{a_i}}{\operatorname{argmin}}\,\sup_{Q \in \mathcal{F}} \; Q\text{-CVaR}_{\varepsilon_i} [-\pi_i(\bm{\tilde{P};x^{-i},u^i})] 
& = \underset{u^i\in S_{a_i}} {\operatorname{argmin}} \,\sup_{Q \in \mathcal{F}} \; \mathbb{E}_Q [-\pi_i(\bm{\tilde{P};x^{-i},u^i})] \\
& = \underset{u^i\in S_{a_i}} {\operatorname{argmin}} \,\sup_{Q \in \mathcal{F}} \; [-\pi_i(\bm{\mathbb{E}_Q [\tilde{P}];x^{-i},u^i})]\\
& =  \underset{u^i\in S_{a_i}} {\operatorname{argmax}} \,\inf_{Q \in \mathcal{F}} \; [\pi_i(\bm{\mathbb{E}_Q [\tilde{P}];x^{-i},u^i})] 
\end{aligned}
\end{equation}\\

The equality in the second line in the above expression follows from the linearity of expectation operator and the linearity of $\pi_i$. See remark ~\eqref{isa} that $\mathbb{E}_Q [\pi_i(\bm{\tilde{P};x^{-i},u^i})]  = [\pi_i(\bm{\mathbb{E}_Q [\tilde{P}];x^{-i},u^i})]$ where $\bm{\mathbb{E}_Q [\tilde{P}]} $is the component-wise expected value of $\bm{\tilde{P}}$.
The equality in the third line is due to the following properties of the linear functions:
$$\max f(x) = - \min [-f(x)]$$ and
$$z\in \underset{x \in S_{a_i}} {\operatorname{argmin}} f(x) = z\in \underset{x\in S_{a_i}} {\operatorname{argmax}} -f(x).$$\\

Now, using the assumption that $\varepsilon_i = 1 \quad \forall i \in \{1,2,...N\}$ and by choosing the right properties for the probability distributions of the payoff matrix we can specify our formulation and create the desire games.\\

\begin{thm}
\emph{(Distributionally Robust Games = Generalization of all other games)}\\
  \label{allequilibria}
  In the distributionally robust games players i's best response to the other players strategies $\bm{x^{-i}} \in S{-i}$ must belong to:\\
\begin{equation} 
\underset{u^i\in S_{a_i}}{\operatorname{argmin}}\,\sup_{Q \in \mathcal{F}} \; Q\text{-CVaR}_{\varepsilon_i} [-\pi_i(\bm{\tilde{P};x^{-i},u^i})]  
\end{equation}\\

If we assume that $\varepsilon_i = 1 \quad \forall i \in \{1,2,...N\}$\\
 Then:
\begin{enumerate}
\item If $ \mathcal{F} = \{ Q : \mathbb{E}_Q [\bm{\tilde{P}}] = \bm{\Psi} \}$ the set of equilibria of the distributionally robust game is equivalent to that of a classical Nash Game.
\item If the ambiguity set is singleton, that is, $ \mathcal{F} = \{ Q \}$, the distributionally robust game game have the same equilibria with a related finite Bayesian Game.
\item If $ \mathcal{F} = \{ Q : Q[\bm{W} \cdot vec(\mathbb{E}_Q [\bm{\tilde{P}}])\leq \bm{h}] = 1$ the set of the distributionally robust optimization equilibria is equivalent to that of the classical Robust game model of Chapter 3.\\
\end{enumerate}
\end{thm}  
  
  \begin{proof}
In ~\eqref{elnet4} we have shown that when $\varepsilon_i = 1 \quad \forall i \in \{1,2,...N\}$ the next equality is true:
\begin{equation}
 \underset{u^i\in S_{a_i}}{\operatorname{argmin}}\,\sup_{Q \in \mathcal{F}} \; Q\text{-CVaR}_{\varepsilon_i} [-\pi_i(\bm{\tilde{P};x^{-i},u^i})] =  \underset{u^i\in S_{a_i}} {\operatorname{argmax}} \,\inf_{Q \in \mathcal{F}} \; [\pi_i(\bm{\mathbb{E}_Q [\tilde{P}];x^{-i},u^i})]
 \end{equation}

 Using this property we obtain the following results:\\
 
\textbf{ For Nash Games:}\\

  If $ \mathcal{F} = \{ Q : \mathbb{E}_Q [\bm{\tilde{P}}] = \bm{\Psi} \}$ 
  
\begin{equation}
\begin{aligned}
\underset{u^i\in S_{a_i}} {\operatorname{argmax}} \,\inf_{Q \in \mathcal{F}} \; [\pi_i(\bm{\mathbb{E}_Q [\tilde{P}];x^{-i},u^i})] 
& = \underset{u^i\in S_{a_i}} {\operatorname{argmax}} \,\inf_{Q \in \mathcal{F}} \; [\pi_i(\bm{\Psi;x^{-i},u^i})]\\
&= \underset{u^i\in S_{a_i}} {\operatorname{argmax}} \;  [\pi_i(\bm{\Psi;x^{-i},u^i})]
\end{aligned}
\end{equation}

which is equivalent to the formulation of best response in the Nash Games (see equation ~\eqref{br}).\\

\textbf{Bayesian Games}:\\

If the ambiguity set is singleton. That is that $ \mathcal{F} = \{ Q \}$.

\begin{equation}
   \underset{u^i\in S_{a_i}} {\operatorname{argmax}} \,\inf_{Q \in \mathcal{F}} \; \mathbb{E}_Q [\pi_i(\bm{\tilde{P};x^{-i},u^i})]  = \underset{u^i\in S_{a_i}} {\operatorname{argmax}} \,\mathbb{E}_Q [\pi_i(\bm{\tilde{P};x^{-i},u^i})]
   \end{equation}\\
   
which is equivalent to the formulation of best response in the Bayesian Games (see equation ~\eqref{chicago}).\\  

\textbf{For Robust Game:}\\

If  $ \mathcal{F} = \{ Q : Q[\mathbb{E}_Q [\bm{\tilde{P}}] \in \mathcal{U}] = 1\} $ where $\mathcal{U} = \{ \bm{P} :\bm{W} \cdot vec(\bm{P})\leq \bm{h} \}$\\

Therefore
\begin{equation}
\underset{u^i\in S_{a_i}} {\operatorname{argmax}} \,\inf_{Q \in \mathcal{F}} \; [\pi_i(\bm{\mathbb{E}_Q [\tilde{P}];x^{-i},u^i})] = \underset{u^i\in S_{a_i}} {\operatorname{argmax}} \,\inf_{\bm{\tilde{S}} \in \mathcal{U}} \; [\pi_i(\bm{\tilde{S};x^{-i},u^i})] 
\end{equation}
where $\bm{\tilde{S}} =  \bm{\mathbb{E}_Q [\tilde{P}]};$

which is equivalent to the formulation of best response in the Robust Games (see equation ~\eqref{loioi}).\\

\end{proof}

\textbf{Conclusions:}\\

We named \textbf{Distributionally Robust Game} the incomplete information game in which:
\begin{itemize}
\item The players commonly know the ambiguity set $\mathcal{F}$ of all possible distributions $Q$ of the payoff matrix $\bm{\tilde{P}}$.
\item Each player adopts, like the Robust games, a worst case approach to the uncertainty, only now the worst case is computed over all probability distributions within the set  $\mathcal{F}$. In particular we assume that all players use a worst case CVaR approach.
\item The risk attitude of each player is assumed to be common knowledge. That is, each player knows the risk levels of all other players and that all other players know that he knows.
\end{itemize}

\section{Computing sample equilibria of Distributionally Robust Games}
In section ~\eqref{k323}, we presented the general method that Bertsimas and Aghassi \cite{aghassi2006robust} used for computing sample robust optimization equilibria of a robust finite game. We also discussed about the difficulty of finding all equilibria of a game and how nowadays the researchers developed various approaches for the solution of this challenge. For example, one of the most widely used numerical technique to overcome this problem is the utilization of a multi linear system.(see \cite{govindan2003global} and \cite{sturmfels2002solving}).\\

In this section, we present theorem ~\eqref{alldistribequilibria}, in which we show that for any finite distributionally robust games with ambiguity set like the one defined in ~\eqref{adas}, and with no private information the set of equilibria is a projection of the solution set of multilinear system of equalities and inequalities. The projection is like in \cite{aghassi2006robust}, the component wise one into a lower dimension space.\\

\begin{thm}
\emph{(Computation of Equilibria in Distributionally Robust Finite Games)}\\
  \label{alldistribequilibria}
  Consider the N-player distributionally robust game in which $i \in \{1,2,....N\}$ has action set $\{1,2,....a_i\}, \, 1 < a_i < \infty, $ in which the ambiguity set is:
  \begin{equation}
  \label{knkn}
\mathcal{F} = \{ Q : Q[\bm{\tilde{P}} \in \mathcal{U} ] = 1 , \;\; \mathbb{E}_Q [vec \bm{\tilde{P}}] = \bm{m}, \;\; \mathbb{E}_Q [\left \lVert vec (\bm{\tilde{P}}) - \bm{m} \right \rVert_1] \leq s \}
\end{equation}
where $\mathcal{U} = \{ \bm{P} : \bm{W} \cdot vec(\bm{P})\leq \bm{h}\}$ is bounded and polyhedral set,\\
and in which there is no private information. The following two conditions are equivalent.\\

\textbf{Condition 1)} $(\bm{x^1,x^2,....x^N})$ is an equilibrium of the distributionally robust game.\\
\textbf{Condition 2)} For all $i \in \{1,2,....N\}$ there exists $ \; \alpha_i, \zeta_i, \rho_i \in \mathbb{R}, \; \gamma_i \in \mathbb{R_+}, \; \bm{\beta^i,\lambda^i, \kappa^i,\delta^i, \nu^i, \tau^i, f^i, \phi^i, g^i} \in \mathbb{R}^{N \prod_{i=1}^N a_1}, \;  \bm{\xi^i, \theta^i} \in \mathbb{R}^m$ such that$ (\bm{x^1,x^2,....x^N},\alpha_i, \zeta_i, \rho_i, \gamma_i, \bm{\beta^i,\lambda^i, \kappa^i,\delta^i, \nu^i, \tau^i, f^i, \phi^i, g^i, \xi^i, \theta^i})$ satisfies:\\
 
\begin{equation}
 \label{gsdf}
 \begin{array}{l@{\qquad}l}
 \quad \zeta_i +  \frac{1}{\varepsilon_i}\alpha_i + \frac{1}{\varepsilon_i}\bm{m}^\top \bm{\beta^i} + \frac{1}{\varepsilon_i}s\gamma = \rho_i, & \bm{e}^\top \bm{x^i}=1,\\
 \quad \alpha_i - \bm{m^\top \lambda^i} +\bm{m}^\top \bm{\kappa^i} + \bm{h^\top \xi^i} \geq 0, & -\bm{\lambda^i} +\bm{\kappa^i} + \bm{W}^\top \bm{\xi^i}- \bm{\beta^i}= 0,\\
 \quad \bm{\lambda^i} +\bm{\kappa^i} - \gamma_i\bm{e} \leq \bm{0}, & \bm{\delta^i} +\bm{\nu^i} - \gamma_i\bm{e} \leq \bm{0}\\
 \quad \alpha_i - \bm{m^\top \delta^i} +\bm{m}^\top \bm{\nu^i} + \bm{h^\top \theta^i}+\zeta_i \geq 0 \\
 \quad -\bm{\delta^i} +\bm{\nu^i} + \bm{W^\top \theta^i}- \bm{\beta^i}- \bm{Y^i(x^{-i})x^i}= 0\\
 \quad -\bm{e}^\top \bm{g^i}-\bm{e}^\top \bm{\phi^i}\leq \frac{1}{\varepsilon_i}s, & -\bm{\tau^i - f^i} = \frac{1}{\varepsilon_i} \bm{m}\\
 \quad  -\bm{\tau^i + \phi^i} \leq \sigma_i \bm{m}, & \bm{\tau^i + \phi^i} \leq - \sigma_i \bm{m}\\
 \quad \bm{W\tau^i} \geq - \sigma_i \bm{h}, & \bm{Wf^i} \geq - \bm{h} \\
 \quad -\bm{f^i +g^i} \leq \bm{m}, & \bm{f^i +g^i} \leq - \bm{m} \\
 \quad \rho_i\bm{e}^\top \leq \bm{f^\top Y^i(x^{-i})} \\
 \quad \bm{\lambda^i,\kappa^i, \delta^i, \nu^i, x^i} \geq \bm{0}, & \bm{\theta^i, \xi^i, \phi^i, g^i} \leq \bm{0} \\
 \quad \gamma \geq 0, & \\
  \end{array}
 \end{equation}

where $\bm{Y^i(x^{-i})} \in \mathbb{R}^{(N \prod_{i=1}^N a_i)\times a_i}$ denotes the matrix such that 
 \begin{equation}
 \label{ooo} 
 vec(\bm{P})^\top \bm{Y^i(x^{-i})x^i} = \pi_i(\bm{P;x^{-i},x^i})
  \end{equation} 
  $\sigma_i$ is a fixed number $\forall i \in \{1,2,....N\}$
  \begin{equation} 
  \sigma_i = \frac{1-\varepsilon_i}{\varepsilon_i}
   \end{equation}
   and parameter $\varepsilon_i $ denotes the risk level of player i.\\
\end{thm}

\textbf{Proof:}
By the Formulation of the Distributionally robust games Condition 1 is equivalent to
\begin{equation} 
 \label{sss}
 x^i \in \underset{u^i\in S_{a_i}}{\operatorname{argmin}}\,\sup_{Q \in \mathcal{F}} \; Q\text{-CVaR}_{\varepsilon_i} [-\pi_i(\bm{\tilde{P};x^{-i},u^i})]  \qquad  \forall i \in \{1,2,....,N\} 
\end{equation}

From Rockafellar and Uryasev \cite{rockafellar2000optimization} and \cite{rockafellar2002conditional} we know that :

\begin{equation} 
Q\text{-CVaR}_{\varepsilon_i} [-\pi_i(\bm{\tilde{P};x^{-i},u^i})] = \min_{\zeta_i \in \mathbb{R}} \; \zeta_i + \frac{1}{\varepsilon_i} \mathbb{E}_Q [ -\pi_i(\bm{\tilde{P};x^{-i},u^i}) - \zeta_i]^+ 
\end{equation}
where $[x]^+ = \max\{x,0\}$.\\

Therefore equation ~\eqref{sss} is equivalent to:
\begin{equation} 
 x^i \in \underset{u^i\in S_{a_i}}{\operatorname{argmin}}\,\sup_{Q \in \mathcal{F}} \; \min_{\zeta_i \in \mathbb{R}} \; \zeta_i + \frac{1}{\varepsilon_i} \mathbb{E}_Q [ -\pi_i(\bm{\tilde{P};x^{-i},u^i}) - \zeta_i]^+  \qquad  \forall i \in \{1,2,....,N\} 
\end{equation}

From Saddlepoint theorem ( Sion's minimax theorem \cite{sion1958general} )  we could exchange the order of supremum and infimum(minimum) resulting: 

\begin{equation}
 \label{asaa}
 x^i \in \underset{u^i\in S_{a_i}}{\operatorname{argmin}}\, \min_{\zeta_i \in \mathbb{R}} \; \zeta_i + \frac{1}{\varepsilon_i}\,\sup_{Q \in \mathcal{F}} \mathbb{E}_Q [ -\pi_i(\bm{\tilde{P};x^{-i},u^i}) - \zeta_i]^+  \qquad  \forall i \in \{1,2,....,N\} 
\end{equation}\\

From the moment problem theory, $ \, \sup_{Q \in \mathcal{F}} \mathbb{E}_Q [ -\pi_i(\bm{\tilde{P};x^{-i},u^i}) - \zeta_i]^+ $ can be cast as the following problem: \\

\begin{equation}
\begin{array}{l@{\quad}l@{\qquad}l}
\displaystyle \text{maximize} & \displaystyle \int_{\mathcal{U}} [ -\pi_i(\bm{P;x^{-i},u^i}) - \zeta_i]^+  \mathrm{d}\mu ( vec (\bm{P})) \\
\displaystyle \text{subject to} & \displaystyle \mu \in \mathcal{M}_+ \mathbb{R}^{N \prod_{i=1}^N a_1} \\
& \displaystyle \int_{\mathcal{U}} \mathrm{d}\mu (vec(\bm{P})) = 1 \\
& \displaystyle \int_{\mathcal{U}} vec(\bm{P}) \mathrm{d}\mu (vec(\bm{P})) = \bm{m} \\
& \displaystyle \int_{\mathcal{U}} \left \lVert vec(\bm{P}) - \bm{m} \right \rVert_1 \mathrm{d}\mu (vec(\bm{P})) \leq s,
\end{array}
\end{equation}

where $\mathcal{M}_+ \mathbb{R}^{N \prod_{i=1}^N a_i}$ is the set of non-negative measures supported on $\mathbb{R}^{N \prod_{i=1}^N a_i}$.\\

There is a duality theory for moment problems(see \cite{wiesemann2014distributionally} and \cite{natarajan2009constructing}) which implies that the following dual problem attains the same optimal value:

\begin{equation}
\begin{array}{l@{\quad}l@{\qquad}l}
\displaystyle \text{minimize} & \displaystyle \alpha_i + \bm{m}^\top \bm{\beta^i} + s \gamma_i \\
\displaystyle \text{subject to} & \displaystyle \alpha_i \in \mathbb{R}, \;\; \bm{\beta^i} \in \mathbb{R}^{N \prod_{i=1}^N a_i}, \;\; \gamma_i \in \mathbb{R}_+ \\
& \displaystyle \alpha_i + vec (\bm{\tilde{P}})^\top \bm{\beta^i} + \left \lVert vec (\bm{\tilde{P}}) - \bm{m} \right \rVert_1 \gamma_i \geq [ -\pi_i(\bm{\tilde{P};x^{-i},u^i}) - \zeta_i]^+ & \displaystyle \forall \bm{P} \in \mathcal{U}.
\end{array}
\end{equation}\\

By replacing the definition of $[\cdot]^+$:

\begin{equation}
\label{main}
\begin{array}{l@{\quad}l@{\qquad}l}
\displaystyle \text{minimize} & \displaystyle \alpha_i + \bm{m}^\top \bm{\beta^i} + s\gamma_i \\
\displaystyle \text{subject to} & \displaystyle \alpha_i \in \mathbb{R}, \;\; \bm{\beta^i} \in \mathbb{R}^{N \prod_{i=1}^N a_1}, \;\; \gamma_i \in \mathbb{R}_+ \\
& \displaystyle \alpha_i + vec (\bm{\tilde{P}})^\top \bm{\beta^i} + \left \lVert vec (\bm{\tilde{P}}) - \bm{m} \right \rVert_1 \gamma_i \geq  -\pi_i(\bm{\tilde{P};x^{-i},u^i}) - \zeta_i & \displaystyle \forall \bm{P} \in \mathcal{U} \\
& \displaystyle \alpha_i + vec (\bm{\tilde{P}})^\top \bm{\beta^i} + \left \lVert vec (\bm{\tilde{P}}) - \bm{m} \right \rVert_1 \gamma_i  \geq 0 & \displaystyle \forall \bm{P} \in \mathcal{U}.
\end{array}
\end{equation}

Substituting this dual formulation of $ \, \sup_{Q \in \mathcal{F}} \mathbb{E}_Q [ -\pi_i(\bm{\tilde{P};x^{-i},u^i}) - \zeta_i]^+ $  into ~\eqref{asaa}, we obtain the following which we call Main Problem. The projection of the solution set of this problem will be the set of the equilibria that we desire.\\
\begin{equation}
\label{nana}
\begin{array}{l@{\quad}l@{\qquad}l}
\displaystyle \underset{u^i,\zeta_i,\alpha_i,\beta^i,\gamma_i}{\operatorname{min}} & \displaystyle \zeta_i + \frac{1}{\varepsilon_i}(\alpha_i + \bm{m}^\top \bm{\beta^i} + s \gamma_i) \\
\displaystyle \text{subject to} & \displaystyle \alpha_i \in \mathbb{R}, \;\; \bm{\beta^i} \in \mathbb{R}^{N \prod_{i=1}^N a_1}, \;\; \gamma_i \in \mathbb{R}_+, \zeta_i \in \mathbb{R} \\
& \displaystyle u^i \in S_{a_i}\\
& \displaystyle \alpha_i + vec (\bm{\tilde{P}})^\top \bm{\beta^i} + \left \lVert vec (\bm{\tilde{P}}) - \bm{m} \right \rVert_1 \gamma_i \geq -\pi_i(\bm{\tilde{P};x^{-i},u^i}) - \zeta_i & \displaystyle \forall \bm{P} \in \mathcal{U} \\
& \displaystyle \alpha_i + vec (\bm{\tilde{P}})^\top \bm{\beta^i} + \left \lVert vec (\bm{\tilde{P}}) - \bm{m} \right \rVert_1 \gamma_i  \geq 0 & \displaystyle \forall \bm{P} \in \mathcal{U}.
\end{array}
\end{equation}\\

This is now a `classical robust optimisation problem' and we use standard duality techniques to simplify the semi-infinite constraints.\\

We know that:
$$f(p) \geq k ,\,\, \forall p \in U \Leftrightarrow \min_{ p \in U} f(p) \geq k $$

Therefore, using this the two robust constraints of the linear program ~\eqref{nana} become:\\
\begin{equation}
\label{dasi}
\min_{\bm{P} \in U} [\alpha_i + vec (\bm{P})^\top \bm{\beta^i} + \left \lVert vec (\bm{P}) - \bm{m} \right \rVert_1 \gamma_i + \pi_i(\bm{P;x^{-i},u^i})] \geq -\zeta_i 
\end{equation}
and
\begin{equation}
\label{das}
\min_{\bm{P} \in U} [\alpha_i + vec (\bm{P})^\top \bm{\beta^i} + \left \lVert vec (\bm{P}) - \bm{m} \right \rVert_1 \gamma_i]  \geq 0 
\end{equation}\\

The left hand sides of the constraints ~\eqref{dasi} and ~\eqref{das} are equivalent to the following problems ~\eqref{mlm1} and ~\eqref{mlm2} respectively. \\
\begin{equation}
\label{mlm1}
\begin{array}{l@{\quad}l@{\qquad}l}
\displaystyle \underset{vec(\bm{P})}{\operatorname{min}} & \displaystyle \alpha_i + vec (\bm{P})^\top \bm{\beta^i} + \left \lVert vec (\bm{P}) - \bm{m} \right \rVert_1 \gamma_i + \pi_i(\bm{P;x^{-i},u^i})\\
\displaystyle \text{subject to} & \displaystyle \bm{W} \cdot vec(\bm{P})\leq \bm{h}.
\end{array}
\end{equation}
and
\begin{equation}
\label{mlm2}
\begin{array}{l@{\quad}l@{\qquad}l}
\displaystyle \underset{vec(\bm{P})}{\operatorname{min}} & \displaystyle \alpha_i + vec (\bm{P})^\top \bm{\beta^i} + \left \lVert vec (\bm{P}) - \bm{m} \right \rVert_1 \gamma_i \\
\displaystyle \text{subject to} & \displaystyle \bm{W} \cdot vec(\bm{P})\leq \bm{h}.
\end{array}
\end{equation}\\

In turn these programs are equivalent to:

\begin{equation}
\label{kti1}
\begin{array}{l@{\quad}l@{\qquad}l}
\displaystyle \underset{vec(\bm{P}),\eta}{\operatorname{min}} & \displaystyle \alpha_i + vec (\bm{P})^\top \bm{\beta^i} + \gamma_i \sum\limits_{j=1}^{N\prod_{i=1}^N a_1} \eta_j +  vec(\bm{P})^\top \bm{Y^i(x^{-i})u^i}\\
\displaystyle \text{subject to} & \displaystyle \bm{W} \cdot vec(\bm{P})\leq \bm{h}\\
& \displaystyle \eta_j \geq {vec(\bm{P})}_j - \bm{m}_j & \displaystyle \forall j= 1,2,...N\prod_{i=1}^N a_i\\
& \displaystyle \eta_j \geq \bm{m}_j - {vec(\bm{P})}_j & \displaystyle \forall j= 1,2,...N\prod_{i=1}^N a_i.
\end{array}
\end{equation}

and 

\begin{equation}
\label{kti2}
\begin{array}{l@{\quad}l@{\qquad}l}
\displaystyle \underset{vec(\bm{P}),\eta}{\operatorname{min}} & \displaystyle \alpha_i + vec (\bm{P})^\top \bm{\beta^i} + \gamma_i \sum\limits_{j=1}^{N\prod_{i=1}^N a_i} \eta_j \\
\displaystyle \text{subject to} & \displaystyle \bm{W} \cdot vec(\bm{P})\leq \bm{h}\\
& \displaystyle \eta_j \geq {vec(\bm{P})}_j - \bm{m}_j & \displaystyle \forall j= 1,2,...N\prod_{i=1}^N a_i\\
& \displaystyle \eta_j \geq \bm{m}_j - {vec(\bm{P})}_j & \displaystyle \forall j= 1,2,...N\prod_{i=1}^N a_i.
\end{array}
\end{equation}

where $\eta_j= |{vec(\bm{P})}_j - \bm{m}_j|,\,\, \forall j= 1,2,...N\prod_{i=1}^N a_i$ and $\bm{Y^i(x^{-i})}$ is as defined in ~\eqref{ooo}.\\

The dual problems of ~\eqref{kti1} and ~\eqref{kti2} are respectively the following:\\

\begin{equation}
\label{dual1}
\begin{array}{l@{\quad}l@{\qquad}l}
\displaystyle \underset{\delta^i,\beta^i,\theta^i}{\operatorname{max}} & \displaystyle \alpha_i - \bm{m^\top \delta^i} +\bm{m}^\top \bm{\nu^i} + \bm{h^\top \theta^i}\\
& \displaystyle -\bm{\delta^i} +\bm{\nu^i} + \bm{W}^\top \bm{\theta^i}- \bm{\beta^i}-\bm{Y^i(x^{-i})u^i}= 0\\
& \displaystyle \bm{\delta^i} +\bm{\nu^i} - \gamma_i\bm{e} \leq \bm{0}\\
& \displaystyle \bm{\delta^i} \geq \bm{0}, \;\; \bm{\nu^i} \geq \bm{0}, \;\; \bm{\theta^i} \leq \bm{0}.
\end{array}
\end{equation}\\

\begin{equation}
\label{dual2}
\begin{array}{l@{\quad}l@{\qquad}l}
\displaystyle \underset{\lambda^i,\kappa^i,\xi^i}{\operatorname{max}} & \displaystyle \alpha_i - \bm{m^\top \lambda^i} +\bm{m}^\top \bm{\kappa^i} + \bm{h^\top \xi^i}\\
& \displaystyle -\bm{\lambda^i} +\bm{\kappa^i} + \bm{W}^\top \bm{\xi^i}- \bm{\beta^i}= 0\\
& \displaystyle \bm{\lambda^i} +\bm{\kappa^i} - \gamma_i\bm{e} \leq \bm{0}\\
& \displaystyle \bm{\lambda^i} \geq \bm{0}, \;\; \bm{\kappa^i} \geq \bm{0}, \;\; \bm{\xi^i} \leq \bm{0}.
\end{array}
\end{equation}\\

We know that:
$$ \exists p \in U: f(p) \geq K \Leftrightarrow \max_{p\in U} f(p) \geq K$$

Subsequently, we substitute the last two problems ~\eqref{dual1} and ~\eqref{dual2} in the Main Problem ~\eqref{nana}. \\

Therefore for each player $i \in\{1,2,...N\},\,\, \exists \,\,\, \alpha_i,\gamma_i, \zeta_i \in \mathbb{R}, \bm{\beta^i,\lambda^i, \kappa^i,\delta^i, \nu^i} \in \mathbb{R}^{N \prod_{i=1}^N a_i}$ and $\bm{\xi^i, \theta^i} \in \mathbb{R}^m $ such that $ (\bm{x^i,\beta^i,\lambda^i, \kappa^i,\delta^i, \nu^i\xi^i, \theta^i,}\alpha_i, \gamma_i, \zeta_i)$ is a minimizer of:\\

\begin{equation}
\begin{array}{l@{\quad}l@{\qquad}l}
\displaystyle \underset{u^i,\alpha_i,\beta^i,\gamma_i, \zeta_i, \lambda^i, \kappa^i,\xi^i,\delta^i,\nu^i,\theta^i}{\operatorname{min}} & \displaystyle \zeta_i +  \frac{1}{\varepsilon_i}\alpha_i + \frac{1}{\varepsilon_i}\bm{m}^\top \bm{\beta^i} + \frac{1}{\varepsilon_i}s\gamma_i\\
& \displaystyle \bm{e}^\top \bm{u^i}=1\\
& \displaystyle \alpha_i - \bm{m^\top \lambda^i} +\bm{m}^\top \bm{\kappa^i} + \bm{h^\top \xi^i} \geq 0\\
& \displaystyle -\bm{\lambda^i} +\bm{\kappa^i} + \bm{W}^\top \bm{\xi^i}- \bm{\beta^i}= 0\\
& \displaystyle \bm{\lambda^i} +\bm{\kappa^i} - \gamma_i\bm{e} \leq \bm{0}\\
& \displaystyle \alpha_i - \bm{m^\top \delta^i} +\bm{m}^\top \bm{\nu^i} + \bm{h^\top \theta^i} + \zeta_i \geq 0\\
& \displaystyle -\bm{\delta^i} +\bm{\nu^i} + \bm{W}^\top \bm{\theta^i}- \bm{\beta^i}-\bm{Y^i(x^{-i})u^i}= 0\\
& \displaystyle \bm{\delta^i} +\bm{\nu^i} - \gamma_i\bm{e} \leq \bm{0}\\
& \displaystyle \bm{\lambda^i} \geq \bm{0}, \;\; \bm{\kappa^i} \geq \bm{0}, \;\; \bm{\xi^i} \leq \bm{0}\\
& \displaystyle \bm{\delta^i} \geq \bm{0}, \;\; \bm{\nu^i} \geq \bm{0}, \;\; \bm{\theta^i} \leq \bm{0}\\
& \displaystyle \bm{u^i} \geq \bm{0}, \;\; \gamma_i \geq 0.
\end{array}
\end{equation}\\

whose dual is:\\

\begin{equation}
\begin{array}{l@{\quad}l@{\qquad}l}
\displaystyle \underset{\tau^i,\rho_i,f^i,\phi^i,g^i,}{\operatorname{max}} & \displaystyle \rho_i\\
& \displaystyle-\bm{e}^\top \bm{g^i}-\bm{e}^\top \bm{\phi^i}\leq \frac{1}{\varepsilon_i}s \\
& \displaystyle-\bm{\tau^i - f^i} = \frac{1}{\varepsilon_i} \bm{m} \\
& \displaystyle -\bm{\tau^i + \phi^i} \leq \sigma_i \bm{m} \\
& \displaystyle\bm{\tau^i + \phi^i} \leq - \sigma_i \bm{m} \\
& \displaystyle\bm{W\tau^i} \geq - \sigma_i \bm{h} \\
& \displaystyle -\bm{f^i +g^i} \leq \bm{m} \\
& \displaystyle \bm{f^i +g^i} \leq - \bm{m} \\
& \displaystyle \bm{Wf^i} \geq - \bm{h} \\ 
& \displaystyle \rho_i\bm{e}^\top \leq \bm{f^\top Y^i(x^{-i})} \\ 
& \displaystyle \bm{\phi^i} \leq \bm{0}, \;\;  \bm{g^i} \leq 0.
\end{array}
\end{equation}\\

 $\sigma_i$ is a fixed number $\forall i \in \{1,2,....N\}$
  \begin{equation} 
  \sigma_i = \frac{1-\varepsilon_i}{\varepsilon_i}
   \end{equation}
   and parameter $\varepsilon_i $ denotes the risk level of player i.\\
   
Condition 2 follows from strong linear programming duality. \\
The reverse direction (Condition 2 $\Longrightarrow$ Condition 1) is also holds as all steps of our proof are based on the equivalence of the two parts.\\

\section{Analysing the ambiguity set}
\label{ambset}

In all distributionally robust games that we develop in this thesis the ambiguity set have the following form:

\begin{equation}
 \label{mkm}
\mathcal{F} = \{ Q : Q[\bm{W} \cdot vec(\bm{\tilde{P}})\leq \bm{h}] = 1 , \;\; \mathbb{E}_Q [vec \bm{\tilde{P}}] = \bm{m}, \;\; \mathbb{E}_Q [\left \lVert vec (\bm{\tilde{P}}) - \bm{m} \right \rVert_1] \leq s \}
\end{equation}

This, combined with different risk levels $\varepsilon_i$ for player $ i \in \{1,2,..N\}$ allows several variations of each distributionally robust game. By changing the values of ambiguity set's uncertain parameters $\bm{W,h, m}$ and $s$ and assuming each time different risk attitudes for the players the set of distributionally robust optimization equilibria which constitute the solution of our problem can change dramatically.\\
At this point, we present the role of each uncertain parameter of the ambiguity set ~\eqref{mkm}. The properties of this set were also mentioned in the beginning of this chapter.\\
Matrix $\bm{W} \in \mathbb{R}^{(m \times N \prod_{i=1}^N a_i)}$  and vector $\bm{h} \in \mathbb{R}^m $ are the two variables which represent the uncertainty polyhedral set in which the uncertain values of the payoff matrix should belong.\\
The maximum distance of all possible $vec(\bm{\tilde{P}})$ from the average vector $\bm{m}$ is denoted by scalar s. \\
Finally, $\bm{m} \in \mathbb{R}^{N \prod_{i=1}^N a_i}$ is the vector that denotes the expected value of $vec(\bm{\tilde{P}})$ for each distribution that belongs in the ambiguity set. \\
\underline{Important assumption:} Vector $\bm{m}$ must belong to the bounded uncertainty polyhedral set of the payoff matrix $\bm{\tilde{P}}$. Otherwise the ambiguity set $\mathcal{F}$ will be empty. \\

\subsubsection{Special Cases of Distributionally Robust Games}

Under certain conditions (special cases), the set of equilibria of distributionally robust finite game with ambiguity set like ~\eqref{mkm} is equivalent to that of a related finite game with complete payoff information (Nash Game) and with the same number of players and the same action spaces. Thus, when our game falls within these special cases one can simply solve the corresponding Nash Game.\\

The special cases of such games are studied in the next three lemmas. \\

\begin{lem}
\label{lem1}
The set of equilibria of a distributionally robust game in which all players are risk neutral ($\varepsilon_i = 1$, $\forall i \in \{1,2,..N\}$) is equivalent to the set of equilibria of a Nash Game with fixed payoff matrix $\bm{\Psi}$ where $vec(\bm{\Psi}) = \bm{m}$. ($\bm{m}$ is the average vector of the ambiguity set ~\eqref{mkm}. ) \\
 \end{lem}
 
 \begin{proof}
 
As we have already mentioned at ~\eqref{elnet4} when $\varepsilon_i = 1$:
\begin{equation}
\underset{u^i\in S_{a_i}}{\operatorname{argmin}}\,\sup_{Q \in \mathcal{F}} \; Q\text{-CVaR}_{\varepsilon_i} [-\pi_i(\bm{\tilde{P};x^{-i},u^i})] =  \underset{u^i\in S_{a_i}} {\operatorname{argmax}} \,\inf_{Q \in \mathcal{F}} \; [\pi_i(\bm{\mathbb{E}_Q [\tilde{P}];x^{-i},u^i})] 
\end{equation}\\
The second constraint of the ambiguity set $\mathcal{F}$ is  $\mathbb{E}_Q [vec \bm{\tilde{P}}] = \bm{m}$. Therefore if we denote with $\bm{\Psi}$ the matrix for which $vec(\bm{\Psi}) = m$ then we get that $\bm{\Psi}=\mathbb{E}_Q [ \bm{\tilde{P}}]$ and that:
\begin{equation}
\begin{aligned}
 \underset{u^i\in S_{a_i}} {\operatorname{argmax}} \,\inf_{Q \in \mathcal{F}} \; [\pi_i(\bm{\mathbb{E}_Q [\tilde{P}];x^{-i},u^i})]  
& = \underset{u^i\in S_{a_i}} {\operatorname{argmax}} \,\inf_{Q \in \mathcal{F}} \; [\pi_i(\bm{\Psi;x^{-i},u^i})] \\
& = \underset{u^i\in S_{a_i}} {\operatorname{argmax}}  [\pi_i(\bm{\Psi;x^{-i},u^i})] 
\end{aligned}
\end{equation}
which is equivalent to the formulation of best response in the Nash Games (see equation ~\eqref{br}).\\

 \end{proof}
\begin{lem}
\label{lem2}
The set of equilibria of a distributionally robust game in which the parameter $s$ of the ambiguity set ~\eqref{mkm} is equal to zero(s=0) is equivalent to the set of equilibria of a Nash Game with fixed payoff matrix $\bm{M}$ where $vec(\bm{M})=\bm{m}$. ($\bm{m}$ is the average vector of the ambiguity set ~\eqref{mkm}. )
\end{lem}

\begin{proof}
The third constraint of the ambiguity set is:
$$\mathbb{E}_Q [\left \lVert vec (\bm{\tilde{P}}) - \bm{m} \right \rVert_1] \leq s $$

For $s=0$($s\longrightarrow0$) become
\begin{equation}
\mathbb{E}_Q [\left \lVert vec (\bm{\tilde{P}}) - \bm{m} \right \rVert_1] \leq 0 
\end{equation}
Then because all the values inside the expectation operator $\mathbb{E}_Q $ are positive we have that 
\begin{equation}
\mathbb{E}_Q [\left \lVert vec (\bm{\tilde{P}}) - \bm{m} \right \rVert_1] = 0 
\end{equation}
and that 
\begin{equation}
Q [\left \lVert vec (\bm{\tilde{P}}) - \bm{m} \right \rVert_1 = 0] =1 
\end{equation}
which is equivalent to 
\begin{equation}
Q [ vec (\bm{\tilde{P}_i}) - \bm{m_i} = 0] =1 ,\,\,\, \forall i \in \{1,2,... \mathbb{R}^{N \prod_{i=1}^N a_1} \}.\\
\end{equation}
Therefore if $s\longrightarrow0$ the third constraint of the ambiguity set is equivalent to:
\begin{equation}
 Q [ vec (\bm{\tilde{P}}) - \bm{m} = 0] =1
\end{equation}
which means that $vec (\bm{\tilde{P}}) = \bm{m}$ for all distributions of the ambiguity set.\\

so: 
\begin{equation}
 \begin{aligned}
\mathcal{F} 
& = \{ Q : Q[\bm{W} \cdot vec(\bm{\tilde{P}})\leq \bm{h}] = 1 , \;\; \mathbb{E}_Q [vec \bm{\tilde{P}}] = \bm{m}, \;\;  Q[vec (\bm{\tilde{P}}) = \bm{m}] = 1 \}\\
& =  \{ Q : Q[vec (\bm{\tilde{P}}) = \bm{m}] = 1 \}
\end{aligned}
\end{equation}

Using the definition of best response we can find now the equivalence between our game and a Nash Game.\\ 

\begin{equation}
\begin{aligned}
\underset{u^i\in S_{a_i}}{\operatorname{argmin}}\,\sup_{Q \in \mathcal{F}} \; Q\text{-CVaR}_{\varepsilon_i} [-\pi_i(\bm{\tilde{P};x^{-i},u^i})]
& = \underset{u^i\in S_{a_i}}{\operatorname{argmin}}\,Q\text{-CVaR}_{\varepsilon_i} [-\pi_i(\bm{M;x^{-i},u^i})] \\
& = \underset{u^i\in S_{a_i}}{\operatorname{argmin}}\, [-\pi_i(\bm{M;x^{-i},u^i})] \\
& = \underset{u^i\in S_{a_i}} {\operatorname{argmax}}  [\pi_i(\bm{M;x^{-i},u^i})] 
\end{aligned}
\end{equation}
where $vec(\bm{M})=\bm{m}$.\\

\end{proof}

\begin{lem}
\label{lem3}
The set of equilibria of a distributionally robust game that had a support single point is equivalent to the set of equilibria of a Nash Game with fixed payoff matrix the one that corresponds to this single point. Single point is named the unique payoff matrix which created from specific values of the matrix $\bm{W}$ and vector $\bm{h}$ of the ambiguity set. The values of matrix $\bm{W}$ and vector $\bm{h}$ are selected in order to make the uncertainty set $U= \{\bm{P}: \bm{W} \cdot vec(\bm{P})\leq \bm{h}\}$ singleton.\\
\end{lem}

\begin{proof}
The uncertainty set $U$ is a singleton. Therefore the first constraint of the ambiguity set 
\begin{equation}
Q[\bm{W} \cdot vec(\bm{\tilde{P}})\leq \bm{h}] = 1 
\end{equation}

is equivalent to:
\begin{equation}
Q[\bm{\tilde{P}} = \bm{C} ] = 1 
\end{equation}
where $\bm{C}$ denotes the support single point, the only matrix of set $U$.\\

Thus the ambiguity set becomes:
\begin{equation}
\mathcal{F} = \{ Q : Q[\bm{\tilde{P}} = \bm{C} ] = 1  , \;\; \mathbb{E}_Q [vec \bm{\tilde{P}}] = \bm{m}, \;\;  Q[vec (\bm{\tilde{P}}) = \bm{m}] = 1 \}\\
\end{equation}
where $vec(\bm{C})= m$.\\

Using this, the desired result follows:\\

\begin{equation}
\begin{aligned}
\underset{u^i\in S_{a_i}}{\operatorname{argmin}}\,\sup_{Q \in \mathcal{F}} \; Q\text{-CVaR}_{\varepsilon_i} [-\pi_i(\bm{\tilde{P};x^{-i},u^i})]
& = \underset{u^i\in S_{a_i}}{\operatorname{argmin}}\,Q\text{-CVaR}_{\varepsilon_i} [-\pi_i(\bm{C;x^{-i},u^i})] \\
& = \underset{u^i\in S_{a_i}}{\operatorname{argmin}}\, [-\pi_i(\bm{C;x^{-i},u^i})] \\
& = \underset{u^i\in S_{a_i}} {\operatorname{argmax}}  [\pi_i(\bm{C;x^{-i},u^i})] 
\end{aligned}
\end{equation}

\end{proof}

The main reason of using the ambiguity set ~\eqref{adas} in the developing of our theorem is because is restrictive enough to imply tractable optimization problems and at the same time is expressive enough to cover a large variety of different ambiguity sets. By choosing carefully the values of $\bm{W,h,m}$ and $s$ of the ambiguity set, together with the possibility of several risk levels $\varepsilon_i$ we are able to solve many variations of one specific distributionally robust game.\\

The practical applicability of these three special cases of distributionally robust games is showed in the next chapter.\\ 

\section{Two concrete examples of finite Ditsributionally Robust Games}
\label{examplesdistribgames}
Having presented our distributionally robust games model we will now illustrate our approach with two concrete examples.

\subsection{Distributionally Robust Free Rider Game}

In subsection ~\eqref{finiterobustgames}, we have explained the robust approach of the classical two player Free Rider Game that Bertsimas and Aghassi first proposed in 2006. In particular, in \cite{aghassi2006robust} authors showed that the robust free rider game with $\tilde{c} \in (1/4,5/8)$ has exactly three equilibria $(1,0),(0,1),(3/8,3/8)$.\\
To develop similar results with them, but for the Distributionally Robust approach that we introduced we change the formulation of the problem and we create the \emph{Distributionally Robust Free rider Game}.\\

\emph{Problem Description:}\\

The \textbf{Distributionally Robust Free rider }is a two player game in which each player has two possible actions, to contribute or not contribute in the common good. If a player decide to contribute then he bears a cost of amount $\tilde{c}$ and both players enjoy a payoff of 1. If none of the players contributes then no one loses or gains money. The players choose strategies simultaneously.\\

The representation of this game is similar with Table ~\ref{rfrider} of section ~\eqref{k323}.\\

\begin{table}[H]
\centering
\caption{Distributionally Robust Free Rider Game}
\begin{tikzpicture}[element/.style={minimum width=1.75cm,minimum height=0.85cm}]
\matrix (m) [matrix of nodes,nodes={element},column sep=-\pgflinewidth, row sep=-\pgflinewidth,]{
         & Contrib  & NoCon  \\
Contrib & |[draw]|($1-\tilde{c},1-\tilde{c}$) & |[draw]|($1-\tilde{c},1$) \\
NoCon  & |[draw]|($1,\,\,\,\,\,\,\,\,\,1-\tilde{c}$) & |[draw]|(0,0) \\
};

\node[above=0.25cm] at ($(m-1-2)!0.5!(m-1-3)$){\textbf{Player 2}};
\node[rotate=90] at ($(m-2-1)!0.5!(m-3-1)+(-1.25,0)$){\textbf{Player 1}};
\end{tikzpicture}
\end{table}

In the classical Free rider game the parameter $\tilde{c}$ is fixed ($\tilde{c}=\check{c}$) and in the robust approach $\tilde{c} \in [\check{c} - \Delta,\check{c} + \Delta]$ where $\Delta$ is fixed strictly positive number. \\
In our new, distributionally robust approach, players have partial information about the probability distribution of the uncertain variable $\tilde{c}$ (about the probability distribution Q of the payoff matrix $\bm{\tilde{P}}$). In particular, the players do not know the exact distribution of the payoff matrix. They only be aware of a common known ambiguity set $\mathcal{F}$ of all possible probability distributions $Q$ that satisfy some specific properties. Subsequently, all players adopt a worst case CVaR approach to the uncertainty which is computed over all probability distributions within the set  $\mathcal{F}$. 
The introduction of the CVaR in the formulation of the game allows the two players to have different risk attitudes. Finally, the risk levels of the players are assumed to be common knowledge and none of the two players has private information.\\

For example, we may consider the distributionally robust free rider game in which the ambiguity set is given by:

\begin{equation}
  \label{frreriderequili}
{\mathcal{F}}_1 = \{ Q : Q[\tilde{c} \in [ \check{c}-\Delta , \check{c}+\Delta ] ] = 1 , \;\; \mathbb{E}_Q [vec (\tilde{\bm{P}}_{(c)})] = \bm{m}, \;\; \mathbb{E}_Q [\left \lVert vec (\tilde{\bm{P}}_{(c)}) - \bm{m} \right \rVert_1] \leq s \}
\end{equation}\\
Where: $\Delta > 0,$  $s \geq 0$,  $\check{c}$ is the mid point of the interval of $\tilde{c}$ and \\

\begin{equation}
 \tilde{\bm{P}}_{(c)} = \begin{pmatrix}
(1-\tilde{c},1-\tilde{c}) & (1-\tilde{c},1)  \\
(1, 1-\tilde{c}) & (0,0) \end{pmatrix}
\end{equation}\\

\emph{Important assumption: } Vector $\bm{m}$ of the second constraint of the ambiguity set $\mathcal{F}$ must belong to the bounded uncertainty polyhedral set of the payoff matrix $\bm{\tilde{P}}$. Otherwise the ambiguity set $\mathcal{F}$ will be empty. For the distributionally Robust free rider game this corresponds to: $\bm{m}\in [ \check{c}-\Delta , \check{c}+\Delta ]$

\subsection{Distributionally Robust Inspection Game}
\emph{Problem Description:}\\
As we have seen in subsection ~\eqref{importansce} at Table ~\ref{robinspectiongame} the Robust Inspection Game its a two player game in which the row player is the employee (possible actions:Shirk or Work) and the column player is the employer(with possible actions Inspect or not Inspect). 

\begin{table}[H]
\caption{Distributionally Robust Inspection Game}
\centering
\begin{tabular}{ |c|c|c| }
    \hline
      & Inspect & NotInspect \\ \hline
    Shirk & ($0 ,-\tilde{h}$) & (w, -w) \\ \hline
    Work & ($w-\tilde{g},\tilde{v}-w-\tilde{h}$) & ($w-\tilde{g},\tilde{v}-w$) \\
    \hline
  \end{tabular}
\label{robinspectiongame}
  \end{table}

The two players choose their actions simultaneously and then they receive the corresponding to the combination of their strategies payoffs. When the employee works he has cost $\tilde{g}$ and his employer has profit equal to $\tilde{v}$. Each inspection costs to the employer $\tilde{h}$ but if he inspects and find the employee shirking then he does not pay him his wage w. In all other cases employee's wage is paid. All values except the payment w of the employee are uncertain. 

In distributionally robust inspection game, players have partial information about the probability distributions of the uncertain variables $\tilde{g},\tilde{v}$ and $\tilde{h}$ (about the probability distribution Q of the payoff matrix $\bm{\tilde{P}}$). In particular, the players do not know the exact distribution of the payoff matrix. They only be aware of a common known ambiguity set $\mathcal{F}$ of all possible probability distributions $Q$ that satisfy some specific properties. Subsequently, all players adopt a worst case CVaR approach to the uncertainty which is computed over all probability distributions within the set  $\mathcal{F}$. 
The introduction of the CVaR in the formulation of the game allows the two players to have different risk attitudes. Finally, the risk levels of the players are assumed to be common knowledge and none of the two players has private information.\\

In particular,the following ambiguity set $\mathcal{F}$ may denote the set of all probability distributions that are consistent with the known distributional properties of Q.\\

\begin{equation}
  \label{Inspectionerequili}
{\mathcal{F}}_2 = \{ Q : Q[(\tilde{g},\tilde{v},\tilde{h}) \in [ \underline{g} , \overline{g} ]\times [ \underline{v} , \overline{v} ]\times [ \underline{h} , \overline{h} ]] = 1 , \;\; \mathbb{E}_Q [vec (\tilde{\bm{P}}_{(g,v,h)})] = \bm{m}, \;\; \mathbb{E}_Q [\left \lVert vec (\tilde{\bm{P}}_{(g,v,h)}) - \bm{m} \right \rVert_1] \leq s \}
\end{equation}\\

Where: $s \geq 0$, and\\

\begin{equation}
 \tilde{\bm{P}}_{(g,v,h)} = \begin{pmatrix}
(0 ,-\tilde{h}) & (w, -w)  \\
(w-\tilde{g},\tilde{v}-w-\tilde{h})) & (w-\tilde{g},\tilde{v}-w) \end{pmatrix}
\end{equation}\\

Vector $\bm{m}$ of the second constraint of the ambiguity set must belong to the bounded polyhedral uncertainty set of payoff matrix $\bm\tilde{P}$. For the Distributionally Robust Inspection Game that is, $\bm{m} \in  [ \underline{g} , \overline{g} ]\times [ \underline{v} , \overline{v} ]\times [ \underline{h} , \overline{h} ]$. Otherwise the ambiguity set will be empty.

\chapter{Numerical Evaluation}
In this chapter, we experimentally evaluate the new model of games described in Chapter 4 and
investigate its practical applicability. Our main goal is to show how the number of equilibria and players' payoffs change under several assumptions about the ambiguity set and the players' risk levels. The games with witch we deal in this chapter are the previously developed concrete examples of Distributionally Robust Free Rider Game and Distributionally Robust Inspection Game.\\

The method that we use to approximately compute the distributionally robust optimization equilibria  and the players' payoffs at each equilibrium of any distributionally robust game is developed as follows:
\begin{enumerate}
\item Check if the ambiguity set of the distributionally robust game can be expressed like the general form of equation ~\eqref{adas}. ( The ambiguity sets of Distributionally Robust Free Rider Game and Distributionally Robust Inspection Game have this property.)
\item Estimate the multi-linear system of equalities and inequalities whose dimension-reducing component-wise projection of the feasible solution set is equivalent with the set of equilibria of the distributionally robust game (see theorem ~\eqref{alldistribequilibria}) 
\item Find the feasible solutions of the multi-linear system and for each solution keep the components that correspond to the strategies of the players (projection of the solution). Additionally, compute the players' payoffs at each equilibrium. These are achieved using the YALMIP modelling language \cite{lofberg2004yalmip}, in Matlab 2014b.\footnote{After the conversion of the game to the multi-linear system of equations and inequalities choosing modelling in YALMIP is the most preferable method as there is no need for extra pre-calculations that are prone to human error.}\\
\end{enumerate}  

All numerical evaluations of this chapter were conducted on a 2.27GHz, Intel Core i5 CPU 430 machine with 4GB of RAM.\\

More specifically, in this chapter using tables and figures we present the results (equilibria and players' payoffs at equilibria) after the implementation of our YALMIP code. The procedure that we must follow before run the code in order to transform the ambiguity set of each game in the general form ~\eqref{adas} is available in Appendix. \\

\section{Special Cases}
\label{specialcases}

In this section we evaluate the number of equilibria and players' payoffs for the\emph{ special cases} of the two aforementioned games. Special cases are like in section ~\eqref{ambset} the following distinct versions of the distributionally robust games:
\begin{itemize}
\item All players of the game are risk neutral. That is, their risk levels are equal to one.($\varepsilon_1=\varepsilon_2=1$) 
\item The maximum distance s of the third constraint of the ambiguity set is equal to zero($s=0$).
\item Existence of support single point in the ambiguity set.
\end{itemize}

As we mentioned before we can create several versions of any distributionally robust game by changing the uncertain variables of the ambiguity set and the players' risk levels.\\
More specifically in the experiments of this section we assume the following:  
\begin{itemize}
\label{itemiz}
\item For the distributionally robust free rider problem the uncertainty set has mid point $\check{c}= 7/16$ and $\Delta=3/16$. 
\item For the distributionally robust inspection game the uncertain parameters can take the following values: $ \tilde{g}\in [8,12], \tilde{v} \in [16,24], \tilde{h} \in [4,6]$ and $ w=15$. 
\end{itemize}

In both games we assume that the initial value of the average vector is $\bm{m}=\bm{m_1}$ where $\bm{m_1}$ denotes the vector that corresponds to the nominal version of the game. Moreover, for every special case of each distributionally robust game we make the same experiments using $\bm{m}=\bm{m_2}$ as average vector. The values $\bm{m_1}$ and $\bm{m_2}$ for each game are available in Appendix. \\
Note that average vector $\bm{m}$ of the ambiguity set $\mathcal{F}$ must belong to the bounded uncertainty polyhedral set of the payoff matrix $\bm{\tilde{P}}$. Otherwise the ambiguity set $\mathcal{F}$ will be empty.\\ 

Let us now focus on each special case of the two distributionally robust games separately.

\subsection{Special Cases of Distributionally Robust Free Rider Game}

In this subsection we present the results of the three aforementioned special cases for the distributionally robust free rider game. The equilibria of each special case and players' payoffs at each equilibrium,  are illustrated in Tables ~\ref{exp1robust}, ~\ref{exper2robust} and ~\ref{exp3robust}. \\ 
All tables have the same form. The first(second,third,...) payoff of player i corresponds to the first(second,third,...) equilibria of the equilibria column of the table.\\

\begin{table}[H]
\caption{Distributionally Robust Free Rider Game: Results when players are risk neutral. That is, their risk levels are equal to one ($\varepsilon_1=\varepsilon_2= 1$). The values of maximum distance s and average vector $\bm{m}$ of the ambiguity set are varied}
\centering
   \begin{tabular}{| l | l | l | l | l |} 
   \hline 
  ($\bm{s}$) &($\bm{m}$) & Equilibria & Payoff Player 1 & Payoff Player 2\\ 
   \hline 
    0 & $\bm{m_1}$ & (0,1), (1,0), (9/16,9/16)& 1, 0.5625, 0.5625 &  0.5625, 1, 0.5625 \\ 
   \hline 
     2& $\bm{m_1}$ & (0,1), (1,0), (9/16,9/16)& 1, 0.5625, 0.5625 &  0.5625, 1, 0.5625 \\ 
    \hline 
    5& $\bm{m_1}$ & (0,1), (1,0), (9/16,9/16) & 1, 0.5625, 0.5625 &  0.5625, 1, 0.5625 \\ 
    \hline 
     10& $\bm{m_1}$ & (0,1), (1,0), (9/16,9/16)& 1, 0.5625, 0.5625 &  0.5625, 1, 0.5625 \\ 
    \hline 
     0 & $\bm{m_2}$ & (0,1), (1,0), (1/2,1/2)&  1, 0.5, 0.5&  0.5, 1, 0.5 \\ 
   \hline 
  2& $\bm{m_2}$ & (0,1), (1,0), (1/2,1/2) & 1, 0.5, 0.5&  0.5, 1, 0.5 \\
    \hline 
    5& $\bm{m_2}$ & (0,1), (1,0), (1/2,1/2)& 1, 0.5, 0.5&  0.5, 1, 0.5 \\ 
    \hline 
     10& $\bm{m_2}$ & (0,1), (1,0), (1/2,1/2)& 1, 0.5, 0.5&  0.5, 1, 0.5 \\ 
    \hline 
   \end{tabular}
\label{exp1robust}
\end{table}

  \begin{table}[H]
\caption{Dstributionally Fobust Free Rider Game: Results when the maximum distance of the third constraint of the ambiguity set is equal to zero ($s=0$). The players are allowed to have any risk attitude(any value for $\varepsilon_i$) and the value of the average vector $\bm{m}$ must  belong in the uncertainty set.}
\centering
   \begin{tabular}{| l | l | l | l | l |} 
   \hline 
   Risk Levels &($\bm{m}$) & Equilibria & Payoff Player 1 & Payoff Player 2\\ 
   \hline 
   $\varepsilon_1=\varepsilon_2= 1/2$ &$\bm{m_1}$ & (0,1), (1,0), (9/16,9/16)& 1, 0.5625, 0.5625 &  0.5625, 1, 0.5625 \\ 
   \hline 
  $\varepsilon_1=\varepsilon_2= 1/100$  & $\bm{m_1}$ & (0,1), (1,0), (9/16,9/16)&1, 0.5625, 0.5625 &  0.5625, 1, 0.5625 \\ 
    \hline 
   $\varepsilon_1=\varepsilon_2= 10/100$ &  $\bm{m_1}$ & (0,1), (1,0), (9/16,9/16)& 1, 0.5625, 0.5625 &  0.5625, 1, 0.5625 \\ 
    \hline 
    $\varepsilon_1= 13/100, \varepsilon_2= 54/100$ &  $\bm{m_1}$ & (0,1), (1,0), (9/16,9/16)& 1, 0.5625, 0.5625 &  0.5625, 1, 0.5625 \\ 
    \hline 
    $\varepsilon_1=\varepsilon_2= 1/2$ &   $\bm{m_2}$ & (0,1), (1,0), (1/2,1/2)& 1, 0.5, 0.5&  0.5, 1, 0.5 \\
   \hline 
  $\varepsilon_1=\varepsilon_2= 1/100$  &$\bm{m_2}$ & (0,1), (1,0), (1/2,1/2)&  1, 0.5, 0.5&  0.5, 1, 0.5 \\
    \hline 
      $\varepsilon_1=\varepsilon_2= 10/100$ &  $\bm{m_2}$ & (0,1), (1,0), (1/2,1/2)& 1, 0.5, 0.5&  0.5, 1, 0.5 \\
    \hline 
    $\varepsilon_1= 13/100, \varepsilon_2= 54/100$ &  $\bm{m_2}$ & (0,1), (1,0), (1/2,1/2)& 1, 0.5, 0.5&  0.5, 1, 0.5 \\
    \hline 
   \end{tabular}
\label{exper2robust}
\end{table}

\begin{table}[H]
\caption{Distributionally Robust Free Rider Game: Results for any $\varepsilon_i$ and any value of maximum distance s, but with support single point(c=1/2)}
\centering
   \begin{tabular}{| l | l | l | l | l | l |} 
   \hline 
   Risk Levels & ($\bm{s}$) & ($\bm{m}$) & Equilibria & Payoff Player 1 & Payoff Player 2\\ 
    \hline 
    $\varepsilon_1=\varepsilon_2= 1/2$ & 3  & $\bm{m_2}$ & (0,1), (1,0), (1/2,1/2)& 1, 0.5, 0.5&  0.5, 1, 0.5  \\ 
   \hline 
  $\varepsilon_1=\varepsilon_2= 1/100$  & 47& $\bm{m_2}$ & (0,1), (1,0), (1/2,1/2)& 1, 0.5, 0.5&  0.5, 1, 0.5  \\
    \hline 
      $\varepsilon_1=\varepsilon_2= 10/100$ & 1& $\bm{m_2}$ & (0,1), (1,0), (1/2,1/2) & 1, 0.5, 0.5&  0.5, 1, 0.5 \\ 
    \hline 
    $\varepsilon_1= 13/100, \varepsilon_2= 54/100$ & 23.4& $\bm{m_2}$ & (0,1), (1,0), (1/2,1/2) & 1, 0.5, 0.5&  0.5, 1, 0.5 \\ 
    \hline 
   \end{tabular}
\label{exp3robust}
\end{table}

\emph{\underline{Discussion of the Results:}}

As we expect, when the distributionally robust finite game is included in one of the three special cases(see section ~\eqref{ambset}) then this is equivalent with a complete information game with fixed payoff matrix the one that corresponds to the average vector $\bm{m}$ of the ambiguity set. For example when $\bm{m}=\bm{m_1}$ the distributionally robust free rider game for the special cases  has 3 equilibria $(x_1^1, x_1^2):(1,0),(0,1)$ and $(9/16,9/16)$. These equilibria are exactly the same with the equilibria of the Nash Game with fixed payoff matrix: \\
\begin{equation}
\bm{\check{P}}=\begin{pmatrix}
(9/16,9/16) & (9/16,1)\\
(1,9/16) & (0,0)\end{pmatrix}, \quad \text{where} \,\ vec(\bm{\check{P}})= \bm{m_1}
\end{equation}
Note that the payoffs of the two players at equilibria are also equivalent with the payoffs of the corresponding Nash Game.

\subsection{Special Cases of Distributionally Robust Inspection Game}
\label{speccasinsgame}

In this subsection we present the results of the three aforementioned special cases for the distributionally robust inspection game. The equilibria of each special case and players' payoffs at each equilibrium,  are illustrated in Tables ~\ref{exp1inspection}, ~\ref{exper2inspection} and ~\ref{exp3inspection}. \\ 

\begin{table}[H]
\caption{ Distributionally Robust Inspection Game : Results when players are risk neutral. That is, their risk levels are equal to one ($\varepsilon_1=\varepsilon_2= 1$). The values of maximum distance s and average vector $\bm{m}$ of the ambiguity set are varied}
\centering
   \begin{tabular}{| l | l | l | l | l |} 
   \hline 
    ($\bm{s}$) & ($\bm{m}$) & Equilibria & Payoff Player 1 & Payoff Player 2\\ 
   \hline 
   0 & $\bm{m_1}$ & (1/3,2/3) & 5 & -1.666\\ 
   \hline 
    2& $\bm{m_1}$ & (1/3,2/3) & 5 & -1.666\\
    \hline 
     5& $\bm{m_1}$ &(1/3,2/3) & 5 & -1.666\\ 
    \hline 
     10& $\bm{m_1}$ & (1/3,2/3) & 5 & -1.666\\ 
    \hline 
   0 & $\bm{m_2}$ & (1/3,3/5) & 6 & -3.666\\ 
   \hline 
   2& $\bm{m_2}$ & (1/3,3/5)& 6 & -3.666 \\
    \hline 
    5& $\bm{m_2}$ & (1/3,3/5)& 6 & -3.666\\ 
    \hline 
     10& $\bm{m_2}$ & (1/3,3/5)& 6 & -3.666 \\ 
    \hline 
   \end{tabular}
\label{exp1inspection}
\end{table}

  \begin{table}[H]
\caption{Distributionally Robust Inspection Game:  Results when the maximum distance of the third constraint of the ambiguity set is equal to zero ($s=0$). The players are allowed to have any risk attitude(any value for $\varepsilon_i$) and the value of the average vector $\bm{m}$ must  belong in the uncertainty set.}
\centering
   \begin{tabular}{| l | l | l | l | l |} 
   \hline 
   Risk Levels & ($\bm{m}$) & Equilibria & Payoff Player 1 & Payoff Player 2\\ 
   \hline 
   $\varepsilon_1=\varepsilon_2= 1/2$ & $\bm{m_1}$ & (1/3,2/3)& 5 & -1.666\\  
   \hline 
  $\varepsilon_1=\varepsilon_2= 1/100$  & $\bm{m_1}$ & (1/3,2/3) & 5 & -1.666\\ 
    \hline 
   $\varepsilon_1=\varepsilon_2= 10/100$ & $\bm{m_1}$ & (1/3,2/3) & 5 & -1.666\\ 
    \hline 
    $\varepsilon_1= 13/100, \varepsilon_2= 54/100$ 0& $\bm{m_1}$ & (1/3,2/3) & 5 & -1.666\\ 
    \hline 
    $\varepsilon_1=\varepsilon_2= 1/2$   & $\bm{m_2}$ & (1/3,3/5)& 6 & -3.666 \\ 
   \hline 
  $\varepsilon_1=\varepsilon_2= 1/100$  & $\bm{m_2}$ & (1/3,3/5)& 6 & -3.666 \\
    \hline 
      $\varepsilon_1=\varepsilon_2= 10/100$ & $\bm{m_2}$ & (1/3,3/5)& 6 & -3.666 \\ 
    \hline 
    $\varepsilon_1= 13/100, \varepsilon_2= 54/100$ & $\bm{m_2}$ & (1/3,3/5)& 6 & -3.666 \\ 
    \hline 
   \end{tabular}
\label{exper2inspection}
\end{table}

 \begin{table}[H]
\caption{Distributionally Robust Inspection Game: Results for any $\bm{\varepsilon_i}$ and any value of $\bm{s}$, but with support single point (h=5,g=9,v=17,w=15)}
\centering
   \begin{tabular}{| l | l | l | l | l | l |} 
   \hline 
   Risk Levels &($\bm{s}$) & ($\bm{m}$) & Equilibria & Payoff Player 1 & Payoff Player 2\\ 
    \hline 
    $\varepsilon_1=\varepsilon_2= 1/2$ & 3  & $\bm{m_2}$ & (1/3,3/5)& 6 & -3.666 \\ 
   \hline 
  $\varepsilon_1=\varepsilon_2= 1/100$  & 47& $\bm{m_2}$ & (1/3,3/5)& 6 & -3.666\\
    \hline 
      $\varepsilon_1=\varepsilon_2= 10/100$ & 1& $\bm{m_2}$ &(1/3,3/5)& 6 & -3.666 \\ 
    \hline 
    $\varepsilon_1= 13/100, \varepsilon_2= 54/100$ & 23.4& $\bm{m_2}$ & (1/3,3/5)& 6 & -3.666 \\ 
    \hline 
   \end{tabular}
\label{exp3inspection}
\end{table}

\emph{\underline{Discussion of the Results:}}

We notice that in all tables the set of equilibria that we obtain are always equivalent with the set of equilibria of the Nash Game with fixed payoff matrix the one that corresponds to the average vector $\bm{m}$ of the ambiguity set. 
For example when $\bm{m}=\bm{m_2}$ the distributionally inspection game has only one unique equilibrium (1/3,3/5). This equilibrium is exactly the same with the Nash Game with payoff matrix $\bm{\check{P}}$ where $vec(\bm{\check{P}})=\bm{m_2}$ :\\

\begin{equation}
\bm{\check{P}}=\begin{pmatrix}
(0,-h) & (w,-w)\\
(w-g,v-w-h) & (w-g,v-w)\end{pmatrix}
=\begin{pmatrix}
(0,-5) & (15,-15)\\
(15-9,17-15-5) & (15-9,17-15)\end{pmatrix}
\end{equation}\\

The payoffs of the two players at the equilibrium are also equivalent with the payoffs of the corresponding Nash Game.

\section{Fixed ambiguity set - Several Risk Levels}

In this section we perform the following experiment for both, Distributionally Robust Inspection Game and Distributionally Robust Free Rider Game.\\

\textbf{The Experiment:}\\
\emph{What would happen to the number of equilibria and to the payments of the two players when the ambiguity set is kept fixed while the values of players' risk levels are varied.}\\

\subsection{For Distributionally Robust Inspection Game}

As we have mentioned in section ~\eqref{examplesdistribgames} the ambiguity set of the Distributionally Inspection Game is the following:

\begin{equation}
  \label{Inspectionerequili}
{\mathcal{F}}_2 = \{ Q : Q[(\tilde{g},\tilde{v},\tilde{h}) \in [ \underline{g} , \overline{g} ]\times [ \underline{v} , \overline{v} ]\times [ \underline{h} , \overline{h} ]] = 1 , \;\; \mathbb{E}_Q [vec (\tilde{\bm{P}}_{(g,v,h)})] = \bm{m}, \;\; \mathbb{E}_Q [\left \lVert vec (\tilde{\bm{P}}_{(g,v,h)}) - \bm{m} \right \rVert_1] \leq s \}
\end{equation}
Where: $s \geq 0$, and
\begin{equation}
 \tilde{\bm{P}}_{(g,v,h)} = \begin{pmatrix}
(0 ,-\tilde{h}) & (w, -w)  \\
(w-\tilde{g},\tilde{v}-w-\tilde{h})) & (w-\tilde{g},\tilde{v}-w) \end{pmatrix}
\end{equation}

In the special cases of this game, (see previous experiments) we have showed that this game has unique equilibrium. More specifically, when the average vector of the ambiguity set is the nominal($\bm{m}=\bm{m_1}$) we showed that this equilibrium is equal to (1/3,2/3)(see subsection ~\eqref{speccasinsgame}).\\
In this experiment we assume that the average vector $\bm{m}$ is the one that corresponds to the nominal version of the game and without loss of generality that the maximum distance s of the third constraint of the ambiguity set is $s=4$. In addition, we assume that $ \tilde{g}\in [8,12], \tilde{v} \in [16,24], \tilde{h} \in [4,6]$ and $ w=15$. Therefore, since all variables of the ambiguity set are kept fixed \footnote{The matrix $\bm{W}$ and vector $\bm{h}$ of the first constraint of the ambiguity set are also fixed because the uncertain parameters of the payoff matrix  ($\tilde{g},\tilde{v}, \tilde{h}$) belong in a specific fixed uncertainty set}, we can say that the ambiguity set is kept fixed in this experiment . The only variables that are allowed to change are the risk levels $\varepsilon_1$ and $\varepsilon_2$ of the two players.\\

The following tables and figures illustrate the number of equilibria and the players' payoffs at these equilibria for the aforementioned fixed ambiguity set while the players' risk levels change.

More specifically Table ~\ref{tablkiuho} shows the equilibria of the previously described game when player 1  is risk neutral ($\varepsilon_1=1$) and player 2  has several risk attitudes. The players' payoffs at equilibria for each combination of the risk levels are given in Figure ~\ref{fig:test1}.

\begin{table}[H]
\caption{Distributionally Robust Inspection Game: The equilibria for different values of Risk levels when the vector $\bm{m}$ of the ambiguity set take the nominal value and the maximum distance is $s=4$. Player 1 is risk neutral. His risk level is kept fixed ($\varepsilon_1=1$) while player 2 has several levels of risk aversion.}
\centering
   \begin{tabular}{| l | p{5cm}  | } 
   \hline 
   Risk Levels & Equilibria  \\ 
   \hline 
    $\varepsilon_1=1, \varepsilon_2= 1$ & (1/3,2/3)  \\ 
   \hline 
   $\varepsilon_1=1, \varepsilon_2= 0.75$ & (0.333, 0.66) \\
    \hline 
    $\varepsilon_1=1, \varepsilon_2= 0.5$ & (0.333, 0.66)\\ 
    \hline 
    $\varepsilon_1=1, \varepsilon_2= 0.25$ & (0.333,0.66)(0.8179,0) (0.9342,0.7069)\\ 
    \hline 
      $\varepsilon_1=1, \varepsilon_2= 0.01$  &(1,0),(0,0.66), (1,0.66),(1,0.1941) (0.333,0.66), (0.9654,0.1387),(1,0.59)\\ 
    \hline 
   \end{tabular}
\label{tablkiuho}
\end{table}

\begin{figure}[H]
\centering
\begin{subfigure}{.4\textwidth}
  \centering
  \includegraphics[width=1\linewidth]{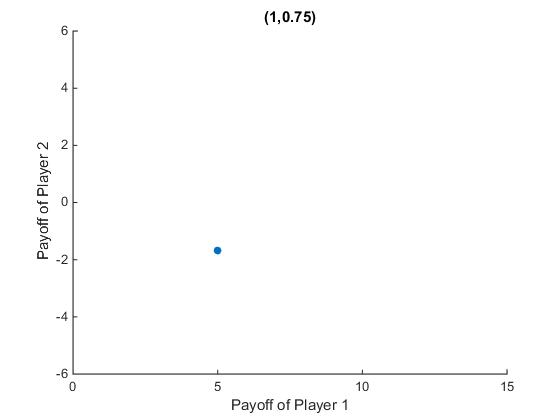}
  \caption{}
  \label{fig:sub1}
\end{subfigure}%
\begin{subfigure}{.4\textwidth}
  \centering
  \includegraphics[width=1\linewidth]{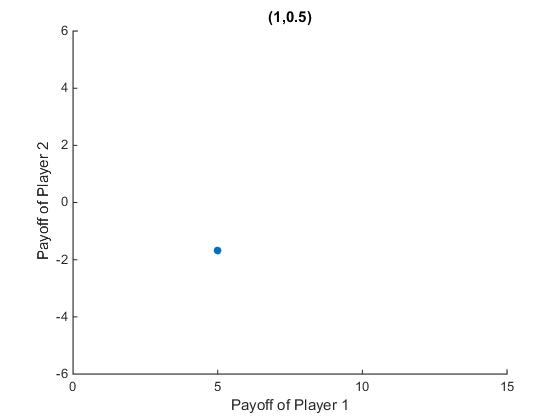}
  \caption{}
  \label{fig:sub2}
\end{subfigure}
\begin{subfigure}{.4\textwidth}
  \centering
 \includegraphics[width=1\linewidth]{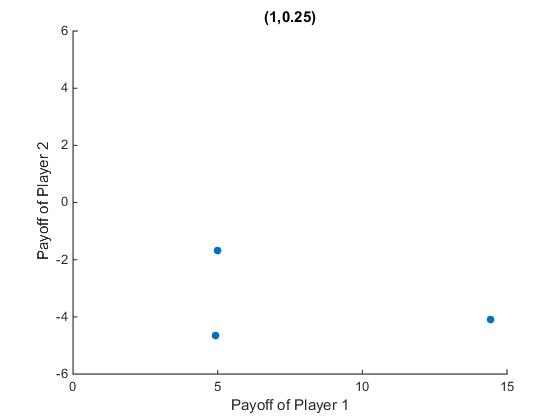}
  \caption{ }
  \label{fig:sub2}
\end{subfigure}
\begin{subfigure}{.4\textwidth}
  \centering
  \includegraphics[width=1\linewidth]{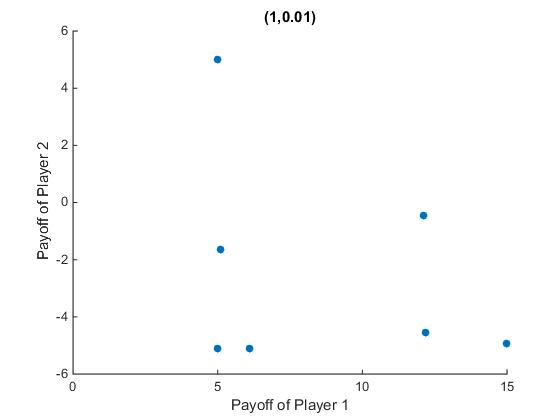}
  \caption{ }
  \label{fig:sub2}
\end{subfigure}
\caption{Graph representation of the payoffs of the two players at equilibria for different risk levels. Risk level of player 1 is kept fixed ($\varepsilon_1=1$) while player 2 has several levels of risk aversion. The title in each sub-figure denotes the risk levels of the two players: (Players 1's risk level, Players 2's risk level). }
\label{fig:test1}
\end{figure}

Subsequently,  Table ~\ref{risk2fixed} shows the equilibria of the previously described game when player 2  is risk neutral ($\varepsilon_2=1$) and player 1  has several risk attitudes. The players' payoffs at equilibria for each combination of the risk levels are given in Figure ~\ref{fig:test2}.

\begin{table}[H]
\caption{Distributionally Robust Inspection Game: The equilibria for different values of Risk levels when the vector $\bm{m}$ of the ambiguity set take the nominal value and the maximum distance is $s=4$. Player 2 is risk neutral. His risk level kept fixed ($\varepsilon_2=1$) while player 1 has several levels of risk aversion.}
\centering
   \begin{tabular}{| l | p{5cm} | } 
   \hline 
   Risk Levels & Equilibria  \\
   \hline 
    $\varepsilon_1=1, \varepsilon_2= 1$ & (1/3,2/3)  \\ 
   \hline 
   $\varepsilon_1=0.75, \varepsilon_2= 1$ & (0.333,0.666),(0.35,0.665), (0.2583,0.96)\\
    \hline 
    $\varepsilon_1=0.5, \varepsilon_2= 1$ & (0.333,0.666),(0.5379,0), (0.3842,0.66) \\ 
    \hline 
    $\varepsilon_1=0.25, \varepsilon_2= 1$ & (0.4427,0),(0.333,0.666), (0,0.3467) \\ 
    \hline 
      $\varepsilon_1=0.01, \varepsilon_2= 1$  & (0,0),(1,1), (0.333,0.666),(0.33,0), (0.335,1) \\ 
    \hline 
   \end{tabular}
\label{risk2fixed}
\end{table}

\begin{figure}[H]
\centering
\begin{subfigure}{.4\textwidth}
  \centering
 \includegraphics[width=1\linewidth]{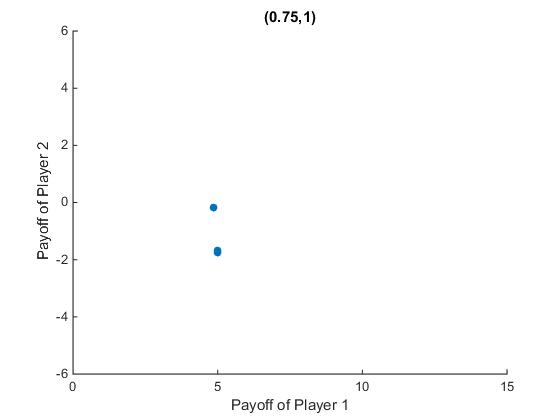}
  \caption{}
  \label{fig:sub1}
\end{subfigure}%
\begin{subfigure}{.4\textwidth}
  \centering
  \includegraphics[width=1\linewidth]{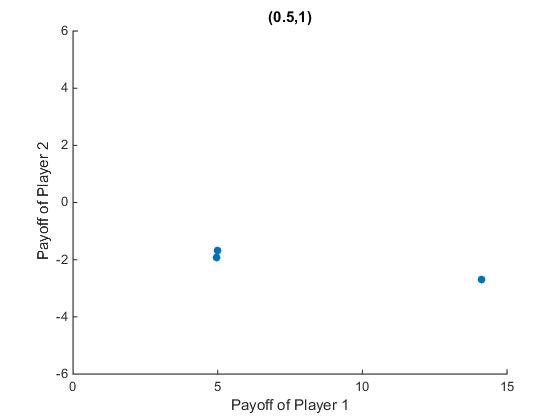}
  \caption{}
  \label{fig:sub2}
\end{subfigure}
\begin{subfigure}{.4\textwidth}
  \centering
 \includegraphics[width=1\linewidth]{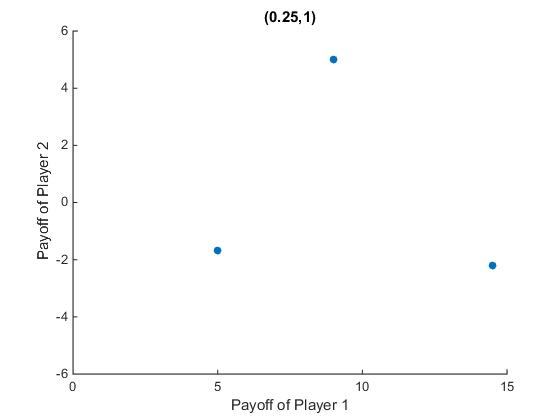}
  \caption{ }
  \label{fig:sub2}
\end{subfigure}
\begin{subfigure}{.4\textwidth}
  \centering
  \includegraphics[width=1\linewidth]{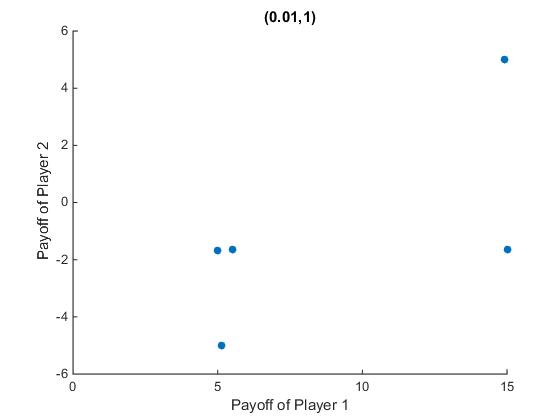}
  \caption{ }
  \label{fig:sub2}
\end{subfigure}
\caption{Graph representation of the payoffs of the two players at equilibria for different risk levels. Risk level of player 2 is kept fixed ($\varepsilon_2=1$) while player 1 has several levels of risk aversion. The title in each sub-figure denotes the risk levels of the two players: (Players 1's risk level, Players 2's risk level)}
\label{fig:test2}
\end{figure}

\textbf{Discussion of the Results:}

In standard optimization problems we know that as the decision maker becomes more risk averse his payoff always decreases. From the previous Figures ~\ref{fig:test1} and ~\ref{fig:test2} we can conclude that in game theory situation this is not always the case. We can not have a general rule,  since now the nature of the problem is more complicated. \\
In all games that we develop in this thesis we assume that players are rational, so they can predict the outcome of the game and choose the strategies that form an equilibrium. For this reason, a difference at risk attitude of a player does not change only his decision but also the decisions of his opponents.\footnote{in the distributionally robust games, risk attitude is assumed to be common knowledge}.
Hence, with increasing of risk aversion of one player the players' payoffs at equilibria may both increase or decrease depending the game. For example, in Figure ~\ref{fig:test1} at Subfigure (d) where the risk levels are (1,0.01) we can observe that for some equilibria the players have large payments and for some others very low.\\

To verify that in the Distributionally Robust Inspection Game we can not have a general rule about what happen in the payoffs of the two players when they choose to play strategies that form equilibria we also present Table ~\ref{005001} and Figure ~\ref{fig:test} witch illustrate the payoffs of the two players at equilibria when both of them are risk averse. More specifically when their risk levels are $\varepsilon_1=\varepsilon_2=0.05$ and $\varepsilon_1=\varepsilon_2=0.01$\\

\begin{table}[H]
\caption{Distributionally Robust Inspection Game: Equilibria when the players' risk levels are $\varepsilon_1=\varepsilon_2=0.05$ and $\varepsilon_1=\varepsilon_2=0.01$}
\centering
   \begin{tabular}{| l | p{5cm} | } 
   \hline 
   Risk Levels & Equilibria  \\
   \hline 
    $\varepsilon_1=\varepsilon_2= 0.05$ & (1,0.66),(1,1)(0.95,0), (0.43,1),(0.333,0.666)  \\ 
    \hline 
      $\varepsilon_1=\varepsilon_2=0.01$  & (1,0),(0,0),(0.332,0), (0.5303,1),(1,0.78) \\ 
    \hline 
   \end{tabular}
\label{005001}
\end{table}

\begin{figure}[H]
\centering
\begin{subfigure}{.5\textwidth}
  \centering
  \includegraphics[width=1\linewidth]{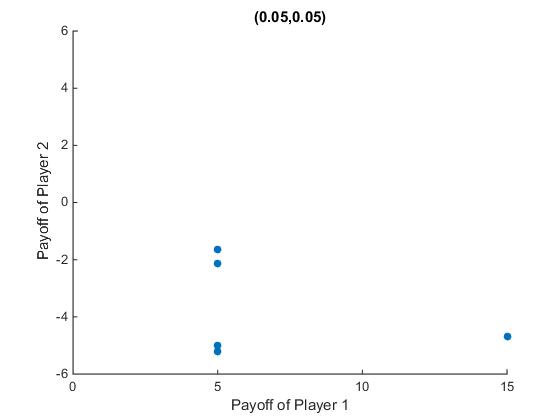}
  \caption{}
  \label{fig:sub1}
\end{subfigure}%
\begin{subfigure}{.5\textwidth}
  \centering
  \includegraphics[width=1\linewidth]{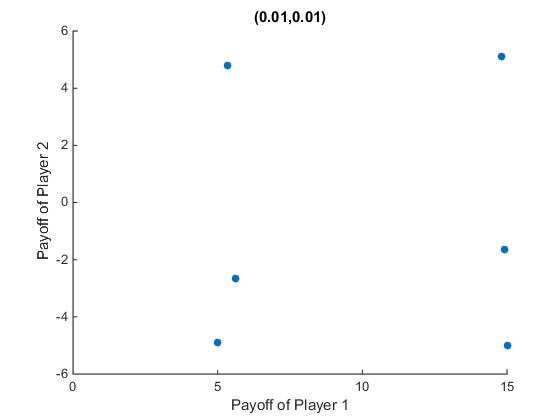}
  \caption{}
  \label{fig:sub2}
\end{subfigure}
\caption{Graph representation of the players' payoffs at equilibria when their risk levels are $\varepsilon_1=\varepsilon_2=0.05$ and $\varepsilon_1=\varepsilon_2=0.01$.The title in each sub-figure denotes the risk levels of the two players: (Players 1's risk level,Players 2's risk level)}
\label{fig:test}
\end{figure}

\subsection{For Distributionally Robust Free Rider Game}

The ambiguity set of the distributionally robust free rider game that we use in this experiment has the following form:

\begin{equation}
  \label{lokl}
{\mathcal{F}}_1 = \{ Q : Q[\tilde{c} \in [1/4,5/8] ] = 1 , \;\; \mathbb{E}_Q [vec (\tilde{\bm{P}}_{(c)})] = \bm{m_1}, \;\; \mathbb{E}_Q [\left \lVert vec (\tilde{\bm{P}}_{(c)}) - \bm{m} \right \rVert_1] \leq 2 \}
\end{equation}
Where: 
\begin{equation}
 \tilde{\bm{P}}_{(c)} = \begin{pmatrix}
(1-\tilde{c},1-\tilde{c}) & (1-\tilde{c},1)  \\
(1, 1-\tilde{c}) & (0,0) \end{pmatrix}
\end{equation}

In this experiment we assume that the uncertain parameter $\tilde{c}$ belongs to $[1/4,5/8]$, that the average vector $\bm{m}$ is the one that corresponds to the nominal version of the game and without loss of generality that the maximum distance s of the third constraint of the ambiguity set is always $s=2$. 
Therefore we can say that the ambiguity set is kept fixed. The only variables that are allowed to change are the risk levels $\varepsilon_1$ and $\varepsilon_2$ of the two players\\
In previous experiments (see Section ~\eqref{specialcases}) we showed that in special cases this distributionally robust game has three equilibria. More specifically, when the average vector of the ambiguity set is the nominal($\bm{m}=\bm{m_1}$) we showed that these equilibria are (0,1),(1,0) and (9/16,9/16). \\
Table ~\ref{sdsaf} and Figure ~\ref{gimmijl} illustrate the number of equilibria and the players' payoffs at these equilibria for the aforementioned fixed ambiguity set while the players' risk levels change.\\

\begin{table}[H]
\caption{Distributionally Robust Free Rider Game: Results for different values of Risk levels when the vector $\bm{m}$ of the ambiguity set take the nominal value and the maximum distance is $s=2$.}
\centering
   \begin{tabular}{| l | p{5cm}  | l | l |} 
   \hline 
   Risk Levels & Equilibria & Payoff Player 1 & Payoff Player 2  \\ 
   \hline 
    $\varepsilon_1=1, \varepsilon_2= 1$ & (0,1),(1,0),(9/16,9/16)& 1, 0.5625, 0.5625 &  0.5625, 1, 0.5625 \\ 
   \hline 
   $\varepsilon_1=1, \varepsilon_2= 0.75$ & (0,1),(1,0),(9/16,9/16)& 1, 0.5625, 0.5625 &  0.5625, 1, 0.5625 \\ 
    \hline 
    $\varepsilon_1=1, \varepsilon_2= 0.5$ & (0,1),(1,0),(9/16,9/16)& 1, 0.5625, 0.5625 &  0.5625, 1, 0.5625 \\  
    \hline 
    $\varepsilon_1=1, \varepsilon_2= 0.25$ & (0,1),(1,0),(9/16,9/16)& 1, 0.5625, 0.5625 &  0.5625, 1, 0.5625 \\  
    \hline 
      $\varepsilon_1=1, \varepsilon_2= 0.01$ & (0,1),(1,0),(9/16,9/16)& 1, 0.5625, 0.5625 &  0.5625, 1, 0.5625 \\  
   \hline 
   $\varepsilon_1=0.75, \varepsilon_2= 1$ & (0,1),(1,0),(9/16,9/16)& 1, 0.5625, 0.5625 &  0.5625, 1, 0.5625 \\ 
    \hline 
    $\varepsilon_1=0.5, \varepsilon_2= 1$ & (0,1),(1,0),(9/16,9/16)& 1, 0.5625, 0.5625 &  0.5625, 1, 0.5625 \\ 
    \hline 
    $\varepsilon_1=0.25, \varepsilon_2= 1$ & (0,1),(1,0),(9/16,9/16)& 1, 0.5625, 0.5625 &  0.5625, 1, 0.5625 \\  
    \hline 
      $\varepsilon_1=0.01, \varepsilon_2= 1$  & (0,1),(1,0),(9/16,9/16)& 1, 0.5625, 0.5625 &  0.5625, 1, 0.5625 \\ 
      \hline
      $\varepsilon_1=0.05, \varepsilon_2= 0.05$ & (0,1),(1,0),(9/16,9/16)& 1, 0.5625, 0.5625 &  0.5625, 1, 0.5625 \\  
    \hline 
      $\varepsilon_1=0.01, \varepsilon_2= 0.01$  & (0,1),(1,0),(9/16,9/16)& 1, 0.5625, 0.5625 &  0.5625, 1, 0.5625 \\
    \hline 
   \end{tabular}
\label{sdsaf}
\end{table}

\begin{figure}[H]
  \centering
   \includegraphics[scale = 0.7]{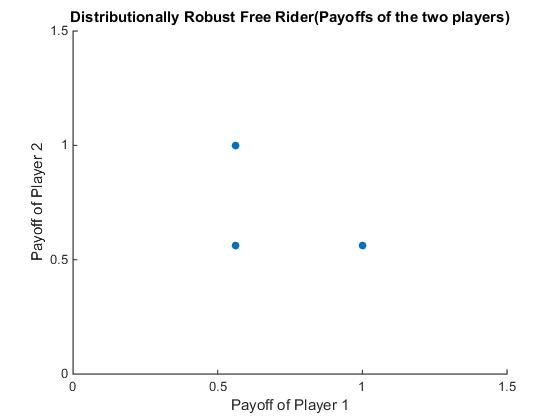}
  \caption{ Distributionally robust Free Rider Game: Graph representation of players' payoffs at equilibria for any combination of risk attitudes.}
  \label{gimmijl}
 \end{figure}

\textbf{Discussion of the Results:}\\
We notice that when the ambiguity set of the distributionally robust free rider game is like the one of equation ~\eqref{lokl} the risk levels of the two players do not influence the results of the game. For all possible combinations of the risk attitudes the equilibria of the game are always the same as well as the payoffs of the two players.\\
This results confirm our previous discussion, that in game theory situation is very difficult to develop a rule about the players' payoffs when their risk attitudes change.

\chapter{Conclusions and Future Direction of Research}
\section{Conclusions}


In this thesis we managed to combine the areas of Game Theory and Distributionally Robust Optimization and proposed a novel model of incomplete information games without private information in which the players use a distributionally robust optimization approach to cope with payoff uncertainty. 

Part of this thesis was devoted in presenting three famous models of finite games. These are the Complete Information Games (Nash Games), the Bayesian Games and the Robust Games without private information. The basic theory of robust games was analysed in detail since our distributionally robust approach presents many similarities with robust games' framework. We proposed a small extension of the special class of such games and we presented three methods which can used for approximately computing sample robust optimization equilibria.

Following that, we introduced and analysed the new model of Distributionally Robust Games. In this model players have only partial information about the probability distribution of the uncertain payoff matrix. This information is expressed through a commonly known ambiguity set of all distributions that are consistent with the known distributional properties. Similar to the robust games framework, players in distributionally robust games adopt a worst case approach. Only now the worst case is computed over all probability distributions within the ambiguity set. More specifically, players use a worst case CVaR (Conditional Value at Risk) approach. This allows players to have several risk attitudes which make our model even more coveted since in real life applications players rarely are risk neutral.

We then showed that under specific assumptions about the ambiguity set and the values of risk levels, distributionally robust game constitutes a true generalization of the three aforementioned finite games (Nash games, Bayesian Games and Robust Games). This means that any finite game of these three categories can be expressed as a distributionally robust game.
Subsequently, we proved that the set of equilibria of an arbitrary distributionally robust game with specified ambiguity set and without private information can be computed as the component-wise projection of the solution set of a multi-linear system of equations and inequalities. For special cases of such games we also showed equivalence to complete information finite games (Nash Games) with the same number of players and same action spaces.

Finally to concretize the idea of a distributionallly robust game we presented Distributionally Robust Free Rider Game and Distributionally Robust Inspection Game. Through these two concrete examples we experimentally evaluated the new model of games and we studied how the number of equilibria and the players' payments change with small changes of the unknown parameters.

\newpage
\section{Future Direction of Research}
In this section we present some possible extensions of this work.

\begin{itemize}
\item\textbf{ Existence of equilibria in Distributionally Robust Games}\\
For any class of finite games that we presented in this thesis (Nash Games, Bayesian Games, Robust Games) a theorem that show the existence of equilibria has already proved. Due to limited period of time we did not have the chance to prove the corresponding theorem for our model. Hence, a theorem that show the existence of equilibria in Distributionally Robust Finite Games seems a great extension of this work.
\item \textbf{Distributionally Robust Games with private information}\\
In this thesis we focus only on the developing of Distributionally Robust Games without private information. However, our work can be generalized to the case of distributionally robust games involving potentially private information. This will require an extension of our computation method for sampling equilibria to such situations that involve private information. A potential proof of the existence of equilibria will also need extension to this context.
\item \textbf{ More general classes of ambiguity sets}\\
Interesting results might arise if we try to make the same work for more general classes of ambiguity sets.  Possible sets can be the ambiguity sets that Wiesemann, Kuhn and Sim have considered in their pioneer paper ``Distributionally Robust Convex Optimization" \cite{wiesemann2014distributionally} where they proposed a unifying framework for modelling and solving distributionally robust optimization problems. 
\end{itemize}

\newpage
\addcontentsline{toc}{chapter}{\textbf{Bibliography}}
\bibliographystyle{plain}
\bibliography{bibliofile}


\lstset{language=Matlab,%
    breaklines=true,%
    morekeywords={matlab2tikz},
    keywordstyle=\color{blue},%
    morekeywords=[2]{1}, keywordstyle=[2]{\color{black}},
    identifierstyle=\color{black},%
    stringstyle=\color{mylilas},
    commentstyle=\color{mygreen},%
    showstringspaces=false,
    numbers=left,%
    numberstyle={\tiny \color{black}},
    numbersep=9pt, 
    emph=[1]{for,end,break},emphstyle=[1]\color{red}, 
}




\newpage
\addcontentsline{toc}{chapter}{\textbf{Appendix}}
\chapter*{Appendix}
\section*{Transformation of the ambiguity sets}

Before running our YALMIP code to obtain the desired results the ambiguity set of a Distributionally Robust Game must be in specific form. In this part of the thesis we present the work that we must do in order to create this form for the Distributionally Robust Free Rider Game and for the Distributionally Robust Inspection Game of the experiments of Chapter 5. \\

In section ~\eqref{examplesdistribgames} we presented the two concrete example of distributionally robust games that we use during this thesis. The ambiguity sets for these games are:

For the distributionally robust free rider game:
\begin{equation*}
{\mathcal{F}}_1 = \{ Q : Q[\tilde{c} \in [ \check{c}-\Delta , \check{c}+\Delta ] ] = 1 , \;\; \mathbb{E}_Q [vec (\tilde{\bm{P}}_{(c)})] = \bm{m}, \;\; \mathbb{E}_Q [\left \lVert vec (\tilde{\bm{P}}_{(c)}) - \bm{m} \right \rVert_1] \leq s \}
\end{equation*}

and for the distributionally robust indpection game:
\begin{equation*}
{\mathcal{F}}_2 = \{ Q : Q[(\tilde{g},\tilde{v},\tilde{h}) \in [ \underline{g} , \overline{g} ]\times [ \underline{v} , \overline{v} ]\times [ \underline{h} , \overline{h} ]] = 1 , \;\; \mathbb{E}_Q [vec (\tilde{\bm{P}}_{(g,v,h)})] = \bm{m}, \;\; \mathbb{E}_Q [\left \lVert vec (\tilde{\bm{P}}_{(g,v,h)}) - \bm{m} \right \rVert_1] \leq s \}
\end{equation*}

In order to be able to solve this problems (using YALMIP) the ambiguity sets must be in the following general form:

\begin{equation*}
\mathcal{F} = \{ Q : Q[\bm{\tilde{P}} \in \mathcal{U} ] = 1 , \;\; \mathbb{E}_Q [vec \bm{\tilde{P}}] = \bm{m}, \;\; \mathbb{E}_Q [\left \lVert vec (\bm{\tilde{P}}) - \bm{m} \right \rVert_1] \leq s \}
\end{equation*}
where 
 $\mathcal{U} = \{ \bm{P} : \bm{W} \cdot vec(\bm{P})\leq \bm{h}\}$ is bounded and polyhedral set.\\
 
 For this reason for each one of the distributionally robust games we must define:
 \begin{itemize}
 \item Matrix $\bm{W} \in \mathbb{R}^{(m \times N \prod_{i=1}^N a_i)}$  and vector $\bm{h} \in \mathbb{R}^m $ which are the two parameters that represent the uncertainty polyhedral set in which the uncertain values of the payoff matrix should belong.
 \item The maximum distance of all possible $vec(\bm{\tilde{P}})$ from  the mean value $\bm{m}$ which is denoted by scalar s. This could be any non-negative number.
 \item  $\bm{m} \in \mathbb{R}^{N \prod_{i=1}^N a_i}$ which is the vector that denotes the expected value of $vec(\bm{\tilde{P}})$ for each distribution that belongs in the ambiguity set. Vector $\bm{m}$ must belong to the bounded uncertainty polyhedral set of the payoff matrix $\bm{\tilde{P}}$. Otherwise the ambiguity set $\mathcal{F}$ will be empty
 \end{itemize}

 In the experiments of Chapter 5 in the case of \textbf{Distributionally Robust Free Rider Game} in order to create the ambiguity set of the general form we must use the following parameters:\\
 
\begin{equation*}
\label{kaliteri}
W = \begin{pmatrix}
1 & 0 & 0 & 0 & 0 & 0 & 0 & 0  \\
-1 & 0 & 0 & 0 & 0 & 0 & 0 & 0  \\
1 & -1 & 0 & 0 & 0 & 0 & 0 & 0  \\
-1 & 1 & 0 & 0 & 0 & 0 & 0 & 0  \\
1 & 0 & -1 & 0 & 0 & 0 & 0 & 0  \\
-1 & 0 & 1 & 0 & 0 & 0 & 0 & 0  \\
1 & 0 & 0 & 0 & 0 & -1 & 0 & 0  \\
-1 & 0 & 0 & 0 & 0 & 1 & 0 & 0  \\
0 & 0 & 0 & 1 & 0 & 0 & 0 & 0  \\
0 & 0 & 0 & -1 & 0 & 0 & 0 & 0  \\
0 & 0 & 0 & 0 & 1 & 0 & 0 & 0  \\
0 & 0 & 0 & 0 & -1 & 0 & 0 & 0  \\
0 & 0 & 0 & 0 & 0 & 0 & 1 & 0  \\
0 & 0 & 0 & 0 & 0 & 0 & -1 & 0  \\
0 & 0 & 0 & 0 & 0 & 0 & 1 & 0  \\
0 & 0 & 0 & 0 & 0 & 0 & -1 & 0 \end{pmatrix}, \quad
h = \begin{pmatrix}
6/8 \\
-3/8 \\
0 \\
0 \\
0 \\
0 \\
0 \\
0 \\
1 \\
-1 \\
1 \\
-1 \\
0 \\
0 \\
0 \\
0  \end{pmatrix} , \quad
m_1 = \begin{pmatrix}
9/16 \\
9/16 \\
9/16 \\
1 \\
1 \\
9/16 \\
0 \\
0  \end{pmatrix}, \quad
m_2 = \begin{pmatrix}
1/2 \\
1/2 \\
1/2 \\
1 \\
1 \\
1/2 \\
0 \\
0  \end{pmatrix}    
\end{equation*}
The values for matrix $\bm{W}$ and vector $\bm{h}$ are the values that correspond to  $\tilde{c} \in [1/4,5/8]$.\\
$\bm{m_1}$ and $\bm{m_2}$ are the two values of the average value $\bm{m}$ of the second constraint of the ambiguity set that we used in our experiments. The value of $\bm{m_1}$ denote the nominal version of $vec(\bm{\tilde{P}})$. That is, it takes the value of $vec(\bm{\tilde{P}})$ when the uncertain parameter $\tilde{c}$ is equal to the mid point of the interval [1/4+5/8] ($\tilde{c}=\check{c}=\dfrac{1/4+5/8}{2}= 7/16 $). $\bm{m_2}$ is the vector that is created if we assume that the uncertain parameter $\tilde{c}$ is equal to 1/2. In order to avoid emptiness of the ambiguity set, both  $\bm{m_1}$ and $\bm{m_2}$ must belong in the uncertainty set that matrix $\bm{W}$ and vector $\bm{h}$ create. \\

The special case of support single point it is the only case for witch we can use different values for $\bm{W}$ and $\bm{h}$. This happens because now we must change the first constraint of the ambiguity set in order to have only one possible payoff matrix $\bm{P}$ in the uncertainty set which it is also for certain the matrix that corresponds to the average vector $\bm{m}$ of the second constraint.\\
In our experiment we assume that single support point is the matrix when $c=1/2$. That is $Q[\tilde{c}=1/2]=1$ which in matrix form is equivalent with $Q[\tilde{\bm{P}}=\bm{C}]=1$ where C is the matrix with parameter c=1/2. In order to run our YALMIP code the first constraint of the ambiguity set must be expressed in terms of the uncertainty set $\mathcal{U} = \{ \bm{P} : \bm{W} \cdot vec(\bm{P})\leq \bm{h}\}$. For this reason the matrix $\bm{W}$ and vector $\bm{h}$ must take the following values:
\begin{equation*}
W = \begin{pmatrix}
1 & 0 & 0 & 0 & 0 & 0 & 0 & 0  \\
-1 & 0 & 0 & 0 & 0 & 0 & 0 & 0  \\
1 & -1 & 0 & 0 & 0 & 0 & 0 & 0  \\
-1 & 1 & 0 & 0 & 0 & 0 & 0 & 0  \\
1 & 0 & -1 & 0 & 0 & 0 & 0 & 0  \\
-1 & 0 & 1 & 0 & 0 & 0 & 0 & 0  \\
1 & 0 & 0 & 0 & 0 & -1 & 0 & 0  \\
-1 & 0 & 0 & 0 & 0 & 1 & 0 & 0  \\
0 & 0 & 0 & 1 & 0 & 0 & 0 & 0  \\
0 & 0 & 0 & -1 & 0 & 0 & 0 & 0  \\
0 & 0 & 0 & 0 & 1 & 0 & 0 & 0  \\
0 & 0 & 0 & 0 & -1 & 0 & 0 & 0  \\
0 & 0 & 0 & 0 & 0 & 0 & 1 & 0  \\
0 & 0 & 0 & 0 & 0 & 0 & -1 & 0  \\
0 & 0 & 0 & 0 & 0 & 0 & 1 & 0  \\
0 & 0 & 0 & 0 & 0 & 0 & -1 & 0 \end{pmatrix}, \quad
h = \begin{pmatrix}
1/2 \\
-1/2 \\
0 \\
0 \\
0 \\
0 \\
0 \\
0 \\
1 \\
-1 \\
1 \\
-1 \\
0 \\
0 \\
0 \\
0  \end{pmatrix}
\end{equation*}

Notice that matrix $\bm{W}$ remain the same but the vector $\bm{h}$ change.\\

 For the case of \textbf{Distributionally Robust Inspection Game} we must execute similar work.\\

More specifically in the experiments of Chapter 5 , the uncertain parameters of the payoff matrix are: $ \tilde{g} \in [8,12], \tilde{v} \in [16,24], \tilde{h} \in [4,6]$ and $w=15$.

Using same justification as in Distributionally Robust Free Rider Game the ambiguity set of this game can be transformed to the general form ambiguity set only if matrix $\bm{W}$ and vectors $\bm{h}$ and $\bm{m}$ of the general form take the following values:\\
\begin{equation*}
W = \begin{pmatrix}
1 & 0 & 0 & 0 & 0 & 0 & 0 & 0  \\
-1 & 0 & 0 & 0 & 0 & 0 & 0 & 0  \\
0 & 1 & 0 & 0 & 0 & 0 & 0 & 0  \\
0 & -1 & 0 & 0 & 0 & 0 & 0 & 0  \\
0 & 0 & 1 & 0 & 0 & 0 & 0 & 0  \\
0 & 0 & -1 & 0 & 0 & 0 & 0 & 0  \\
0 & 0 & 0 & 1 & 0 & 0 & 0 & 0  \\
0 & 0 & 0 & -1 & 0 & 0 & 0 & 0  \\
0 & 0 & 0 & 0 & 1 & 0 & 0 & 0  \\
0 & 0 & 0 & 0 & -1 & 0 & 0 & 0  \\
0 & 0 & 0 & 0 & 0 & 1 & 0 & 0  \\
0 & 0 & 0 & 0 & 0 & -1 & 0 & 0  \\
0 & 0 & 0 & 0 & 0 & 0 & 1 & 0  \\
0 & 0 & 0 & 0 & 0 & 0 & -1 & 0  \\
0 & 0 & 0 & 0 & 0 & 0 & 0 & 1  \\
0 & 0 & 0 & 0 & 0 & 0 & 0 & -1 \end{pmatrix}, \quad
h = \begin{pmatrix}
0 \\
0 \\
-4 \\
6\\
15 \\
-15\\
-15 \\
15 \\
7 \\
-3 \\
5 \\
5 \\
7 \\
-3 \\
9 \\
-1  \end{pmatrix} , \quad
m_1 = \begin{pmatrix}
0 \\
-5 \\
15 \\
-15 \\
5 \\
0 \\
5 \\
5  \end{pmatrix} , \quad
m_2 = \begin{pmatrix}
0 \\
-5 \\
15 \\
-15 \\
6 \\
-3 \\
6 \\
2  \end{pmatrix}  
\end{equation*}

The value of $\bm{m_1}$ denote the nominal version of the $vec(\bm{P})$. That is, it takes the value of $vec(\bm{P})$ when the uncertain parameters $\tilde{g},\tilde{v}$ and $\tilde{h}$ are equal to the mid points of their intervals ($\tilde{g}=10, \tilde{v}=20, \tilde{h}=5$ and $w=15$). Moreover, $\bm{m_2}$ is the vector that is created if we assume that the parameters of the distributionally robust inspection game are equal to $g=9, v=17 h=5$ and$ w=15$. In order to avoid emptiness of the ambiguity set, both  $\bm{m_1}$ and $\bm{m_2}$ must belong in the uncertainty set that matrix $\bm{W}$ and vector $\bm{h}$ create.\\

Again for the special case of support single point (in our experiment when h=5,g=9,v=17 and w=15), the uncertainty set must be singleton. Therefore, matrix $\bm{W}$ and vector $\bm{h}$ must change and take the following values:  \\

\begin{equation*}
W = \begin{pmatrix}
1 & 0 & 0 & 0 & 0 & 0 & 0 & 0  \\
-1 & 0 & 0 & 0 & 0 & 0 & 0 & 0  \\
0 & 1 & 0 & 0 & 0 & 0 & 0 & 0  \\
0 & -1 & 0 & 0 & 0 & 0 & 0 & 0  \\
0 & 0 & 1 & 0 & 0 & 0 & 0 & 0  \\
0 & 0 & -1 & 0 & 0 & 0 & 0 & 0  \\
0 & 0 & 0 & 1 & 0 & 0 & 0 & 0  \\
0 & 0 & 0 & -1 & 0 & 0 & 0 & 0  \\
0 & 0 & 0 & 0 & 1 & 0 & 0 & 0  \\
0 & 0 & 0 & 0 & -1 & 0 & 0 & 0  \\
0 & 0 & 0 & 0 & 0 & 1 & 0 & 0  \\
0 & 0 & 0 & 0 & 0 & -1 & 0 & 0  \\
0 & 0 & 0 & 0 & 0 & 0 & 1 & 0  \\
0 & 0 & 0 & 0 & 0 & 0 & -1 & 0  \\
0 & 0 & 0 & 0 & 0 & 0 & 0 & 1  \\
0 & 0 & 0 & 0 & 0 & 0 & 0 & -1 \end{pmatrix}, \quad
h = \begin{pmatrix}
0 \\
0 \\
-5 \\
5\\
15 \\
-15\\
-15\\
15 \\
6\\
-6 \\
-3 \\
3 \\
6 \\
-6 \\
2 \\
-2  \end{pmatrix}
\end{equation*}

\end{document}